\documentclass[11pt]{article}
\usepackage [applemac] {inputenc}
\usepackage{color}
\usepackage[final]{graphicx}
\usepackage{amssymb, amsmath,enumerate}

 \normalbaselines

\newtheorem{Theorem}{Theorem}[section]

\newtheorem{Proposition}[Theorem]{Proposition}
\newtheorem{Lemma}[Theorem]{Lemma}
\newtheorem{Corollary}[Theorem]{Corollary}

\newtheorem{Remark}[Theorem]{Remark}

\newtheorem{Assumption}[Theorem]{Assumption}

\newtheorem{Algorithm}[Theorem]{Algorithm}
\def\qed{\hfill\hbox{\hskip 6pt\vrule width6pt height7pt
depth1pt  \hskip1pt}\bigskip}

\begin{document}

\title{Asymptotic Static Hedge via Symmetrization}

\author{\begin{tabular}{ccccc}
Jir\^o Akahori\footnote{
The First author was supported by JSPS KAKENHI Grant Number $23330109$,
$24340022$, $23654056$ and $25285102$.} & Flavia Barsotti\footnote{The views presented in this paper are solely those of the author and do not necessarily represent those of UniCredit S.p.a.} &  Yuri Imamura \footnote{
The third author was supported by JSPS KAKENHI Grant Number $24840042$.}
\\
{\footnotesize Dept. of Mathematical Sciences
} & {\footnotesize Risk Methodologies, Group Financial Risks}  & {\footnotesize Dept. of Business Economics} \\
{\footnotesize Ritsumeikan University, Japan} & {\footnotesize UniCredit S.p.a., Italy}  & {\footnotesize Tokyo University of Science, Japan} \\
{\footnotesize akahori@se.ritsumei.ac.jp} &{\footnotesize Flavia.Barsotti@unicredit.eu}  &
{\footnotesize imamuray@rs.tus.ac.jp}
\end{tabular}
}

\date{\today}
\maketitle

\textwidth=160 mm \textheight=220mm \parindent=8mm \frenchspacing \vspace{ 3
mm}

\begin{abstract}
This paper is a continuation
of \cite{FJY1} where the authors
i) showed that a payment at a random time, which we call {\em timing risk}, 
is decomposed into an integral of
static positions of 
knock-in type barrier options,
ii) proposed an iteration of 
static hedge of a timing risk
by regarding the hedging error by a static hedge strategy of 
Bowie-Carr \cite{BC} type 
with respect to a barrier option
as a timing risk, 
and iii) showed that 
the error converges to zero 
by infinitely many times 
of iteration 
under a condition on 
the integrability of a relevant function.

Even though
many diffusion models including generic 1-dimensional ones satisfy the required condition, 
a construction of the iterated static hedge that is applicable
to any uniformly elliptic diffusions is postponed to 
the present paper because of 
its mathematical difficulty. 
We solve the problem in this paper by relying on the symmetrization, a technique 
first introduced in \cite{AII}
and generalized in \cite{AI}, 
and also work on {\em parametrix},
a classical technique from perturbation theory to construct a fundamental solution
of a partial differential equation. 
Due to a lack of continuity in
the diffusion coefficient, however, 
a careful study of the integrability of the relevant functions
is required. 
The long lines of proof 
itself could be a contribution 
to the parametrix analysis. 
\\[1cm]
Keywords: static hedge, barrier option, parametrix, symmetrization\\
MSC 2010: primary 91G20 secondary 91G80
\end{abstract}

%\newpage
%\tableofcontents
%\newpage

\section{Introduction}

The present paper is a continuation 
of our previous paper \cite{FJY1}, 
and focuses on solving mathematical difficulty by exploring various mathematical techniques.

Let us first recall the content of 
\cite{FJY1}. 
As the title says, the paper focuses on 
how the timing risk--- a payment at a random time (a {\em stopping time} to be 
mathematically precise)--- is evaluated. 
It answers the question by how it can be {\em statically} hedged. 
The first contribution of \cite{FJY1}
was to show that 
a timing risk is decomposed into 
an integral of continuum of static positions of knock-in options. 
This is an extension of an observation made by P. Carr and J. Picron in \cite{CP}, where just a constant payment at a stopping time was treated, and 
the decomposition was resulted by an elementary integration by part. 
The general case requires 
a more advanced 
mathematics with an argument 
on delta-approximating kernels.
In \cite{CP} where the underlying asset price 
is assumed to be a geometric Brownian motion, each knock-in option is 
hedged by a static position of a call-type option and a put-type option, 
the strategy proposed by 
J. Bowie and P. Carr in \cite{BC}
\footnote{It is often called 
semi-static hedge. Semi-static hedge represents the hedge of knock-out/knock-in options 
by simply holding positions in plain vanilla options: this topic has been widely discussed and extensively studied 
since the paper \cite{BC}. 
After this seminal contribution, the related financial literature has developed different directions of research. 
One stream of studies has focused on the extension of the reflection principle (i.e. the key tool in the Black-Scholes setting) to a 'weaker' symmetry property
(see, for example, \cite{CL}) or to a more general setting (e.g.\cite{IT}). 
To provide a concrete example, as an extreme case, the paper by \cite{CN}
obtained an exact semi-static hedging formula in a general one-dimensional diffusion environment  by constructing 
an operator which maps the pay-off function (of the option to be hedged)
to a function that admits an exact semi-static hedging formula. 
The approach has then been extended in \cite{BN} as ``weak reflection principle'' and may work for jump processes.}. 
As a consequence, 
the timing risk is hedged without error 
if an integral---infinitesimal amount 
for each maturity---of Bowie-Carr type strategies is allowed. 
The integral of static positions 
is referred to as Carr-Picron type
hedging strategy 
in \cite{FJY1}. 
In a general case, 
Bowie-Carr strategy, and therefore
Carr-Picron one brings about {\em hedging error}.
Since the error is again a timing risk, 
it is decomposed into an integral 
of static positions of knock-in options
to which Carr-Picron type strategy can be applied. 
The second contribution
of \cite{FJY1} is that
to claim that the error will be dramatically 
reduced by repeating this procedure, 
and converges to zero finally. 

The mathematics behind the above mentioned 
results of \cite{FJY1} is {\em parametrix}.
The parametrix method is a classical way to 
construct a fundamental solution to a partial differential equation as an convergent series, called heat kernel expansion (see e.g. \cite{MR0181836}).
Recently, the method has been successfully applied in finance and related fields (e.g. 
\cite{C} \cite{BK}, among others).
The parametrix, not like the Watanabe expansion in Malliavin calculus, 
does not require smoothness but 
ellipticity in the diffusion coefficients.
It heavily depends on the integrability 
of the second-order differentiation of 
the approximating kernel,
and this is obtained by the ellipticity and
(H\"older) continuity of the coefficient.
The conditions for the parametrix to work 
is postulated as assumptions in \cite{FJY1}.
Even though
many diffusion models including generic 1-dimensional ones satisfy the required condition, 
a construction of the iterated static hedge that is applicable
to any uniformly elliptic diffusions is postponed to 
the present paper. 

The contribution of the present paper is two-fold. Firstly, we propose
a systematic way for constructing an exact static hedging strategy of (single) barrier options (instead of general timing risks, 
to avoid detailed economic discussions) under a general multi-dimensional diffusion setting, 
in contrast with the existing results based on price-expansion like \cite{KTY} or \cite{STY}. 
Secondly, we give an example 
with discontinuous diffusion coefficients 
where the parametrix method can give a heat kernel expansion, 
which is convergent if the discontinuity 
is ``controllable" (see Theorem \ref{errorrep}). 

The present paper describes a methodological proposal by stating existence and convergence of asymptotic static hedging errors by leveraging on both parametrix techniques and kernel symmetrization,
a technique 
first introduced in \cite{AII}
and generalized in \cite{AI}. 
This is done for a fairly large class of multi-dimensional stochastic assets' dynamics but with the uniformly elliptic condition. 
First order, second order and higher orders hedging errors are derived and their integral representation is reported. Existence and asymptotic convergence are then proved. The proposed methodology allows to build in a systematic manner exact static hedging strategies of barrier options; an example with discontinuous diffusion coefficients is discussed, by showing how parametrix techniques can bring to an exact heat kernel expansion with Dirichlet condition.

%\item About the static hedge error: first order (eq. 28), second order (eq. 38), higher orders (eq. XXX); in each case we show existence of the error and asymptotic convergence. The result stated in Theor. 3.7 is re-iterated to obtain higher order errors (main tool). Do we want to stress that the definition - identification of the hedging error is done via the definition of $p_t(\cdot), \pi(\cdot)$?
%\item Symmetrization: introduced in the paper as technique to build the static hedge under the general multi-dimensional setting. Already introduced in the paper by Akahori, Imamura [2]. Is it better to cite Carr-Nadotchy at this point?
%\item TO CHECK: asymptotic static hedge based on symmetric kernels; in which way do we want to stress the novelty of this paper/research contribution in the static hedge field?
%\end{itemize}
%}}

The paper is organized as follows: Section \ref{GSSHE} recalls the main results achieved by \cite{FJY1}. 
%related to timing risk in terms of asymptotic static hedge of barrier options, characterization and expansion of the hedging error.
		% its error, hedge of the error, and the expansion. 
Section \ref{Parametrix} provides and discusses the main theoretical results of the proposed methodology by conducting the analysis through subsequent steps: introduction of semi-static hedges based on symmetrization under a more general mathematical setting than in \cite{FJY1} (Subsection \ref{sshviasymm}); assumptions for the underlying asset price process (Subsection \ref{assonX}); integral decomposition of the hedging error and derivation of the first order hedging error (Subsection \ref{ErrF}, Theorem \ref{RCSVF}); second order hedging error (Subsection \ref{SOSSH}, Theorem \ref{2ndexp});
%as a whole presents our main results. 
%In section \ref{sshviasymm}, we introduce
%semi-static hedges based on symmetrization,
%which is partly in the framework of 
%\cite{FJY1} but mathematically more demanding. 
%In section \ref{assonX}, we list assumptions on the price process of the underlying asset. 
%We give a remark in section \ref{univ}
%that the half-space/symmetrization setting 
%can work on ``diffeomorphic" spaces. 
%In section \ref{ErrF}, we show 
%that the hedging error can be decomposed into an integral of knock-in options
%(Theorem \ref{RCSVF}). Based on this decomposition, the second order hedge is 
%achieved (Theorem \ref{2ndexp})
%in section \ref{SOSSH}.
%The proof yields some detailed estimates,
%among which the ones in Theorem \ref{intbh} are most important, 
%are required to ensure certain integrability conditions, 
%which are far more non-trivial 
%than the corresponding ones in \cite{FJY1}.
Subsection \ref{HOH} extends the basic ideas
of the preceding sections to the identification of higher orders hedging errors. 
Section \ref{sec:conc} provides concluding remarks and Appendix \ref{App} contains the proofs of the main theoretical results presented in the paper. %, e.g. Theorems \ref{intbh} and \ref{nhint}.

	\section{A Framework of Asymptotic Static Hedges: A Quick Review of \cite{FJY1}}	\label{GSSHE}
{The aim of this Section is to recall the 
framework of %to timing risk and 
asymptotic hedging error identification and expansion and the main theoretical results achieved in \cite{FJY1}. 

We first recall the strategy of semi-static hedge of barrier options. 
Let $ X $ be a diffusion process and 
$ \tau $ be the first exit time of $ X $ 
out of a domain $ D \subset \mathbf{R}^d $. We want to hedge 
the knock-out option
by holding two plain options.
Suppose that its pay-off is 
given by $ f(X_T) 1_{ \{ \tau > T\} } $,
where $ f $ is, for the moment, 
a bounded measurable 
function on $ \mathbf{R}^d $.
The hedge strategy we will be working on is as follows:
long position of the option 
whose pay-off is
$ f (X_T) 1_{\{ X_T \in D \} } $, 
and the short position 
of the one with $ \hat{f} $,
where $ \hat{f} $ is a measurable 
function on $ \mathbf{R}^d $
such that 
$ \hat{f} = 0 $ on $ D $. Then, 
\begin{itemize}
\item If $ X $ never exit $ D $, 
then the hedge works apparently. 
\item On the event $ \{ \tau < T\} $, 
at time $ \tau $ the hedger liquidates the portfolio. 
The cost is $$ e^{-r (T-\tau)} E [ ( f(X_T) 1_{\{ X_T \in D \} }- \hat{f} (X_T) ) | \mathcal{F}_\tau ] . $$ 
\end{itemize}
If the latter was also zero, we could say that the static hedge works
perfectly but otherwise
the latter could be understood as the error of the static hedge. 

We can also consider the
static hedge of the knock-in pay-off 
$ f(X_T) 1_{ \{ \tau < T\} } $ by holding $ f (X_T) 1_{ \{X_T \in D^c\} } $ 
and $ \hat{f} $. Then, 
\begin{itemize}
\item If $ X $ never exit $ D $, 
then the hedge works apparently; nothing versus nothing. 
\item At $ \tau (< T) $, the hedger sell the option of pay-off $ \hat{f} $
and buy the one with pay-off $ f (X_T) 1_{ \{X_T \in D \}} $. 
Then the cost is again 
$ e^{-r (T-\tau)} E [ ( f (X_T) 1_{ \{X_T \in D \}}  - \hat{f} (X_T) ) | \mathcal{F}_\tau ] $. 
\item At the maturity $ T $, the pay-off is zero:
$ f (X_T) $ versus $ f (X_T) 1_{ \{ X_T \in D \}} + f(X_T) 1_{ \{ X_T \in D^c \}} $.
\end{itemize}
Thus in both cases, the {\em hedge error} evaluated at $ t\, (< \tau) $ is 
\begin{equation}\label{hdgerr}
e^{-r (T-t)} E[ E [1_{ \{\tau < T
\} } \pi (f) (X_T)  | \mathcal{F}_\tau ]| \mathcal{F}_t], 
\end{equation}
where
$ \pi(f) (x) := f (x) 1_{ \{x \in D \}}- \hat{f} (x) $.
In other words, we have
\begin{equation}\label{koerr}
    \begin{split}
&E [ f (X_T) 1_{\{ \tau > T \} }|\mathcal{F}_{t \wedge \tau} ] (\text{knock-out option to be hedged}) 
\\& =E [ \pi(f)(X_T) |\mathcal{F}_{t \wedge \tau}  ] (\text{plain options to hedge})
- E[ E [1_{ \{\tau < T
\} } \pi (f) (X_T) | \mathcal{F}_\tau ]|
\mathcal{F}_{t} ] 
(\text{hedging error})
\end{split}
\end{equation}
and
\begin{equation}\label{kierr}
\begin{split}
&E [ f (X_T) 1_{\{ \tau < T \} }|\mathcal{F}_{t \wedge \tau}  ](\text{knock-in option to be hedged})
\\
&  =E [ \pi^\bot (f)(X_T) |\mathcal{F}_{t \wedge \tau}  ] (\text{plain options to hedge})
+  E[ E [1_{ \{\tau < T
\} } \pi (f) (X_T) | \mathcal{F}_\tau ] \mathcal{F}_{t} ] (\text{hedging error}), 
\end{split}
\end{equation}
where $ \pi^{\bot}(f)(x) := f(x) 1_{ \{x \in D^c \}}+\hat{f}(x) $.
Here we assume that 
$ \hat{f} (x) = \widehat{ f(x) 1_{ \{x \in D\}}} $, 
which implies $ \widehat{\pi (f)} = \hat{f} $, $ \pi^2 (f) 
= \pi(f) $, and so on. 

The first main result of the paper \cite{FJY1} 
is to replace the hedging error \eqref{hdgerr}, 
%which is regarded as a pay-off of knock-in option, 
with the integration of knock-in options maturing at $s $
with pay-off 
\begin{equation}\label{dblconv1}
    (\mathcal{S})^1_{t} f(x) := 
\int_{\mathbf{R}^d} 
(L_x -\partial_t )p_t(x,y) \pi(f) (y) \, dy,
\end{equation}
where $ L_x $ is the infinitesimal generator 
of $ X $ acting on the variable $ x $, and 
$ p $ is a kernel approximating Dirac's delta as $ t \to 0 $
with the property that for $x \in \partial D$ 
\begin{equation*}
    \int_{\mathbf{R}^d} \pi (f) (y) 
    p_t(x, y) dy = 0,
\end{equation*}
or equivalently, 
\begin{equation}\label{relppi}
    \int_{D} p_t(x,y) f(y) \,dy = 
\int_{D^c} p_t(x,y) \hat{f} (y) \,dy. 
\end{equation}
We note that the joint integrability in $ (t, y) $ of
\begin{equation*}
    h_0 (t,x ,y) := (L_x -\partial_t )p_t(x,y)
\end{equation*}
is a naive requirement but 
if we could assume it, 
then everything works properly. 

We assume that
$ X $ has a smooth transition density $ q_t(x,y) $ and that $ q_t(x,y) $ is 
the transition density of the adjoint 
semigroup. 
The following {\em error formula} 
is established in \cite{FJY1} 
by using a fundamental relation in parametrix. 
\begin{Theorem}[\cite{FJY1}]\label{GSH}
Suppose that
\begin{equation}\label{inrv1}
    \int_0^T \int_{\mathbf{R}^d}
    q_s(x, z) |  (\mathcal{S})^1_{T-s} f(z) | \, dz ds < \infty
\end{equation}
for each $ x \in \mathbf{R}^d $. Then, the hedging error is decomposed into the integral of knock-in options:
\begin{equation*}
    E[ E [1_{ \{\tau < T
\} } \pi (f) (X_T) | \mathcal{F}_\tau ] \mathcal{F}_{t}] 
= \int_0^T E [ 1_{ \{ \tau \leq s \}} 
( (\mathcal{S})^1_{T-s} (f) ) (X_s) 
| \mathcal{F}_{\tau \wedge t}] \, ds. 
\end{equation*}
Consequently, we have the following formulas for knock-out and knock-in options, respectively;
\begin{equation}\label{KOGSH}
\begin{split}
&E [ f (X_T) 1_{\{ \tau \geq T \} }|\mathcal{F}_{t \wedge \tau} ]
\\& =E [ \pi(f)(X_T) |\mathcal{F}_{t \wedge \tau} ] 
- \int_0^T E [ 1_{ \{ \tau < s \}} 
( (\mathcal{S})^1_{T-s} (f) ) (X_s) 
| \mathcal{F}_{\tau \wedge t}] \, ds,
\end{split}
\end{equation}
and 
\begin{equation}\label{KIGSH}
\begin{split}
&E [ f (X_T) 1_{\{ \tau < T \} }|\mathcal{F}_{t \wedge \tau}  ]
\\& =E [ \pi^\bot (f)(X_T) |\mathcal{F}_{t \wedge \tau}  ] 
+ \int_0^T
E [ 1_{ \{ \tau < s \}} 
( (\mathcal{S})^1_{T-s} (f) ) (X_s) 
| \mathcal{F}_{\tau \wedge t}] \, ds.
\end{split}
\end{equation}
\end{Theorem}

Since the integrand of the second term of the right-hand-side of \eqref{KIGSH} 
is again a pay-off of knock-in option,
the formula can be iterated repeatedly
to obtain an asymptotic expansion.

We suppose that, for $ n \in \mathbf{N} $, 
\begin{equation}\label{inrv2}
    \int_{0=s_0 < s_1 < \cdots < s_n < T}
    \int_{\mathbf{R}^{dN}} q_{T-s_n}(x, y_N) \prod_{j=1}^{n} | h_0 (s_j-s_{j-1}, y_{j}, y_{j-1})| dy_j ds_j < \infty 
\end{equation}
Then, we can define  
operators $ (\mathcal{S} )_t^n $ for 
$ h = 2,\cdots, n $ recursively by 
\begin{equation*}
\begin{split}
(\mathcal{S} )^h_t f(x) = 
\int_0^t 
(\mathcal{S} )^1_s
(\mathcal{S} )^{h-1}_{t-s} f(x)\,ds.
\end{split}
\end{equation*}
The following asymptotic expansion 
formula for the semi-static hedge
is obtained in \cite{FJY1}. 
\begin{Theorem}[\cite{FJY1}]\label{HSS}
With \eqref{inrv2}, we have that for $n \in \mathbf{N}$ 
\begin{equation*}%\label{HSSform}
\begin{split}
& E [ f (X_T) 1_{\{ \tau \geq T \} }|\mathcal{F}_{t \wedge \tau}  ]
\,\, (\text{resp.} E [ f (X_T) 1_{\{ \tau < T \} }|\mathcal{F}_{t \wedge \tau}  ]) \\ 
&= E [ \pi(f)(X_T) |\mathcal{F}_{t \wedge \tau}  ] 
\,\, 
(\text{resp.} E [ \pi^{\bot} f (X_T) 1_{\{ \tau < T \} }|\mathcal{F}_t ]) 
\\
& \quad \mp \sum_{h=1}^{n-1} \int_0^T
E [ \pi^\bot ((\mathcal{S})^h_{T-s} 
(f) ) (X_s)| \mathcal{F}_{\tau \wedge t} ] \,ds \\
& \qquad \mp  \int_0^T
E [  1_{ \{ \tau < s \}} ( (\mathcal{S})^n_{T-s} 
(f) ) (X_s)| \mathcal{F}_{\tau \wedge t} ] \,ds,
\end{split}
\end{equation*}
where we understand $ \sum_{h=1}^{0} (\cdots) =0 $ conventionally. 

Furthermore, if \eqref{inrv2}
holds for any $ n \in \mathbf{N} $ 
	and the quantity goes to zero as 
$ n \to \infty $, then we have that
$ \sum_{h=1}^n \pi^\bot (\mathcal{S})^h_{T-s} (f) (x) $
converges uniformly in $ x $, 
$ \sum_{h=1}^\infty \pi^\bot (\mathcal{S})^h_{T-s} (f) (X_s) $
is integrable in $ (s,\omega) $, and 
\begin{equation*}\label{exact}
\begin{split}
& E [ f (X_T) 1_{\{ \tau \geq T \} }|\mathcal{F}_{t \wedge \tau}  ]
\,\, (\text{resp.} E [ f (X_T) 1_{\{ \tau < T \} }|\mathcal{F}_{t \wedge \tau}  ]) \\ 
&= E [ \pi(f)(X_T) |\mathcal{F}_{t \wedge \tau}  ] 
\,\, 
(\text{resp.} E [ \pi^{\bot} f (X_T) 1_{\{ \tau < T \} }|\mathcal{F}_{t \wedge \tau}  ]) 
\\
& \qquad \mp  \int_0^T
E [ \sum_{h=1}^\infty \pi^\bot ( (\mathcal{S})^h_{T-s} 
(f) ) (X_s)| \mathcal{F}_{\tau \wedge t} ] \,ds. \\
\end{split}
\end{equation*}
\end{Theorem}

%%%%%%%%%%%%%%%%%%%%%%%%%%%%%%%%%%%%%%%%%%%%%%%%%%%%%%%%%%%
\section{Static Hedge via Symmetrization }\label{Parametrix}
This Section deals with the static hedge problem by showing how to build asymptotics of static hedge error by resorting to parametrix techniques and kernel symmetrization. 
The main theoretical results are presented under a more general setting than the one considered in \cite{FJY1}. 
The intermediate steps underlying the analysis are presented and discussed in separate Subsections. 
The introduction of semi-static hedges based on symmetrization under a fairly general class of multi-dimensional models is contained in Subsection \ref{sshviasymm}. 
The assumptions considered for the underlying asset price process and their implications are threated in Subsection \ref{assonX}.
The integral decomposition of the hedging error and the derivation of the first order hedging error (Theorem \ref{RCSVF}) are contained in Subsection \ref{ErrF}. 
The results for the second order hedging error is instead given in Theorem \ref{2ndexp}, Subsection \ref{SOSSH}. 
Then, Subsection \ref{HOH} shows how to extend the basic ideas of the preceding Subsections to the identification of higher orders hedging errors.

%%%%%%%%%%%%%%%%%%%%%%%%%%%%%%%%%%%%%%%%%%%%%%%%%%%%%%%%%%%
\subsection{Semi-static hedge based on symmetrization}\label{sshviasymm}
This Subsection deals with the introduction of semi-static hedges based on symmetrization under a fairly general class of multi-dimensional models. A key element to be considered is the existence of a proper pair of the map  
	$ f1_{\{ x \in D\}} 
	\mapsto \hat{f} $ and the density $ p $
	in \eqref{relppi}.
%The above framework is somewhat abstract
%in that it does not say anything about
%the existence of proper pair of the map  
%$ f1_{\{ x \in D\}} 
%\mapsto \hat{f} $ and the density $ p $
%in \eqref{relppi}. 
Let us start from an example. In \cite{FJY1}, two specific cases 
are presented: the one dimensional case, 
and the multi-dimensional case based on {\em put-call
symmetry} introduced in \cite{AI}. 
The one dimensional case relies on the reflection 
principle of $1$-dimensional Brownian motion to pick up $ p $ and $ \pi $,
which are, respectively, the heat kernel 
of the standard Brownian motion 
and the reflection with respect to the boundary $ K $: 
\begin{equation}\label{refl1}
\pi (f) (x) = f (x)1_{\{x >K\} } - f (2K-x) 1_{\{ x \leq K \} }. 
\end{equation}
When working under the one-dimensional case, almost all diffusions
can be smoothly transformed to Brownian motion
with drift; since the knock-out region can always be characterized as an interval, the transformation would just shift it to a different interval\footnote{In \cite{FJY1} the cases where the region is a half line are studied. A similar approach can be extended to consider the cases with double boundaries.}. 

When working under the multi-dimensional setting, the same does not apply. Indeed, 
it is not always true that a generic 
diffusion process can be smoothly transformed into a Brownian motion with drift. This holds only for some special cases. Moreover, the knock-out/in region $D$ has not always the same shape, i.e. it cannot always be characterized as an interval, thus we cannot leverage on homeomorphic properties.

In this paper, we consider a multi-dimensional setting by focusing on a specific class of knock-out/in regions, i.e. those which are diffeomorphic to 
a hyper-halfspace\footnote{Observe that by the diffeomorphism, 
the diffusion matrix can take 
any form, so we do not assume any 
specific form 
in the diffusion/drift coefficients
except for the uniform ellipticity.}.
Let us introduce the following notation and setting. Let define the region $D$ as
%In section \ref{univ}, 
%we discuss how far the following setting reachs. 

%Here we give our setting. Let
\begin{equation*}
D := \{ x \in \mathbf{R}^d | \langle x, \gamma \rangle > k \},
\end{equation*}
for some $ \gamma \in \mathbf{R}^d $
with $ |\gamma|=1 $ and $ k \in \mathbf{R} $, and $ \theta $ being the reflection 
with respect to $ \partial D $ defined as
\begin{equation*}
\theta (x) = x - 2 {\langle \gamma, x \rangle \gamma}
%{|\gamma|^2} 
+ 2 %\frac
{k \gamma} %{|\gamma|^2}
= \left(I- 2 %\frac
{\gamma \otimes \gamma} %{|\gamma|^2}
\right) x 
+ 2 %\frac
{k \gamma}.  %{|\gamma|^2}. 
\end{equation*}
The methodological proposal is to choose function $ \pi $ with the same approach reported in Equation \eqref{refl1}, by considering its multi-dimensional version as:
\begin{equation}\label{sympi}
\pi (f) (x) = f(x) 1_{\{x \in D \}} - f(\theta(x))1_{ \{ x \not\in D \}}. 
\end{equation}
}
For the delta-approximating kernel,
we rely on the symmetrization 
introduced in \cite{FJY1}. 
We suppose that 
the infinitesimal generator 
of $ X $(already transformed one)
is given by
\begin{equation}\label{genrt}
\frac{1}{2} A(x) \cdot \nabla^{\otimes 2} + b(x) \cdot \nabla
\equiv \frac{1}{2} \sum_{i,j} a_{i,j} (x) \frac{\partial^2}{\partial x_i \partial x_j}
+ \sum_i b_i (x) \frac{\partial}{\partial x_i}, 
\end{equation}
where $ A $ and $ b $
are functions on $ \mathbf{R}^d $, 
$ d \times d $-
positive definite matrix valued,
and $ \mathbf{R}^d $ valued, respectively.
Let
\begin{equation}\label{DPCS}
p_t (x,y) := (2 \pi)^{-\frac{d}{2}} \{\det \tilde{A} (y)t\}^{-\frac{1}{2}} 
e^{- \frac{1}{2t} \langle \tilde{A}(y)^{-1} (x -y), x -y \rangle},
\end{equation}
where 
\begin{equation}\label{smtrz}
\tilde{A} (x) = 
\begin{cases}
A(x) & x \in D \\
\Psi A (\theta(x)) \Psi & x  \not\in D,
\end{cases}
\end{equation}
and
$ \Psi = I- 2 {\gamma \otimes \gamma} $.
Observe that this is 
the symmetrization of $ A $ 
with respect to the reflection 
$\theta $ introduced in \cite{AI}. 

We can now state the following result linking function $p_t(x,y)$ and $\pi(\cdot)$ given, respectively, in  
\eqref{DPCS} and \eqref{relppi}.

\begin{Proposition}
The function $p_t (x,y)$ defined in \eqref{DPCS} satisfies 
\eqref{relppi} with respect 
to function $\pi(\cdot)$ defined in \eqref{sympi}. 
\end{Proposition}

\begin{proof}
Since $ \Psi^2 = I $ and $ x = \theta (x) $ for $ x \in \partial D $,
\begin{equation*}
\begin{split}
& p_t (x, \theta(y) ) \\
&= 
(2 \pi)^{-\frac{d}{2}} \{\det \tilde{A}(\theta(y))t\}^{-\frac{1}{2}} 
e^{- \frac{1}{2t} \langle \tilde{A} (\theta(y))^{-1} (x -\theta(y)), x -\theta(y) \rangle} \\
&= (2 \pi)^{-\frac{d}{2}} \{\det \Psi \tilde{A} (y) \Psi t \}^{-\frac{1}{2}} 
e^{- \frac{1}{2t} \langle \tilde{A}(y)^{-1} \Psi (\theta(x) -\theta(y)), \Psi( \theta(x) -\theta(y))  \rangle} \\
& =(2 \pi)^{-\frac{d}{2}} \{\det \tilde{A}
(y) t \}^{-\frac{1}{2}} 
e^{- \frac{1}{2t} \langle \tilde{A}(y)^{-1} \Psi^2 (x -y), \Psi^2( x -y)  \rangle} \\
&= p_t (x,y).  
\end{split}
\end{equation*}
Therefore, 
\begin{equation}\label{pcsond}
\begin{split}
& \int_{\mathbf{R}^d} \pi (f) (y) p_t (x, y) \,dy = 
\int_{\mathbf{R}^d} f(y) 1_{ \{y \in D\}} 
\left( p_t (x, y) 
- p_t (x, \theta (y)) \right) \,dy =0
\end{split}
\end{equation}
for any bounded measurable $ f $ and $ x \in \partial D $.
\qed
\end{proof}

Thus, $ \pi $ of \eqref{sympi} and 
$ p $ of \eqref{DPCS} can be chosen
as a specific example of the framework of \cite{FJY1}, but it turns out that 
the integrability conditions \eqref{inrv1} or \eqref{inrv2} may fail.

The formula \eqref{pcsond}
may economically mean the following.
The kernel 
$ p $ is a kind of 
fictitious transition probability
of the underlying process.
If it were the real one,  
the price at $ \tau $ 
of the option with pay-off $ \pi (f) $
would be zero, and therefore 
the static hedge by the option with pay-off
$ f (\theta (x)) $ works without error.

%\subsection{Assumptions on $ X $}\label{assonX}
\subsection{Underlying asset price dynamics}\label{assonX}
This Subsection aims to describe the mathematical setting characterizing the assumptions on the underlying asset price dynamics. Specific assumptions on both parameters $A$ and $b$ are provided and discussed.

%We will be working on the following assumptions on $ A $ and $ b $
\begin{Assumption}\label{Hypo}
There exist positive constants $ m $ and $ M $ such that 
\begin{equation}\label{uelli}
m |y|^2 \leq %\Vert 
\langle A (x) y,y \rangle %\Vert 
\leq M |y|^2  \quad \forall x,y \in \mathbf{R}^d 
\end{equation}
where $ a_{ij}, b_j  $ have any order of derivatives, all bounded above.
\end{Assumption}
Notice that $ A $ and $ b $ are Lipschitz continuous under Assumption \ref{Hypo}.
In particular, by considering the case
\begin{equation*}
a_{\infty} 
:= \left\{ \sum_{i,j} d \max_k \left( \sup_{x \in \mathbf{R}^d} 
| \partial_k a_{i,j} (x) | \right)^2 \right\}^{\frac{1}{2}},
\end{equation*}
we have
\begin{equation}\label{Lipschitz1}
\Vert A (x) - A(y) \Vert \leq a_{\infty} | x -y |,
\end{equation}
where
$\Vert M \Vert \equiv (\mathrm{Tr}MM^* )^{\frac{1}{2}} 
$ for a matrix $M$. Moreover, Assumption \ref{Hypo} implies what follows (see e.g. \cite[Theorem 1.11, Theorem 1.15]{MR0181836}) on the transition density of $X$. 

Under Assumption \ref{Hypo}, the transition density $q_t (x,y)$ associated to $ X $
\begin{equation*}
q_t (x,y) = P (X_t \in dy |X_0=x)/dy
\end{equation*}
exists, it is twice continuously differentiable in $ (x,y) $ and 
continuously differentiable in $ t $. Moreover, there exists a constant $C_q >0$ such that for $M_0 > M $ the transition density satisfies the following inequalities
%, for any 
%$M_0 > M $ and some constant $C_q >0$, 
\begin{equation}\label{ineq-q}
q_t(x,y) \leq  C_q t^{-\frac{d}{2}}\exp\{-\frac{ |x-y|^2}{4M_0t}\},
\end{equation}
\begin{equation}\label{ineq-q1}
|\nabla q_t (x,y) | \leq  C_q 
t^{-\frac{d+1}{2}}\exp\{-\frac{ |x-y|^2}{4M_0t}\},
\end{equation}
and 
\begin{equation}\label{adj0}
\partial_s q_s (x,y) = (L_x q_s ) (x,y) = (L^*_y q_s )(x,y),
\end{equation}
where $ L_x $ is the infinitesimal generator of $ X $ (see (\ref{genrt}))
acting on the variable $ x $, and 
$ L^*_y $ is the adjoint of $ L $, acting on the variable $ y $. The adjoint $ L^*_y $ can be written under the following form:
\begin{equation*}
\begin{split}
L^*_y &= \frac{1}{2} \nabla^{\otimes 2}_y \cdot A(y) -  \nabla_y \cdot b (y) \\
& \equiv \frac{1}{2} \sum_{i,j} a_{i,j} (y) \frac{\partial^2}{\partial y_i \partial y_j} 
+ \sum_i \left( \sum_j \frac{\partial  a_{ij}}{\partial y_j}  (y)
- b_i (y) \right) \frac{\partial}{\partial y_i}
+ \frac{1}{2} 
\sum_{i,j} \frac{\partial^2 a_{ij}}{\partial y_i \partial y_j}(y)
- \sum_i \frac{\partial b_i }{\partial y_i} (y). 
\end{split}
\end{equation*}
Notice that we have
\begin{equation}\label{adj1}
\int_{\mathbf{R}^d} (L^*_y q_s )(x,y) g(y) \,dy = \int_{\mathbf{R}^d} q_s (x,y) L_y g(y) \,dy
\end{equation}
for any test function $ g \in C_0^\infty (\mathbf{R}^d) $ 
(see e.g. \cite{MR0181836}). 

Let us consider the operator $L^y_z$ defined as
\begin{equation*}
L^y_z = \frac{1}{2} \tilde{A}(y) \cdot \nabla^{\otimes 2},
%+ b(x) \cdot \nabla
\end{equation*}
acting on the variable $ z $. By considering $ (x,y) \in \mathbf{R}^d \times \mathbf{R}^d $, we then have
\begin{equation}\label{Gauss}
\partial_s p_s (x,y) = (L^y_x p_t) (x,y).
%=   (L^z}^*_z p_t) (x,z),
\end{equation}

\subsection{Hedging error formula}
\label{ErrF}

We shall establish 
the Error formula 
corresponding to Theorem \ref{GSH}.
Due to the lack of continuity 
in $ \tilde{A} $, this requires 
extra efforts. 

Recall that  %$ y \in \mathbf{R}^d $, 
\begin{equation*}
\begin{split}
h_0 (t,x,y) &=(L_x - \partial_t ) p_{t} (x,y)
\\& = (L_x - L^y_x ) p_{t} (x,y).
\end{split}\end{equation*}

\begin{Lemma}\label{fof}
For $ y \in \mathbf{R}^d $,
\begin{equation}\label{param22}
\begin{split}
q_t ( x, y) - p_t ( x, y)
%& = \lim_{\epsilon \downarrow 0} 
%\int_\epsilon^{t-\epsilon} ds \int_{\mathbf{R}^d} dz
%\{ (L^*_z q_s) (x,z) p_{t-s} (z, y) - q_s (x,z) 
%(L_z^y p_{t-s}) (z,y) \} \\
&= \int_0^t ds \int_{\mathbf{R}^d}\, dz \,\, q_s (x,z)  
h_0 (t-s, z,y). 
\end{split}
\end{equation}
\end{Lemma}

The equation (\ref{param22}) is the key to
the parametrix theory (see e.g. \cite{BK}). 
To give a proof to Lemma \ref{fof} is somewhat 
difficult since we have explicitly, 
\begin{equation}\label{hzero}
\begin{split}
& h_0 (t,z,y) %:=(L_z - L^y_z ) p_{t} (z,y) \\ 
\\& 
= \frac{1}{2}\{ A(z) - \tilde{A}(y) \} \cdot \nabla^{\otimes 2} p_t (z,y) 
+ b (z) \cdot \nabla p_t (z,y) \\
&= \frac{1}{2}\{ A(z) - \tilde{A}(y) \} \cdot 
\left( \frac{1}{t^2} \{\tilde{A}(y)\}^{-1} (z-y) 
\otimes \{\tilde{A}(y)\}^{-1}(z-y)
- \frac{1}{t} \{\tilde{A}(y)\}^{-1}\right) p_t (z,y) \\
& \qquad - b(z) \cdot \frac{1}{t} \{\tilde{A}(y)\}^{-1} (z-y) p_t (z,y) \\
&= \frac{1}{2t^2}   \{ A(z) - \tilde{A}(y) \} 
\{\tilde{A}(y)\}^{-1}(z-y) 
\cdot 
\{\tilde{A}(y)\}^{-1}(z-y) p_t (z,y) \\
& \qquad - 
\frac{1}{2t} \{\tilde{A}(y)\}^{-1} \cdot 
\Big( \{ A(z) - \tilde{A}(y) \}
+ 2b(z) \otimes (z-y) \Big) p_t (z,y). 
\end{split}
\end{equation}
We recall here that 
the integrability in $ (t,z) \in [0,T] \times \mathbf{R}^d $ of 
the terms from the second order derivative are normally   
retrieved by the continuity of $ A $ in the classical 
parametrix theory (see e.g. \cite[Chapter 1. Section 4.]{MR0181836}). 
Here, it becomes a very naive problem
since the symmetrized diffusion matrix $ \tilde{A} $
in most cases  
fails to be continuous at $ \partial D $.
To overcome this difficulty, 
we introduce a parameter that can be 
as sufficiently small as possible
when necessary. 
Set 
\begin{equation}\label{discon}
\delta 
:= 2 \sup_{ x \in \partial D} \Vert [A(x), \gamma \otimes \gamma ] \Vert.
\end{equation} 
Then the constant $ \delta $ controls the discontinuity 
in the following sense:
\begin{Lemma}
\label{disconlem}
For $ x \in D $ and $ y \in D^c $,  
\begin{equation}\label{condiscon}
\Vert A (x) - \tilde{A} (y) \Vert \leq a_\infty | x- y |
+ \delta. 
\end{equation}
\end{Lemma}
A proof of the Lemma \ref{disconlem} will be given in the Appendix \ref{lemdiscon}.

Thus, if $ \delta = 0 $, we have the Lipschitz continuity of $ \tilde{A} $ and
therefore, the integrability of $ h_0 $.
If this is the case, we can establish the 
convergent expansion by using standard 
theory (see \cite{MR0181836}, \cite{BK}, and \cite{FJY1}). Without the continuity, 
the standard approach does not work. 
However, we have the following 
estimate 
which is critical to obtain 
the result contained in Theorem \ref{intbh}.
	\begin{Lemma}\label{hestimate}
		For $x, y \in \mathbf{R}^d$, 
		\begin{equation}\label{inequalityh_0}
		|h_0 (t,x,y)| \leq C_1 t^{-\frac{1}{2}} p_t^{2M} (x,y)
		+ (\delta 1_{\{x \in D \} } + 2M d 1_{\{ x \in D^c\} })C_2 t^{-1 } p_t^{2M} (x,y) 1_{\{y \notin D\}}, 
		\end{equation}
		where 
		\begin{equation}\label{C1C2}
		\begin{split}
		C_1 &:= 2^{\frac{d}{2}} m^{- \frac{2+d}{2}} M^{\frac{1+d}{2}}
		(4m^{-1} M K_{\frac{3}{2}} a_\infty 
		+ d^{\frac{1}{2}} K_{\frac{1}{2}} a_\infty + b_\infty 
		), \\
		C_2 &:=   2^{\frac{d}{2}} m^{- \frac{2+d}{2}} M^{\frac{d}{2}}
		(2 m^{-1} M K_1 + 2^{-1} d^{\frac{1}{2}})
		\end{split}
		\end{equation}
		with $ \delta $ and 
		$ a_\infty $ as defined in
		\eqref{discon} and (\ref{Lipschitz1}), respectively, 
		\begin{equation*}
		b_\infty =\max_{1 \leq i \leq d} \Vert b_i \Vert_{\infty}, 
		\end{equation*}
		\begin{equation}\label{Kbeta}
		\sup_{x \geq 0 } |x^{\beta} e^{-x}| =:K_\beta < \infty,
		\end{equation}
		and
		\begin{equation*}\label{p2M}
		p_t^{2M} (x,y) =(4 \pi M  t)^{-d/2} 
		e^{- |x-y|^2/{4Mt}}, 
		\end{equation*}
		$ M $ being the same 
		as the one appearing in (\ref{uelli}) of Assumption \ref{Hypo}. 
	\end{Lemma}
	\begin{proof}See Appendix \ref{hestimate_Pr}.
		\qed\end{proof}
        
Let 
\begin{equation*}\label{hone}
h (t,x,y) := h_0 (t,x,y) - h_0 (t,x,\theta(y)),
\end{equation*}
The following estimates 
are also essential in obtaining our economic results,
so we state them separately as Theorem.
\begin{Theorem}\label{intbh}
Under Assumption \ref{Hypo}, we have the following 
inequalities.\\

(i) There exists a constant $ C_3 $ %depending on $ T $ 
such that, 
for $ t \in [0,T] $ and $ x  \in \mathbf{R}^d $, 
\begin{equation*}
\begin{split}
\int_{D} |h (t,x,y)| dy 
&\leq \int_{\mathbf{R}^d} |h_0 (t,x,y)| \,dy 
\\& \leq C_3\left( t^{-\frac{1}{2}} + t^{-1}  \left( e^{-\frac{(k- \langle \gamma, x \rangle )^2}{4M t}}1_{ \{x \in D \} } + 1_{ \{x \not\in D \}} \right) \right). 
\end{split}\end{equation*}\\

(ii) There exists a constant $ C_4 $ depending on $ T $ such that,
for $ s,t \in [0,T] $, $ (x,y,z,w) \in \mathbf{R}^d \times \mathbf{R}^d \times \mathbf{R}^d \times \mathbf{R}^d$ and $ M_0 > M $,
\begin{equation*}
q_s (x,y) | h (t, z,w)| \leq 
2 q_s (x,y) | h_0 (t, z,w)| \leq
C_4 s^{-\frac{d}{2}} t^{-1} 
\exp\left( -\frac{ |x-y|^2}{4M_0 s} \right) \exp\left( -\frac{ |z-w|^2}{4M t} \right).
\end{equation*}
In particular, they are integrable in $ (z,y) \in \mathbf{R}^d \times \mathbf{R}^d $.\\

(iii) Further, there exists a constant $ C_5$ depending on $ T $
such that 
\begin{equation*}\label{Hintegrability}
\left| \int_{\mathbf{R}^d} q_s (x,z)  h_0 (t, z,y) \,dz 
\right| \leq C_5 s^{-\frac{1}{2}} t^{-\frac{1}{2}}
(s+t)^{-\frac{d}{2}} \exp\left( -\frac{ |x-y|^2}{4M_0 (t+s)} \right),
\end{equation*}
for any $y \in \mathbf{R}^d $. 
In particular, 
\begin{equation*}
\int_{\mathbf{R}^d} q_s (x,z) h_0 (t-s, z,y) \,dz,
\end{equation*}
and hence
\begin{equation*}
\int_{\mathbf{R}^d} q_s (x,z) h (t-s, z,y) \,dz, 
\end{equation*}
are integrable in $ (s,y) \in [0,t]\times \mathbf{R}^d $ for any $ t \in [0,T] $.
\end{Theorem}

\begin{proof} See Appendix \ref{pintbh}.
		\qed
\end{proof}

The point here is that 
the singularity of $ t^{-1} $ in the estimate 
(i) is handled by integration by part in (iii), 
using the integrability of (ii) and the Gaussian
estimates \eqref{ineq-q} and \eqref{ineq-q1} of $ q $ and $ \nabla q $.

\begin{Remark}
We note that we do not have the integrability
of \eqref{inrv1} here, so we cannot apply
Theorem \ref{GSH}. 
\end{Remark}

The first assertion of Theorem \ref{intbh}
ensures that we can define an operator 
$ S_t $ on $ L^\infty (D) $ for each $ t > 0 $
by
\begin{equation*}
\begin{split}
{S}_{t} f (x)
&= \int_{D} h(t, x ,y) f(y) \,dy,
\end{split}
\end{equation*}
just as \eqref{dblconv1}.
\begin{Corollary}\label{bddness}
For for each $ t > 0 $, 
$ S_t $ is an operator
on $ L^\infty(D) $ 
into $ L^\infty (\mathbf{R}^d) $. 
\end{Corollary}
\begin{proof}
It directly follows from 
(i) of Theorem \ref{intbh}. \qed
\end{proof}

By leveraging on 
Lemma \ref{fof}
that is derived mathematically from 
Theorem
\ref{intbh}, we can now state the hedging error formula (integral decomposition) under the proposed multi-dimensional setting, corresponding to the one provided in Theorem
\ref{GSH} by \cite{FJY1} as follows.
\begin{Theorem}\label{RCSVF}
Suppose that $ f $ is bounded. 
Under the Assumption \ref{Hypo}, 
the formulas \eqref{KOGSH} and \eqref{KIGSH} hold,
by replacing the notation $ \mathcal{S}^1 $ with $ S $: in other words, for any $t < T$, 
\begin{equation}\label{error5+}
\begin{split}
& E[ E [ 1_{ \{ \tau < T \}} \pi (f) (X_T) | \mathcal{F}_\tau ] | \mathcal{F}_t ]
= \int_0^T E [1_{\{\tau < s \}}
S_{T-s} f(X_s) |\mathcal{F}_{\tau \wedge t}] \,ds. 
\end{split}
\end{equation}
%In other words, the hedging error process 
%is equal to the integral with respect to $ s $ of the value process of 
%knock-in options with pay-off $ S_{T-s} f (X_s) $. 
%Here we a priori assumed that $ P (\tau > 0 ) =1 $. 
\end{Theorem}

\begin{proof} See Appendix \ref{RCSVF_Pr}. \qed
	\end{proof}

\subsection{Second order semi-static hedges}\label{SOSSH}

As we have seen in the previous section, 
the hedge error is represented by the integral 
with respect to $ s $ of knock-in options
with pay-off $ S_{T-s} f (X_s) $. 
For each of them 
we construct the static hedge by 
$ \pi^\bot S_{T-s} f (X_s) $
with infinitesimal amount $ ds $. 

To be more precise, for the knock-in option 
with pay-off $ S_{T-s} f $ for each $ s $, 
we adopt the Bowie-Carr type strategy by 
the option with pay-off $ \pi^\bot S_{T-s} f $;
we construct a portfolio 
composed of options with pay-off 
\begin{equation*}
\begin{split}
\pi^\bot S_{T-s} f (X_s) 
&= 
\{ S_{T-s} f (X_s) + S_{T-s} f (\theta(X_s))\}
1_{\{ X_s \not\in D \}}, 
\end{split}
\end{equation*}
at the volume ``$e^{-r(T-s)} ds$" for each $ s $. 
Note that 
$ \pi^\bot S_{T-s} f(X_s) $ may not be integrable in 
$ (s, \omega) \in [0,T] \times \Omega $, 
although it is in $ L^1 (P) $ for each $ s $ since 
$ S_{T-s} f  $ is bounded. 
Once it is conditioned, however, we retrieve the integrability;
\begin{Lemma}\label{intgso} The random variable 
$ E [\pi^\bot S_{T-s} f (X_s) |\mathcal{F}_\tau]  $ 
is jointly integrable in $ (s, \omega) \in [0,T] \times 
\Omega $%\{ \tau < s\} $. 
\end{Lemma}
\begin{proof}See Appendix \ref{intgso_Pr}.\qed\end{proof}

Let us consider the value of the ``portfolio". 
Until the knock-in time $ \tau $, all the options whose maturity 
is before $ \tau $ are cleared with pay-off zero. 
At the knock-in time, the hedger sells all the options 
at the price 
\begin{equation*}
E [\pi^\bot S_{T-s} f (X_s) |\mathcal{F}_\tau]. 
\end{equation*}
Thus, the value at time $ t $ of the strategy should be defined as
\begin{equation*}\label{vopt}
\Pi_t^{2,s} := e^{-r(T-t)} E [ %1_{\{\tau < s \}}
E [\pi^\bot S_{T-s} f (X_s) |\mathcal{F}_\tau] 
| \mathcal{F}_{t \wedge \tau} ], 
\end{equation*}
which, on $ \{ t < \tau \} $, is equal to
\begin{equation*}
e^{-r(T-t)} %E [ %1_{\{\tau < s \}}
E [\pi^\bot S_{T-s} f (X_s) %|\mathcal{F}_\tau] 
| \mathcal{F}_{t} ]. 
\end{equation*}
Since it is integrable in $ s \in [0,T]$, 
the total value at time $ t $ of the portfolios is given by
\begin{equation}\label{intval2}
\int_0^T \Pi_t^{2,s} \,ds =e^{-r(T-t)} \int_0^T %E [
E [\pi^\bot S_{T-s} f (X_s) %|\mathcal{F}_\tau]  
| \mathcal{F}_{t \wedge \tau} ] ds.
\end{equation}

\begin{Remark}
Lemma \ref{intgso} ensures the change of the order 
of the integrals to have another expression of the totality of the 
portfolio as
\begin{equation*}
\int_0^T \Pi_t^{2,s} \,ds =
e^{-r(T-t)} E [ \int_{0}^T 
E [ \pi^\bot S_{T-s} f (X_s) |\mathcal{F}_\tau] 
ds | \mathcal{F}_{t \wedge \tau} ],
\end{equation*}
In particular, discounted by $ e^{rt} $, it is a martingale.
This means that the portfolio is arbitrage-free, 
or should we say, it is still within the classical 
arbitrage theory.
\end{Remark}

As we have discussed in section \ref{GSSHE} as \eqref{koerr} and \eqref{kierr}, 
the hedge error of the strategy that holding $ \pi^\bot (\cdot) $ for a knock-in option 
coincides with the one 
by the $ \pi (\cdot) $ strategy
 for the corresponding knock-out option. 
So the error, evaluated at $ t $  for each maturity $ s $ is given by, in the infinitesimal form, 
\begin{equation*}
\begin{split}
\mathrm{Err}^s_{2,t}ds &:= 
e^{-r(s-t)} 
E [ E [ 1_{ \{ \tau < s \} } 
\pi S_{T-s} f (X_s) |\mathcal{F}_{\tau}]
|\mathcal{F}_t] e^{-r(T-s)} ds \\
&= e^{-r(s-t)} 
E [ E [ %1_{ \{ \tau \leq s \} } 
\pi S_{T-s} f (X_s) |\mathcal{F}_{\tau \wedge s}]
|\mathcal{F}_t] e^{-r(T-s)} ds.
\end{split}
\end{equation*}

\begin{Lemma}\label{interr2}
The error $ \mathrm{Err}^s_{2,t} $ for maturity $ s $
is integrable in $ s \in [0,T] $ and
\begin{equation}\label{2ndinterr}
\begin{split}
\int_0^T \mathrm{Err}^s_{2,t}  ds
&= e^{-r(T-t)}  
\int_0^T 
\int_u^T   E[ 1_{ \{ \tau < u \} }  S_{s-u} S_{T-s} f (X_u)
|\mathcal{F}_{t \wedge \tau}] ds du . \\ 
\end{split}
\end{equation}
\end{Lemma}

\begin{proof}See Appendix \ref{interr2_Pr}.\qed \end{proof}

Combining (\ref{intval2}) and Lemma \ref{interr2}
we have the following 
%%%%%%%%%%%%%%%%%%%%%%%%%%%%%%%%%%%%%%%%%%%%%%%%%%%%%%%
\begin{Theorem}\label{2ndexp}
It holds that, for each $ t>0 $, 
\begin{equation}\label{2nderr}
\begin{split}
& e^{-r(T-t)} \left(- E [ f (X_T) 1_{ \{ \tau > T \}} | \mathcal{F}_{t \wedge \tau} 
] + E [ \pi f (X_T)  | \mathcal{F}_{t \wedge \tau} 
] - \int_0^T ds \,\,E [ \pi^\bot S^{}_{T-s}f (X_s) 
| \mathcal{F}_{t \wedge \tau}] \right) \\
& \qquad %\hspace{3cm} 
= e^{-r(T-t)} \int_0^T  \,\,
\int_u^T E [  1_{ \{ \tau \leq u \} } S_{s-u} S_{T-s} f (X_u) |\mathcal{F}_{t \wedge \tau}] \, ds du \\
& \qquad %\hspace{3cm} 
\left( = e^{-r(T-t)} E [ \int_0^T  \,\,
\int_u^T E [  1_{ \{ \tau \leq u \} } S_{s-u} S_{T-s} f (X_u) |\mathcal{F}_\tau] \, ds du |\mathcal{F}_{t \wedge \tau}] \right).
\end{split}
\end{equation}
\end{Theorem}
In \eqref{2nderr}, the left-hand-hand side 
is the value of the knock-out option in short position 
and the static hedging position of 
the first and the second order.  
So the formula claims that the hedging error evaluated at
time $ t $ equals to the price of 
the doubly integrated knock-in options.

\begin{proof}
has been already done. \qed
\end{proof}

\begin{Remark}
Notice that the proposed framework is 
weaker than the one studied in \cite{FJY1}; here we identify 
the second order hedge via two-parameters,
while in \cite{FJY1} we have a one-parameter family of hedges.
The reason why we express it by double integral is that we are missing the integrability to ensure the change of the order.
The double integrability comes from (iii) of Theorem \ref{intbh} with the aid of integration by part.
\end{Remark}
%%%%%%%%%%%%%%%%%%%%%%%%%%%%%%

\subsection{Higher orders semi-static hedges}
\label{HOH}

This Subsection is devoted to the discussion of asymptotics of semi-static hedges, for orders higher than two. Let us consider for a moment the third order as an example. Equation (\ref{2nderr}) 
may suggest that the third order semi-static hedge can be written as function of the options with pay-off
\begin{equation*}
\pi^\bot S_{s-u} S_{T-s} f (X_u) 
%= \pi^\bot \int_0^{T-u}  S_{s} S_{T-s-u} f (x) ds
%S^{*2}_{T-u} f (x) 
\end{equation*}
maturing at $ u \in (0,T] $, parameterized by $ s \in (u,T] $
with infinitesimal amount $ e^{-r(T-u)} ds du  $. 
Once the integrability of 
\begin{equation*}
E [ \pi^\bot S_{s-u} S_{T-s} f (X_u) | \mathcal{F}_\tau ],
\end{equation*}
in $ (u,s) $ is established, 
we can say that the value of the hedging portfolio is given by
\begin{equation*}
e^{-r(T-t)} \int_0^T \int_u^{T} %
E [\pi^\bot  S_{s-u} S_{T-s} f (X_u) | \mathcal{F}_{t \wedge \tau} ] ds   
 du,
\end{equation*}
which is equivalent to
\begin{equation*}
e^{-r(T-t)} E [ \int_{0}^T 
\int_u^{T}  E [ \pi^\bot S_{s-u} S_{T-s} f (X_u) |\mathcal{F}_\tau] \,ds du
| \mathcal{F}_{t \wedge \tau} ]. 
\end{equation*}

Furthermore, for each $ (u,s) $, the error $\mathrm{Err}^{u,s}_{3,t}$ should be defined as
\begin{equation*}
\mathrm{Err}^{u,s}_{3,t} := e^{-r(T-t)} 
E [ E[ \pi
S_{s-u} S_{T-s}
f (X_s) ds | \mathcal{F}_{\tau \wedge s}]
|\mathcal{F}_t].
\end{equation*}
Notice that, by showing the integrability of $\mathrm{Err}^{u,s}_{3,t}$ in $ (u,s) $, by following Proposition \ref{interr2} we can write: 
\begin{equation*}
\begin{split}
%& \mathrm{Err}_{3,t} (f(X_T), f(\theta(X_T);D) := 
\int_0^T \int_0^s \mathrm{Err}^s_{3,t} \, du ds 
&= e^{-r(T-t)}\int_0^T  \int_u^T \int_s^{T} 
E [ 1_{ \{ \tau \leq u \} } 
 S_{s-u} S_{v-s} S_{T-v} 
f (X_u)  |\mathcal{F}_{t \wedge \tau}] dv ds du. \\
%&= e^{-r(T-t)} \int_0^T du \,\,
%E [ 1_{ \{ \tau \leq u \} } 
%S^{*3}_{T-u} f (X_u) |\mathcal{F}_{t \wedge \tau}].
\end{split}
\end{equation*}
Based on the above observation, we can thus construct the 
$ n $-th order static hedge and 
the corresponding error with the aggregation of $ 1,\cdots,n $-th 
hedges for any $ n \geq 3 $. 
The following Theorem extends the results stated in Theorem \ref{intbh} and has a key role in the determination of higher order hedges.
%represents a key  
%plays an essential role. 
%which enables us to define 
%$ S^{k*}_t f $ for $ k \geq 2 $. 
\begin{Theorem}\label{nhint} 
The following holds: 
(i) For $ n \geq 2 $, and for $ y_{n+1} \in \mathbf{R}^d $ 
and $0=  u_0 < u_1 <  \cdots   <u_n \leq T$, 
\begin{equation*}\label{primineq}
\begin{split}
&\int_{D^{n}}
\prod_{i=1}^{n}|h( u_{i}-u_{i-1}, y_{i+1}, y_{i})|dy_1 \cdots dy_n  \\
&\leq  % (u_N-u_{N-1})^{-\frac{1}{2}} 
\bigg( C_1 
+ 1_{ \{ y_{n+1} \in D^c \}} 
2MdC_2 (u_n- u_{n-1})^{-\frac{1}{2}} 
+ 1_{ \{y_{n+1} \in D \}} 
\frac{\delta C_2}{2}(u_n- u_{n-1})^{-\frac{1}{2}}
e^{-\frac{(\langle y_{n+1}, \gamma \rangle -k)^2}{4M (u_n-u_{n-1})}}
\bigg) \\
& \qquad %\hspace{3cm} 
\times \prod_{i=1}^{n}(u_{i}-u_{i-1})^{-\frac{1}{2}}
\sum_{I \subset \{ 1, \cdots, n-1\}} 
(\frac{\delta C_2}{2})^{|I|} C_1^{|I^c|}
\prod_{ j \in I }
(u_{j+1}-u_{j-1})^{-\frac{1}{2}},
\end{split}
\end{equation*}
where $ C_1 $ and $ C_2 $ are the ones given in 
\eqref{C1C2}, and $ M $ is the constant in 
\eqref{uelli}.\\

\noindent
(ii) For $ y_{n+1} \in D $, $ u_n \in (0,T)$ 
and $f \in L^\infty (D) $, %there exist constants $C_4$ such that for arbitrary $ \xi >0 $, 
\begin{equation}\label{secondeq}
\begin{split}
&\int_{0=u_0 < u_1 <\cdots < u_n }
du_1 \cdots du_{n-1}
\int_{D^N} |\prod_{i=1}^{n} h(u_{i} - u_{i-1}, y_{i+1}, y_i) 
||f(y_1)| dy_1 \cdots dy_n 
\\& \leq 
||f ||_{\infty} ( C_6 \delta)^{N-1} C_7 
\left(
 u_n^{-\frac{1}{2}}
+
(\langle y_{n+1}, \gamma \rangle -k)^{-\frac{3}{4}}
u_n^{-\frac{5}{8}} 
\right),
\end{split}
\end{equation}
where $ C_6 $ is a constant independent of 
$ N $ and $ \delta $, 
while $ C_7 $ is a constant depending on $ \delta $. 
\end{Theorem}
%%%%%%%%%%%%%%%%%%%%%%%%%%%%%%%%%%%%%%%%%%%%%%%%
	\begin{proof}
		See Appendix \ref{pnhint}. \qed
		\end{proof}

By Theorem \ref{nhint} we can define 
operators $ S_{u_n}^{*n} $ for $ u_n \in [0,T] $, $ n \geq 2 $, 
on $ L^\infty (D) $ by
\begin{equation*}
S_{u_n}^{*n}f(y_{n+1})
=\int_{0=u_0 < u_1 <\cdots < u_n }
du_1 \cdots du_{n-1}
\int_{D^n} \prod_{i=1}^{n} h(u_{i} - u_{i-1}, y_{i+1}, y_i) 
f(y_1) dy_1 \cdots dy_n ,
\end{equation*}
for $ y_{n+1} \in \mathbf{R}^d $.

\begin{Remark}\label{imppts}
For $ n\geq 2 $, 
\begin{equation*}
\begin{split}
S^{*n}_t f  (x) &= \int_0^t \int_D h(s, x, y) S^{*(n-1)}_{t-s} f (y) dyds \\
& = \int_0^t S_s  S^{*(n-1)}_{t-s} f (x) ds
\end{split}
\end{equation*}
with the convention that $ S^{*1}_t = S_t $.  
\end{Remark}

The following Theorem contains one of our most relevant theoretical results, by extending Theorem \ref{2ndexp}.
\begin{Theorem}\label{errorrep}
Under Assumption \ref{Hypo}, we have, for $ n \geq 2 $:\\

(i) the options for the $ n $-th hedge, 
$ E [ \pi^\bot S_{s-u} S^{*(n-1)}_{T-s} f(X_u) 
| \mathcal{F}_{\tau}] $, are integrable in 
$ (s,u,\omega) \in \{ (s,u) : 0 \leq u \leq s \leq T \}
\times \Omega $,\\

(ii) and the corresponding error is
\begin{equation*}
E [ 1_{ \{ \tau \leq u \} } S_{s-u}
S^{*n}_{T-s} f (X_u) |\mathcal{F}_{t \wedge \tau}],
\end{equation*}
which is also integrable in 
$ (s,u,\omega) \in \{ (s,u) : 0 \leq u \leq s \leq T \}
\times \Omega $.\\

(iii) As a consequence, we have, 
for each $ t $, 
\begin{equation}\label{param3}
\begin{split}
& e^{-r(T-t)}
\bigg( - E [ f (X_T) 1_{ \{ \tau > T \}} | \mathcal{F}_{t \wedge \tau} ] + E [ \pi f (X_T)  | \mathcal{F}_{t \wedge \tau} ] 
%\\&\hspace{3cm}
- \int_0^T \,
E [ \pi^\bot S_{T-s}f (X_s) 
| \mathcal{F}_{t \wedge \tau}] \,ds  
\\&\qquad 
- \sum_{h=2}^n \int_0^T \int_u^T \,
E [ \pi^\bot S_{s-u} S^{*(h-1)}_{T-s}f (X_u) 
| \mathcal{F}_{t \wedge \tau}] \,ds du  \bigg) \\
& %\hspace{3cm} 
= e^{-r(T-t)} \int_0^T \int_u^T \,
E [ 1_{ \{ \tau \leq u \} } S_{s-u}
S^{*n}_{T-s} f (X_u) |\mathcal{F}_{t \wedge \tau}]
\,ds du . 
\end{split}
\end{equation}\\

(iv) 
If $ \delta $ is sufficiently small, 
the right-hand-side of \eqref{param3} 
converges uniformly in $t $ to $0$ almost surely as $n  \rightarrow \infty$.\\

(v)
If $ \delta $ is sufficiently small, 
the series 
\begin{equation*}
     \sum_{h=2}^n \int_0^T \int_u^T \,
E [ \pi^\bot S_{s-u} S^{*(h-1)}_{T-s}f (X_u) 
| \mathcal{F}_{t \wedge \tau}] \,ds du
\end{equation*}
is absolutely convergent uniformly in $t $
almost surely as $ n \to \infty $, and 
\begin{equation*}
    \begin{split}
     E [ f (X_T) 1_{ \{ \tau > T \}} | \mathcal{F}_{t \wedge \tau} 
]
&=  E [ \pi f (X_T)  | \mathcal{F}_{t \wedge \tau} ] 
- \int_0^T \,
E [ \pi^\bot S_{T-s} f(X_s) 
| \mathcal{F}_{t \wedge \tau}] \,ds  
\\&\qquad  %\hspace{3cm} 
-\int_0^T \int_u^T \,
\sum_{h=2}^\infty E [ \pi^\bot S_{s-u} S^{*(h-1)}_{T-s} f(X_u) 
| \mathcal{F}_{t \wedge \tau}] \,ds du. 
    \end{split}
\end{equation*}
\end{Theorem}
%%%%%%%%%%%%%%%%%%%%%%%%%%%%%%%%%%%%%%%%%%%%%%%%%%%%%%%%%%%%%%%%
\begin{proof}
See Appendix \ref{errorrep_Pr}.
\qed\end{proof}

Roughly speaking, the results (i)--(iii) are obtained by repeating the procedure we did for 
Theorem \ref{2ndexp}.
To get the convergence result (iv) and (v)
we need extra efforts. 
The right-hand-side of \eqref{param3} basically gives the error estimate as 
multiple integral like Taylor expansion case.
If, let say, the integrand were bounded, 
the term would be dominated by $ (C^n)/n! $ 
for some constant $ C $ to ensure 
the convergence, but in our case, 
the naive integrability appearing all the time 
in this paper prevents from such a nice estimate.
Instead we work on a more precise estimate, 
with a reduction to a determinantal equation 
in Lemma \ref{spacon} and the
hyper-geometrical estimate in Lemma \ref{timcon}, in place of the standard 
exponential type estimate. 
In addition, a careful treatment of Gaussian 
type estimates is required to obtain (v).

	\section{Conclusion}
\label{sec:conc}
%	SECTION ADDED.\\
In the context of static hedge, the present paper introduces a methodology allowing to obtain asymptotic static hedge results for a fairly large class of multi-dimensional underlying assets' dynamics. From a financial point of view, we consider the problem of an investor who wants to hedge a portfolio of barrier options. 
The present paper extends the existing literature on static hedge by discussing the existence of asymptotic static hedging error and its convergence.  
Starting from the main results stated in \cite{FJY1}, the paper extends the asymptotic static hedge error construction to a more general mathematical setting. Both parametrix techniques and kernel symmetrization are considered to build in a systematic way 
the exact static hedging strategies of barrier options.

%%%%%%%%%%%%%%%%%%%%%%%%%%%%%

\appendix
\section{Appendix}\label{App}

\subsection{Proof of Lemma \ref{disconlem}}\label{lemdiscon}

%\begin{proof}
Let us introduce $x_D$ defined as
\begin{equation*}
x_D := \frac{ (k- \langle y, \gamma \rangle) x + 
(\langle x, \gamma \rangle -k) y}{\langle x-y,\gamma \rangle} ,
\end{equation*}
representing the intersection between the hyperplane 
and the straight line from $ x $ to $ y $. 
Notice that 
$ k- \langle y, \gamma \rangle \geq 0 $ and 
$ \langle x, \gamma \rangle -k > 0 $ since $ x \in D $ and 
$ y \in D^c $. 
As a consequence, we can write:
\begin{equation*}
\begin{split}
& A (x) - \tilde{A} (y)\\
&  = A (x) - \Psi A(\theta(y)) \Psi \\
& = A (x) - A (x_D) 
+ A (x_D) - \Psi A(\theta(x_D)) \Psi + \Psi A(\theta(x_D)) \Psi - \Psi A(\theta(y))\Psi \\
& \leq \Vert A (x) - A (x_D) \Vert 
+ \Vert \Psi ( \Psi^{-1} A (x_D) - A(\theta(x_D)) \Psi) \Vert
+  \Vert \Psi ( A(\theta(x_D)) -  A(\theta(y)) ) \Psi \Vert.
\end{split}
\end{equation*}
Since $ \Psi = I - 2 \gamma \otimes \gamma $ is orthogonal and $ \Psi^2 =1 $,
we have
\begin{equation*}
\Vert \Psi ( \Psi^{-1} A (x_D) - A(\theta(x_D)) \Psi) \Vert
= 2 \Vert [ A (x_D), \gamma \otimes \gamma ] \Vert
\end{equation*}
and 
\begin{equation*}
\Vert \Psi ( A(\theta(x_D)) -  A(\theta(y)) ) \Psi \Vert
= \Vert  A(\theta(x_D)) -  A(\theta(y)) \Vert.
\end{equation*}
Further, by using the above results we can write:
\begin{equation*}
\begin{split}
& \Vert A (x) - A (x_D) \Vert 
+ \Vert  A(\theta(x_D)) -  A(\theta(y)) \Vert \\
&\leq a_\infty( | x - x_D | + |x_D - y | ) \\
&= a_\infty \left( \left|\frac{\langle x, \gamma \rangle -k }
{\langle x-y,\gamma \rangle} (x-y) \right| 
+ \left|\frac{k- \langle y, \gamma \rangle}{\langle x-y,\gamma \rangle} (x-y)
\right| \right) \\
&= a_\infty | x-y |. 
\end{split}
\end{equation*}
Thus, the result stated in inequality (\ref{condiscon}) follows. \qed

%\end{proof}
%%%%%%%%%%%%%%%%%%%%%%%%%%%%%%%%%%%%%%%%%%%%%%%%%%%%%%%%%%%%%%
%%%%%%%%%%%%%%%%%%%%%%%%%%%%%%%%%%%%%%%%%%%%%%%%%%%%%%%%%%%%%%
\subsection{Proof of Lemma \ref{hestimate}}\label{hestimate_Pr}
%\begin{proof}
Before entering into the proof, we list below direct consequences of the inequalities (\ref{uelli}) of Assumption \ref{Hypo}. 
We write eigenvalues of $ A(y) $ by
$ \lambda_1 (y), \cdots, \lambda_d (y) $. 
Then, 
\begin{equation*}
m \leq \lambda_i (y) \leq M 
\end{equation*}
for any $ i $ and $ y \in \mathbf{R}^d $, 
and therefore,  
\begin{equation*}
	\begin{split}
	& m d^{\frac{1}{2}} \leq \Vert A (y) \Vert 
	= (\sum_i \lambda^2_i(y) )^{\frac{1}{2}}
	\leq M d^{\frac{1}{2}}, \\
	\end{split}
\end{equation*} 
and
\begin{equation*}
	m^d \leq \det A (y) = \prod_i \lambda_i (y) \leq M^d.
\end{equation*}
Moreover, since the eigenvalues of $ A^{-1} (y) $ are
$ \lambda_1^{-1} (y), \cdots, \lambda_d^{-1} (y) $,
we have that, 
for $ x \in \mathbf{R}^d $,  
\begin{equation*}
	M^{-1} |x|^2 \leq 
	({A}(y)^{-1} x) \cdot x \leq m^{-1} |x|^2,
\end{equation*}
\begin{equation}\label{invbdd}
	\begin{split}
	& M^{-1} d^{\frac{1}{2}} \leq \Vert {A}(y)^{-1} \Vert  \leq m^{-1}d^{\frac{1}{2}}, \\
	\end{split}
\end{equation}
and
\begin{equation}\label{invbdd2}
	\begin{split}
	&|{A}(y)^{-1} x |^2 = | {A}(y)^{-1} (A^{-\frac{1}{2}}x)
	\cdot A^{-\frac{1}{2}}x| \\
	& \leq m^{-1} |A^{-\frac{1}{2}}x|^2 \leq 
	m^{-1} | {A}(y)^{-1} x \cdot  x |  \\
	& \leq m^{-2} |x|^2.
	\end{split}
\end{equation}
Since $ \Psi $ in (\ref{smtrz}) is an orthogonal matrix, 
the inequalities in (\ref{uelli}), and hence the ones in the above, 
are valid for $ \tilde{A} $ as well. 
	
Now we start with looking at the equation (\ref{hzero}) to see that 
\begin{equation*}\label{decomp01}
	|h_0 (t,x,y)|  \leq I_1 + I_2 + I_3, 
\end{equation*}
where 
\begin{equation*}
	\begin{split}
	I_1 &:= \frac{1}{2t^2}  \Vert A(x) - \tilde{A}(y) \Vert
	|\tilde{A}(y)^{-1}(x-y)|^2 
	p_t (x,y), \\
	I_2 &:= \frac{1}{2t}\Vert A(x) - \tilde{A}(y)\Vert 
	\Vert \tilde{A}(y)^{-1}\Vert
	p_t (x,y), \\
	\end{split}
	\end{equation*}
	and 
	\begin{equation*}
	\begin{split}
	I_3 &:= 
	\frac{1}{t}|b(x) \cdot \tilde{A}(y)^{-1} (x-y) | p_t (x,y). 
	\end{split}
	\end{equation*}
By using the inequalities listed above, we have that 
	\begin{equation}\label{ptzy}
	\begin{split}
	p_t (x,y) 
	&= 
	(2 \pi)^{-\frac{d}{2}} \{\det \tilde{A}(y)t\}^{-\frac{1}{2}} 
	e^{- \frac{1}{2t} \langle \tilde{A}(y)^{-1} (x-y), x-y \rangle}
	\\
	&\leq 
	(2 \pi)^{-\frac{d}{2}} {m}^{-\frac{d}{2}} t^{-\frac{d}{2}}
	e^{- \frac{1}{2Mt} |x-y|^2}
	\\
	&= 
	2^{\frac{d}{2}}
	{m}^{-\frac{d}{2}}
	{M}^{\frac{d}{2}}
	p_t^{2M} (x,y) 
	e^{- \frac{1}{4Mt} |x-y|^2},
	\end{split}
	\end{equation}
	and 
	\begin{equation}\label{bA-1}
	\begin{split}
	| b(x) \cdot \tilde{A}(y)^{-1} (x-y) | 
	& \leq |b(x)| |\tilde{A}(y)^{-1} (x-y)|  \\
	& \leq b_\infty m^{-1}  \, |x-y|. 
	\end{split}
	\end{equation}
By (\ref{ptzy}) and (\ref{bA-1}), 
we obtain that, for $y \in \mathbf{R}^d $, 
	\begin{equation}\label{I3}
	\begin{split}
	I_3 &\leq (
	2^{\frac{d}{2}} 
	m^{-1-\frac{d}{2}}
	M^{ \frac{d}{2}} b_\infty ) \,
	t^{-1} |x-y|
	p_t^{2M} (x,y) 
	e^{- \frac{1}{4Mt} |x-y|^2} \\
	&\leq (2^{\frac{d}{2}} 
	m^{-1-\frac{d}{2}}
	M^{\frac{1}{2} + \frac{d}{2}} 
	b_\infty ) 
	t^{-\frac{1}{2}}
	p_t^{2M} (x,y) 
	\left(
	\frac{1}{4Mt} |x-y|^2
	\right)^{\frac{1}{2}}
	e^{- \frac{1}{4Mt} |x-y|^2} \\
	& \leq 2^{\frac{d}{2}} 
	m^{-1-\frac{d}{2}}
	M^{\frac{1}{2} + \frac{d}{2}} 
	b_\infty K_{\frac{1}{2}}t^{-\frac{1}{2}} 
	p_t^{2M} (x,y)=: 
	C_{13}' t^{-\frac{1}{2}} 
	p_t^{2M} (x,y).
	\end{split}
	\end{equation}
To estimate $ I_1 $ and $ I_2 $, we first consider the case of $ y \in D $. 
Since $ \tilde{A} (y) = A(y) $ in that case, we can use (\ref{Lipschitz1}), and by \eqref{invbdd2} and (\ref{ptzy}), we obtain that 
	\begin{equation}\label{I1-0}
	\begin{split}
	I_1 &\leq (
	2^{-1+\frac{d}{2}}
	m^{-2-\frac{d}{2}}
	{M}^{\frac{d}{2}}
	a_\infty )
	t^{-2}|x-y|^{3} 
	p_t^{2M} (x,y) 
	e^{- \frac{1}{4Mt} |x-y|^2} \\
	& = ( 2^{\frac{d}{2}+2}
	m^{-2-\frac{d}{2}}M^{\frac{3}{2}+\frac{d}{2}}
	a_\infty )
	t^{-\frac{1}{2}}
	p_t^{2M} (x,y) 
	\left(
	\frac{1}{4Mt} |x-y|^2
	\right)^{\frac{3}{2}}
	e^{- \frac{1}{4Mt} |x-y|^2} \\
	& \leq 2^{\frac{d+4}{2}}
	m^{-2-\frac{d}{2}}M^{\frac{3}{2}+\frac{d}{2}}
	a_\infty
	K_{\frac{3}{2}} t^{-\frac{1}{2}} 
	p_t^{2M} (x,y) =: C_{11}'t^{-\frac{1}{2}} 
	p_t^{2M} (x,y).
	\end{split}
	\end{equation}
Similarly, with \eqref{invbdd} in addition, 
	\begin{equation}\label{I2-0}
	\begin{split}
	I_2 &\leq
	( 2^{-1+\frac{d}{2}}  
	m^{-1-\frac{d}{2}}
	M^{\frac{d}{2}} a_\infty d^{\frac{1}{2}})
	t^{-1}|x-y| %^{\alpha} 
	p_t^{2M} (z,y) 
	e^{- \frac{1}{4Mt} |x-y|^2} \\
	& = (2^{\frac{d}{2}}
	m^{-1-\frac{d}{2}}
	M^{\frac{1}{2} + \frac{d}{2}} a_\infty d^{\frac{1}{2}})
	t^{-\frac{1}{2}} 
	p_t^{2M} (x,y) 
	\left(
	\frac{1}{4Mt} |x-y|^2
	\right)^{\frac{1}{2}}
	e^{- \frac{1}{4Mt} |x-y|^2} \\
	& \leq 2^{\frac{d}{2}}
	m^{-1-\frac{d}{2}}
	M^{\frac{1}{2} + \frac{d}{2}} a_\infty d^{\frac{1}{2}} 
	K_{\frac{1}{2}} t^{-\frac{1}{2}} 
	p_t^{2M} (x,y) =: C_{12}'t^{-\frac{1}{2}} 
	p_t^{2M} (x,y).
	\end{split}
	\end{equation}
Thus, we obtained that, for $ y \in D $, 
	\begin{equation}\label{h0-1}
	|h_0 (t,x,y)| 
	\leq (C_{13}' + C_{11}' + C_{12}') t^{-\frac{1}{2}} p_t^{2M} (x,y)
	= C_1 t^{-\frac{1}{2}} p_t^{2M} (x,y). 
	\end{equation}
	
Next, we consider the case where $ y \not\in D $. 
As has been remarked already, 
$ \tilde {A} $ is not continuous in general. 
We first consider the case $ x \in D^c $, where we can only use,
instead of \eqref{Lipschitz1}, 
\begin{equation}\label{2M}
	\begin{split}
	\Vert A (x) - \tilde A (y) \Vert
	\leq 2 \sup_{x \in D^c} \Vert A (x) \Vert \leq 2 M d. 
	\end{split}  
\end{equation}
We need to modify the estimates of $ I_1 $ and $ I_2 $.
By \eqref{2M} instead of \eqref{Lipschitz1} but still with \eqref{invbdd} and \eqref{ptzy}, 
	we have 
	\begin{equation}\label{I1d}
	\begin{split}
	I_1 & \leq ( 2^{\frac{d}{2}}
	m^{2-\frac{d}{2}}M^{1+\frac{d}{2}} d) 
	t^{-2}|x-y|^2  
	p_t^{2M} (x,y) 
	e^{- \frac{1}{4Mt} |x-y|^2} \\
	&= ( 2^{\frac{d}{2}+2}
	m^{2-\frac{d}{2}}M^{2+\frac{d}{2}} d) 
	t^{-1} 
	p_t^{2M} (x,y) 
	\left(
	\frac{1}{4Mt} |x-y|^2
	\right)
	e^{- \frac{1}{4Mt} |x-y|^2} \\
	& \leq ( 2^{\frac{d}{2}+2}
	m^{-2-\frac{d}{2}}M^{2+\frac{d}{2}} K_1 d) 
	t^{-1}
	p_t^{2M} (x,y) =: C_{21}t^{-1}
	p_t^{2M} (x,y),
	\end{split}
	\end{equation}
	and with \eqref{invbdd}, \eqref{ptzy}, and \eqref{2M}, 
	\begin{equation}\label{I2d}
	\begin{split}
	I_2 & \leq 
	( 2^{\frac{d}{2}}
	m^{-1-\frac{d}{2}}M^{1+\frac{d}{2}} d^{\frac{3}{2}}) 
	t^{-1}
	p_t^{2M} (x,y) 
	e^{- \frac{1}{4Mt} |x-y|^2} \\
	&\leq ( 2^{\frac{d}{2}}
	m^{-1-\frac{d}{2}}M^{1+\frac{d}{2}} d^{\frac{3}{2}}) 
	t^{-1}
	p_t^{2M} (x,y) =: C_{22}t^{-1}
	p_t^{2M} (x,y). 
	\end{split}
	\end{equation}
Combining \eqref{I1d}, \eqref{I2d}
with \eqref{I3}, we obtain that
for $ x, y \in D^c $, 
\begin{equation}\label{h0-3}
	\begin{split}
	|h_0 (t,x,y)| &\leq C_{13}' t^{-\frac{1}{2} }p_t^{2M} (x,y)
	+ (C_{21} + C_{22} t^{-1}p_t^{2M} (x,y) \\
	& \leq C_{1}t^{-\frac{1}{2} }p_t^{2M} (x,y) + 2Md C_2 (C_{21} + C_{22}) t^{-1}p_t^{2M} (x,y). 
	\end{split} 
	\end{equation}
	
Finally, we consider the case 
$ x \in D $ and $ y \in D^c $. 
We can then rely on \eqref{condiscon}.
We can actually combine 
\eqref{I1-0} and \eqref{I1d}
to obtain 
\begin{equation*}
	\begin{split}
	I_1 & \leq C_{11}' t^{-\frac{1}{2}} p_t^{2M} (x,y) + C_{21} \frac{\delta}{2Md} 
	t^{-1} p_t^{2M} (x,y),
	\end{split}
	\end{equation*}
	and by \eqref{I2-0} and \eqref{I2d},
	\begin{equation*}
	\begin{split}
	I_2 & \leq C_{12}' t^{-\frac{1}{2}} p_t^{2M} (x,y) + C_{22} \frac{\delta}{2Md} 
	t^{-1} p_t^{2M} (x,y). 
	\end{split}
	\end{equation*}
Since we still have \eqref{I3}, we obtain 
\begin{equation}\label{h0-2}
	\begin{split}
	|h_0 (t,x,y)| 
	&\leq (C_{11}' + C_{12}'+C_{13}')  t^{-\frac{1}{2}} p_t^{2M} (x,y) 
	+ \frac{\delta}{2Md} (C_{21} + C_{22})t^{-1 } p_t^{2M} (x,y)\\
	& = C_1 t^{-\frac{1}{2}} p_t^{2M} (x,y) + \delta C_2 t^{-1 } p_t^{2M} (x,y). 
	\end{split}
	\end{equation}
	By putting 
	(\ref{h0-1}), (\ref{h0-2})
	and \eqref{h0-3}
	together 
	we have \eqref{inequalityh_0}. 
	\qed
%\end{proof}

\subsection{Proof of Theorem \ref{intbh}}\label{pintbh}

\subsubsection{Proof of (i) of Theorem \ref{intbh}}

We first note that 
\begin{equation*}
\begin{split}
\int_{D} |h (t,x,y)| dy & \leq \int_{D} |h_0 (t,x,y)| dy 
+ \int_{D} |h_0 (t,x,\theta(y))| dy \\
&= \int_{D} |h_0 (t,x,y)| dy 
+ \int_{D^c} |h_0 (t,x,y)| dy.
\end{split}
\end{equation*}
Applying Lemma \ref{hestimate}, we now see that, by taking 
$ C_3 := \max ( C_1, 2Md C_2, \delta C_2) $, 
\begin{equation*}
\begin{split}
\int_{D} |h (t,x,y)| dy 
& \leq  C_3 t^{-\frac{1}{2}}
\int_D p^M_{2t} (x,y) \,dy + 
C_3 \int_{D^c} \left( t^{-1} p^M_{2t} (x,y) + t^{-\frac{1}{2}} p^M_{2t} (x,y)  
\right) \,dy \\
&= C_3 t^{-\frac{1}{2}} \int_{\mathbf{R}^d} p^M_{2t} (x,y) \,dy 
+  C_3 t^{-1} \int_{D^c} p^M_{2t} (x,y) \,dy \\
&=  C_3 t^{-\frac{1}{2}} + C_3 t^{-1} \int^k_{-\infty} 
\frac{1}{\sqrt{4 \pi M t}} 
e^{-\frac{(\langle x, \gamma\rangle - z)^2 }{4 M t}} \,dz. 
\end{split} 
\end{equation*}
In the case $x \in D$, 
since $ \langle x, \gamma\rangle -k > 0 $ and $k-z \geq 0$,  
we have that 
\begin{equation*}
\begin{split}
(\langle x, \gamma\rangle - z)^2 
& = (\langle x, \gamma\rangle -k + k- z)^2 \\
& = (\langle x, \gamma\rangle -k)^2 + (k- z)^2 
+ 2 (\langle x, \gamma\rangle -k)(k-z) \\
& \geq  (\langle x, \gamma\rangle -k)^2 + (k- z)^2.
\end{split}
\end{equation*}
Therefore, 
\begin{equation*}
\begin{split}
& \int^k_{-\infty} 
\frac{1}{\sqrt{4 \pi M t}} 
e^{-\frac{(\langle x, \gamma\rangle - z)^2 }{4 M t}} \,dz \\
&\leq  I_{\{x \in D\}} \int^k_{-\infty} 
\frac{1}{\sqrt{4 \pi M t}} 
e^{-\frac{(\langle x, \gamma\rangle - z)^2 }{4 M t}} \,dz 
+I_{\{x \in D^c \}} 
\int^{\infty}_{-\infty} 
\frac{1}{\sqrt{4 \pi M t}} 
e^{-\frac{(\langle x, \gamma\rangle - z)^2 }{4 M t}} \,dz \\
& \leq 1_{\{ x \in D\} } e^{- \frac{(\langle x, \gamma\rangle - k)^2 }{4 M t}}
\int^k_{-\infty} \frac{1}{\sqrt{4 \pi M t}} 
e^{-\frac{(k - z)^2 }{4 M t}} \,dz 
+ 1_{\{ x \in D^c \} } \\
& =\frac{1}{2} 
1_{\{ x \in D\} } e^{- \frac{(\langle x, \gamma\rangle - k)^2 }{4 M t}}
+ 1_{\{ x \in D^c \} }. 
\end{split}
\end{equation*} 
This completes the proof. \qed

%%%%%%%%%%%%%%%%%%%%%%%%%%%%%%%%%%%%%%%%%%%%%%%%%%%%%%%%%%%%%%
%%%%%%%%%%%%%%%%%%%%%%%%%%%%%%%%%%%%%%%%%%%%%%%%%%%%%%%%%%%%%%
\subsubsection{Proof of (ii) of Theorem \ref{intbh}}
It is a direct consequence of (\ref{ineq-q}) and 
Lemma \ref{hestimate}. \qed 

\subsubsection{Proof of (iii) of Theorem \ref{intbh}}

Let us recall that,
for $y \in \mathbf{R}^d $, %we see that 
\begin{equation*}\label{eq-qh}
\begin{split}
&\int_{\mathbf{R}^d} q_s (x,z)  h_0 (t-s, z,y) \,dz 
\\&= 
\int_{\mathbf{R}^d} q_s (x,z)  (L_z - L^y_z )
p_{t-s} (z,y) \,dz 
\\&= 
\int_{\mathbf{R}^d} q_s (x,z) \left(  
\frac{1}{2}\{ A(z) - \tilde{A}(y) \} \cdot \nabla^{\otimes 2}_z 
p_{t-s} (z,y)
+ b (z) \cdot \nabla_z 
p_{t-s} (z,y) \right) \,dz . 
\end{split}
\end{equation*}
Below we perform integration by parts;
\begin{equation*}
\begin{split}
&\int_{\mathbf{R}^d} q_s (x,z)  
\frac{1}{2}\{ A(z) - \tilde{A}(y) \} \cdot \nabla^{\otimes 2}_z 
p_{t-s} (z,y)dz \\
&= \int_{\mathbf{R}^d} 
\sum_{i,j=1}^d (q_s (x,z)  \frac{1}{2} 
\{ a_{i,j}(z) - \tilde{a}_{i,j}(y) \} )
\partial_{z_j} \partial_{z_i} 
p_{t-s} (z,y)dz
\\
&= \frac{1}{2}\int_{\mathbf{R}^d} 
\sum_{i,j=1}^d \partial_{z_j} (q_s (x,z)  
\{ a_{i,j}(z) - \tilde{a}_{i,j}(y) \} )
\partial_{z_i} 
p_{t-s} (z,y)dz
\\&= \frac{1}{2}\int_{\mathbf{R}^d} 
\sum_{i,j=1}^d 
(\partial_{z_j} q_s) (x,z)  
\{ a_{i,j}(z) - \tilde{a}_{i,j}(y) \}
+q_s (x,z)
\partial_{z_j}\{ a_{i,j}(z) - \tilde{a}_{i,j}(y) \}
 )\partial_{z_i} 
p_{t-s} (z,y)dz
\\&= \frac{1}{2}\int_{\mathbf{R}^d} 
( 
\{ A(z) - \tilde{A}(y) \}\nabla_z q_s (x,z)
+q_s (x,z)
\ ^{t}\nabla_z A(z) %- \tilde{A}(y) \}
 )\cdot \nabla
p_{t-s} (z,y)dz
%\\&= \frac{1}{2}\int_{\mathbf{R}^d} ( \{ A(z) - A(y) \}\nabla q_s (x,z)
%+q_s (x,z) \ ^{t}\nabla A(z) )\cdot \nabla p_{t-s} (z,y)dz
.
\end{split}
\end{equation*}
Therefore we obtain that 
\begin{equation*}
\begin{split}
&\left|\int_{\mathbf{R}^d} q_s (x,z)  h_0 (t-s, z,y) \,dz \right|
\\&= 
\left|\int_{\mathbf{R}^d} 
\left(\frac{1}{2}\{ A(z) - \tilde{A}(y) \}\nabla_z q_s (x,z)
+\frac{1}{2}
\ ^{t}\nabla_z A(z) q_s (x,z)
+ b (z) q_s (x,z)\right)
\cdot \nabla_z 
p_{t-s} (z,y) \,dz \right|
\\&\leq
\int_{\mathbf{R}^d} 
\left(\frac{1}{2}\left| \{A(z) - \tilde{A}(y) \}\nabla_z q_s (x,z) \right|
+\frac{1}{2}
\left|\ ^{t}\nabla_z A(z) \right|q_s (x,z)
+ \left|b (z) \right|q_s (x,z) \right)
\\ &\qquad 
\times (t-s)^{-1}\left| \tilde{A}(y)^{-1}(z-y)\right|
p_{t-s} (z,y) \,dz.
\end{split}
\end{equation*}
%First we consider the case that $ y \in D $,
%where $ \tilde{A} (y) = A (y) $. 

By \eqref{uelli}, 
\begin{equation*}
\begin{split}
p_{t-s} (z,y) 
&= 
(2 \pi)^{-\frac{d}{2}} \{\det \tilde{A}(y)(t-s)\}^{-\frac{1}{2}} 
e^{- \frac{1}{2(t-s)} \langle \tilde{A}(y)^{-1} (z-y), z-y \rangle}
\\
&\leq 
(2 \pi)^{-\frac{d}{2}} {m}^{-\frac{d}{2}} (t-s)^{-\frac{d}{2}}
e^{- \frac{1}{2M(t-s)} |z-y|^2}
\end{split}
\end{equation*}
and since $M_0$ in (\ref{ineq-q}) and (\ref{ineq-q1}) is greater than
$ M $ , 
we have that 
\begin{equation*}
p_{t-s} (z,y)  \leq m^{-\frac{d}{2}}M_0^{\frac{d}{2}} p_{t-s}^{M_0}(z,y).
\end{equation*}
Therefore, 
\begin{equation*}\label{eq-pm0}
\begin{split}
 &(t-s)^{-1}
\left|\tilde{A}(y) (z-y) \right|p_{t-s} (
z,y) \\
 & \leq (t-s)^{-1} m^{-1}
\left|z-y\right|p_{t-s} (z,y) 
\\&\leq 
 (t-s)^{-1}
\left|z-y\right|(2 \pi)^{-\frac{d}{2}} {m}^{-\frac{d}{2}} (t-s)^{-\frac{d}{2}}
e^{- \frac{1}{2M_0(t-s)} |z-y|^2}
\\&= 
 (t-s)^{-\frac{1}{2}}
(2 \pi)^{-\frac{d}{2}} {m}^{-\frac{d}{2}} (t-s)^{-\frac{d}{2}}
e^{- \frac{1}{4M_0(t-s)} |z-y|^2} \\
& \qquad 
\times (4M_0)^{\frac{1}{2}}\left\{\left( \frac{1}{4M_0(t-s)} |z-y|^2\right)^{\frac{1}{2}}
e^{- \frac{1}{4M_0(t-s)} |z-y|^2}\right\}
\\&\leq
2^{1+\frac{d}{2}}(t-s)^{-\frac{1}{2}} {m}^{-\frac{d}{2}} 
M_0^{\frac{d}{2} + \frac{1}{2}} K_{\frac{1}{2}}
p_{(t-s)}^{2M_0}(z,y).
%\Vert x^{\frac{1}{2}}e^{- x}\Vert_{\infty}.
\end{split}
\end{equation*}

On the other hand, since
\begin{equation*}\label{2pd}
\begin{split}
&\Vert 
A(z) - \tilde{A}(y) \Vert \leq \Vert A(z) \Vert + 
\Vert \tilde{A}(y) \Vert 
\leq 2 Md,
\end{split}
\end{equation*}
the inequality (\ref{ineq-q1}) implies that, 
\begin{equation*}
\frac{1}{2}\left| \{A(z) - \tilde{A}(y) \}\nabla_z q_s (x,z) \right|
\leq M d \, | \nabla_z q_s (x,z) | 
\leq M d\,C_q \,
s^{-\frac{d+1}{2}}e^{-\frac{ |x-z|^2}{4M_0s}},
\end{equation*}
and 
\begin{equation*}
\begin{split}
& \frac{1}{2} \left|\ ^{t}\nabla_z A(z) \right|q_s (x,z) \leq 
q_s (x,z) \left(  \sum_i 
\left|\sum_i \partial_{z_i} a_{ij} (z) \right|^2 \right)^{\frac{1}{2}} \\
& \leq q_s (x,z) \max_{1 \leq i,j \leq d}\Vert \partial_ia_{i,j}\Vert_{\infty} 
d^{\frac{3}{2}} 
\\& \leq 
C_q s^{-\frac{d}{2}} 
\max_{1 \leq i,j \leq d}\Vert \partial_ia_{i,j}\Vert_{\infty} 
d^{\frac{3}{2}} e^{-\frac{ |x-z|^2}{4M_0s}}
\\& 
=: 
C_q p d^{\frac{3}{2}} 
s^{-\frac{d}{2}} e^{-\frac{ |x-z|^2}{4M_0s}}
\end{split}
\end{equation*}
by (\ref{ineq-q}). Also by (\ref{ineq-q}), 
%and $ b_\infty = \max_{1 \leq i \leq d} ||b_{j}||_{\infty}$, 
we see that 
\begin{equation*}
%\begin{split}
|b(z)|q_s (x,z)
\leq b_\infty 
C_q s^{-\frac{d}{2}} 
e^{-\frac{ |x-z|^2}{4M_0s}}.
%\end{split}
\end{equation*}

Combining these altogether, 
we have that 
\begin{equation*}\label{eq56}
\begin{split}
&\left|\int_{\mathbf{R}^d} q_s (x,z)  h_0 (t-s, z,y) \,dz \right|
\\
\\&\leq 2^{1 + \frac{d}{2}}(t-s)^{-\frac{1}{2}} {m}^{-\frac{d}{2}} 
M_0^{\frac{d}{2} + \frac{1}{2}} K_{\frac{1}{2}}
\\ &\qquad \times 
\int_{\mathbf{R}^d}\left\{
M d\,C_q \,
s^{-\frac{d+1}{2}}e^{-\frac{ |x-z|^2}{4M_0s}} 
+
C_q pd^{\frac{3}{2}} 
s^{-\frac{d}{2}} e^{-\frac{ |x-z|^2}{4M_0s}}
+
K
C_q 
s^{-\frac{d}{2}} e^{-\frac{ |x-z|^2}{4M_0s}}
\right\}
p_{(t-s)}^{2M_0}(z,y)dz
\\&= 
2^{1 + \frac{3d}{2}}\pi^{\frac{d}{2}}
(t-s)^{-\frac{1}{2}} {m}^{-\frac{d}{2}} 
M_0^{\frac{3d}{2} + \frac{1}{2}} K_{\frac{1}{2}}
\\ &\qquad \times 
\int_{\mathbf{R}^d}
\left\{
M d\,C_q \,
s^{-\frac{1}{2}}
+
C_q pd^{\frac{3}{2}} 
+
KC_q 
\right\}
p_{s}^{2M_0}(x,z)
p_{(t-s)}^{2M_0}(z,y)dz
\\& \leq 
C_5(t-s)^{-\frac{1}{2}} (s^{-\frac{1}{2}}+1)  
\int_{\mathbf{R}^d}p_{s}^{2M_0}(x,z)
p_{(t-s)}^{2M_0}(z,y)dz
\\ & = 
C_5(t-s)^{-\frac{1}{2}} (s^{-\frac{1}{2}}+1)  
p_{2t}^{M_0}(x,y), 
\end{split}
\end{equation*}
where 
\begin{equation*}
C_5 := 2^{2 + \frac{3d}{2}}\pi^{\frac{d}{2}}
{m}^{-\frac{d}{2}} 
M_0^{\frac{3d}{2} + \frac{1}{2}} K_{\frac{1}{2}}
C_q\max\{
M d\,
,\ 
 pd^{\frac{3}{2}} 
+
b_\infty \}.
\end{equation*}
This completes the proof. \qed

\subsection{Proof of Lemma \ref{fof}.} \label{fof_Pr}
We first notice that 
	%when $ y \in D $,
\begin{equation*}\label{param1}
	\begin{split}
	\partial_s \{ q_s (x,z) p_{t-s} (z,y)\}
	= (L^*_z q_s) (x,z) p_{t-s} (z, y) - q_s (x,z) 
	(L_z^y p_{t-s}) (z,y), 
	\end{split}
\end{equation*}
by  (\ref{adj0}) and (\ref{Gauss}). 
Since 
\begin{equation*}\label{delta1}
	\lim_{s \downarrow 0} 
	\int_{\mathbf{R}^d} q_s (x,z) p_{t-s} (z,y) \,dz = p_t ( x, y) 
	\end{equation*}
	and 
	\begin{equation*}\label{delta2}
	\lim_{s \uparrow t} \int_{\mathbf{R}^d} q_s (x,z) p_{t-s} (z,y) \,dz 
	= q_t ( x, y), 
\end{equation*}
we have, 
	%where the subscript of the operator
	%$ \tilde{L} $ or $ L^* $ 
	%stands for the variable on which it is acting, 
\begin{equation*}\label{param2+}
	\begin{split}
	q_t ( x, y) - p_t ( x, y)
	& = \lim_{\epsilon \downarrow 0} 
	\int_\epsilon^{t-\epsilon} ds \int_{\mathbf{R}^d} dz
	\{ (L^*_z q_s) (x,z) p_{t-s} (z, y) - q_s (x,z) 
	(L_z^y p_{t-s}) (z,y) \} \\
	&= \lim_{\epsilon \downarrow 0} 
	\int_\epsilon^{t-\epsilon} ds 
	\int_{\mathbf{R}^d}\, dz \,\, q_s (x,z)  
	(L_z - L^y_z ) p_{t-s} (z,y) \\
	&= \lim_{\epsilon \downarrow 0} 
	\int_\epsilon^{t-\epsilon} ds 
	\int_{\mathbf{R}^d}\, dz \,\, q_s (x,z)  
	h_0(t-s, z,y)
	\\
	&= \int_0^t ds \int_{\mathbf{R}^d}\, dz \,\, q_s (x,z)  
	h_0(t-s, z,y).  \\
	\end{split}
\end{equation*}
The second equality and the last equality follows from (\ref{adj1})
and the integrability implied by (iii) 
of Theorem \ref{intbh}, respectively.
\qed	

\subsection{Proof of Theorem \ref{RCSVF}.}\label{RCSVF_Pr}	

By leveraging on the optional sampling theorem, 
\begin{equation}\label{errtau}
\begin{split}
& E [ 1_{ \{ \tau < T \}} \pi (f) (X_T) | \mathcal{F}_\tau ] \\
&= 
\left( \int_{D} f(y) q_{T-\tau} (X_\tau, y ) \,dy 
- \int_{D^c} f(\theta(y)) q_{T-\tau} (X_\tau, y ) \,dy \right) 1_{
	\{ \tau \leq T \} }.
\end{split}
\end{equation}
By applying Lemma \ref{fof}, %(\ref{param22}), 
we have
\begin{equation*}
\begin{split}
&\int_{D} f(y) q_{T-\tau} (X_\tau, y ) \,dy \\
& = \int_{D} f(y) p_{T-\tau} (X_\tau, y ) \,dy 
+ \int_{D} \left( \int_0^{T-\tau}  \int_{\mathbf{R}^d}\,  \,\, q_s (X_\tau,z)  
h_0 (T-\tau -s,z,y)dz\,ds \right) f(y) \,dy.
\end{split}
\end{equation*}
By (iii) of Theorem \ref{intbh}, 
we can change the order in the latter integral; 
$ \int_D dy \int_0^{T-\tau} ds $
to $ \int_0^{T-\tau} ds \int_D dy $,
and (ii) of Theorem \ref{intbh}
ensures that 
$ \int_D dy \int_{\mathbf{R}^d} dz $ 
can be replaced with 
$ \int_{\mathbf{R}^d} dz  \int_D dy $.
Thus we have 
\begin{equation*}
\begin{split}
&\int_{D} f(y) q_{T-\tau} (X_\tau, y ) \,dy \\
&=\int_{D} f(y) p_{T-\tau} (X_\tau, y ) \,dy +\int_\tau^{T} ds \int_{\mathbf{R}^d}\, dz \,\, q_{s-\tau} (X_\tau,z)  
\int_{D} h_0 (T-s,z,y) f(y) \,dy.
\end{split}
\end{equation*}
Similarly,
\begin{equation}\label{paramdc}
\begin{split}
&\int_{D^c} f(\theta(y)) q_{T-\tau} (X_\tau, y ) \,dy \\
&= \int_{D^c} f(\theta(y)) p_{T-\tau} (X_\tau, y ) \,dy 
+ \int_\tau^{T} ds \int_{\mathbf{R}^d}\, dz \,\, q_{s-\tau} (X_\tau,z)  
\int_{D^c} h_0 (T-s,z,y) f(\theta(y)) \,dy.
\end{split}
\end{equation}
The right-hand-side of (\ref{paramdc}) 
is equal to
\begin{equation*}
\int_{D} f(y) p_{T-\tau } (X_\tau, \theta(y) ) \,dy 
+ \int_\tau^{T} ds \int_{\mathbf{R}^d}\, dz \,\, q_{s-\tau} (X_\tau,z)  
\int_{D} h_0 (T-s,z,\theta(y)) f(y) \,dy
\end{equation*}
since $ \theta (D) = D^c \setminus \partial D $ 
and $ \theta^2 |_D = \mathrm{id}_D $. 
Now we see that (\ref{errtau}) equals 
\begin{equation}\label{errtau2}
\begin{split}
& \bigg( \int_{D} f(y) \{ p_{T-\tau} (X_\tau, y )- p_{T-\tau} (X_\tau, \theta(y) ) \} \,dy \\
& + \int_\tau^{T} ds \int_{\mathbf{R}^d}\, dz \,\, q_{s-\tau} (X_\tau,z)  
S_{T-s} f(z) \bigg) 1_{ \{ \tau < T \}}. 
\end{split}
\end{equation}
%Note that (iii) of Proposition \ref{beta} 
%ensures that the latter integral exists. 
We know from (\ref{pcsond}) that the first term of (\ref{errtau2}) 
is zero, and hence we have
\begin{equation}\label{error5}
\begin{split}
& E[ E [ 1_{ \{ \tau < T \}} \pi (f) (X_T) | \mathcal{F}_\tau ] | \mathcal{F}_t ] \\
&=  E [1_{ \{ \tau < T \}} 
\int_\tau^T E [%1_{\{\tau \leq s \}}
S_{T-s} f(X_s) |\mathcal{F}_\tau] \,ds
| \mathcal{F}_t].  \\
\end{split}
\end{equation}
By decomposing $ \{ \tau < T \} 
= \{ \tau < t \} \uplus \{ t \leq \tau < T \} $, we have that 
\begin{equation*}
\begin{split}
\text{The right-hand-side of \eqref{error5}}
&=  1_{\{\tau < t\}} 
\int_\tau^T E [%1_{\{\tau \leq s \}}
S_{T-s} f(X_s) |\mathcal{F}_\tau]\,ds \\
& \hspace{1cm} + 1_{\{t \leq \tau\} } 
E [
\int_{\tau \wedge T}^T E [%1_{\{\tau \leq s \}}
S_{T-s} f(X_s) |\mathcal{F}_{\tau \wedge T}] \,ds 
| \mathcal{F}_t] \\
&=  1_{\{\tau < t\}} 
\int_0^T E [1_{\{\tau < s \}}
S_{T-s} f(X_s) |\mathcal{F}_\tau]\,ds \\
& \hspace{1cm} +  1_{\{t \leq \tau\} } 
E [
\int_0^T E [1_{\{\tau < s \}}
S_{T-s} f(X_s) |\mathcal{F}_\tau] \,ds 
| \mathcal{F}_t].  \\
%&= e^{-r(T-t)} \int_t^T E [ 1_{\{\tau \leq s \}} 
%S_{T-s} f(X_s)  | \mathcal{F}_t]\,ds.  \hspace{2cm} \qed \\
\end{split}
\end{equation*}
Thus,  
\begin{equation}\label{GTMR3}
 E[ E [ 1_{ \{ \tau < T \}} \pi (f) (X_T) | \mathcal{F}_\tau ] |\mathcal{F}_t ] \\
= E [ \int_0^T
E [ 1_{ \{ \tau < s \}} S_{T-s}
f (X_s)  | \mathcal{F}_{\tau} ] \,ds | \mathcal{F}_{t} ].
\end{equation}

On the other hand, 
\begin{equation*}
\begin{split}
1_{\{\tau < s\}} E [
S_{T-s} f(X_s) |\mathcal{F}_\tau] &= 1_{\{\tau < s\}} \int_{\mathbf{R}^d} 
q_{s-\tau} (X_\tau,z) \left( \int_D h (T-s, z, y) f (y) \,dy \right) dz \\
&= 1_{\{\tau < s\}} \int_D \int_{\mathbf{R}^d} 
q_{s-\tau} (X_\tau,z)  h (T-s, z, y) dz  f (y) \,dy,
\end{split}
\end{equation*}
where the change of the order is valid
by (ii) of Theorem \ref{intbh} as we have seen, 
and by (iii) of Theorem \ref{intbh}, 
we have 
\begin{equation*}
\begin{split}
& \left| \int_D \int_{\mathbf{R}^d} 
q_{s-\tau} (X_\tau,z)  h (T-s, z, y) dz  f (y) \,dy \right| \\
&\leq  \int_D \left|\int_{\mathbf{R}^d} 
q_{s-\tau} (X_\tau,z)  h (T-s, z, y) dz
\right| |f (y)| \,dy  \\
& \leq C'' (s-\tau)^{-\frac{1}{2}}
(T-s)^{-\frac{1}{2}} 
\int_D 
(T-\tau)^{-\frac{d}{2}} \exp\left( -\frac{ |x-y|^2}{4M_0 (T-\tau)} \right) |f(y)| \,dy \\
& \leq C'' (4 \pi M_0)^{\frac{d}{2}} 
(s-\tau)^{-\frac{1}{2}}
(T-s)^{-\frac{1}{2}},
\end{split}
\end{equation*}
which is jointly integrable in $ (s,\omega) \in [0,T]\times \{ \tau < s \}$. Therefore 
we can change the order of the integral with respect to $ s $
and the conditional expectation with respect to $ \mathcal{F}_t $ in 
(\ref{error5}) as 
\begin{equation*}
\begin{split}
& 1_{\{t \leq \tau\} } 
E [
\int_0^T E [1_{\{\tau < s \}}
S_{T-s} f(X_s) |\mathcal{F}_\tau] \,ds 
| \mathcal{F}_t]
= 
1_{\{t \leq \tau\} } \int_0^T E [ E [1_{\{\tau < s \}}
S_{T-s} f(X_s) |\mathcal{F}_\tau]  | \mathcal{F}_t]  \,ds \\
& \hspace{2cm}  = 1_{\{t \leq \tau\} } 
\int_0^T E [1_{\{\tau < s \}}
S_{T-s} f(X_s)  | \mathcal{F}_t]  \,ds. \\
\end{split}
\end{equation*}
The proof is concluded by observing that we have obtained the expression given in Equation \eqref{error5+}. 
\qed

\subsection{Proof of Lemma \ref{intgso}.}\label{intgso_Pr}	
Since
\begin{equation*}
	\begin{split}
	E [\pi^\bot S_{T-s} f (X_s) |\mathcal{F}_\tau] 
	= 1_{ \{\tau \leq s \}} \int_{D^c} q_{s -\tau}
	(X_{\tau}, y) \{ S_{T-s} f (y) + S_{T-s} f (\theta(y) )
	\} \,dy \\
	+ 1_{ \{ \tau > s \}} 1_{\{X_s \in D^c\}} 
	(S_{T-s} f (X_s) + S_{T-s} f (\theta(X_s)) ) 
	\\
	\end{split}
\end{equation*}
and since %the last term $ = 0 $, 
$ \{ \tau > s \} \cap \{X_s \in D^c\} = \emptyset $, the second term is zero. 
Therefore
we see that
\begin{equation*}
	\begin{split}
	| E [\pi^\bot S_{T-s} f (X_s) |\mathcal{F}_\tau] |
	& \leq 1_{ \{\tau < s \}}  \left| \int_{D^c} q_{s -\tau}
	(X_{\tau}, y) \int_{D^c} h (T-s, y, z) f(z) \,dz\, dy \right|  \\
	&+ 1_{ \{\tau < s \}} \left| \int_{D^c} q_{s -\tau}
	(X_{\tau}, y) \int_D h (T-s, \theta(y),z)f(z)  \,dz\, dy \right| \\
	& \leq 1_{ \{\tau < s \}} \left| \int_{\mathbf{R}^d} q_{s -\tau}
	(X_{\tau}, y) \int_D h (T-s,y, z) f(z) \,dz\, dy \right| \\
	&+ 1_{ \{\tau < s \}} \left| \int_{D} q_{s -\tau}
	(X_{\tau}, y) \int_D h (T-s, y, z) f(z) \,dz\, dy \right| \\
	&+ 1_{ \{\tau < s \}} \left| \int_{D^c} q_{s -\tau}
	(X_{\tau}, y) \int_D h (T-s, \theta(y), z) f(z) \,dz\, dy \right| \\
	& =: 1_{ \{\tau < s \}}(II_1 + II_2 + II_3).
	\end{split}
\end{equation*}
Theorem \ref{intbh} (iii) implies that  
\begin{equation*}
	\begin{split}
	II_1 &\leq  \Vert f \Vert_\infty \int_D \left| \int_{\mathbf{R}^d} q_{s -\tau}
	(X_{\tau}, y) h (T-s, y, z) \,dy \,\right| dz  \\
	& \leq  \Vert f \Vert_\infty \int_{\mathbf{R}^d} \left| \int_{\mathbf{R}^d} q_{s -\tau}
	(X_{\tau}, y) h_0 (T-s, y, z) \,dy \,\right| dz \\
	&\leq \Vert f \Vert_\infty C_3 (s-\tau)^{-\frac{1}{2}} (T-s)^{-\frac{1}{2}}
	\end{split}.
\end{equation*}
%for some positive constant $ C_1 $. 
Here the change of the order of the integral is valid by (ii) of Theorem \ref{intbh} as before. 
Also, 
\begin{equation*}
	\begin{split}
	II_2 &\leq \Vert f \Vert_\infty \int_{D} q_{s-\tau} (X_\tau,y) 
	\int_D |h(T-s,y,z)| dy dz \\
	& \leq \Vert f \Vert_\infty C_3
	\int_D (s-\tau)^{-\frac{d}{2}}e^{-\frac{ |X_\tau-y|^2}{4M_0 (s-\tau) }}
	((T-s)^{-\frac{1}{2}} + (T-s)^{-1}  e^{-\frac{|k- \langle \gamma, y \rangle|^2}{4M (T-s)}}) \,dy 
	\end{split}
\end{equation*}
by (\ref{ineq-q}) and Theorem \ref{intbh} (i). 
Here, we note that
\begin{equation*}
	\begin{split}
	& \int_D (s-\tau)^{-\frac{d}{2}}e^{-\frac{ |X_\tau-y|^2}{4M_0 (s-\tau) }}
	(T-s)^{-\frac{1}{2}}  e^{-\frac{|k- \langle \gamma, y \rangle|^2}{4M (T-s)}}
	\,dy \\
	& = (s-\tau)^{-\frac{d}{2}}(T-s)^{-\frac{1}{2}}
	\int_D e^{- \sum_{i=1}^{d-1}
		\frac{ \langle X_\tau -y, \gamma_i \rangle^2}{4M_0 (s-\tau)} }
	e^{-\frac{ \langle X_\tau -y, \gamma \rangle^2}{4M_0 (s-\tau)} }
	e^{-\frac{(k- \langle \gamma, y \rangle)^2}{4M (T-s)}} \,dy=:II_2', \\
	\end{split}
\end{equation*}
where $ \{ \gamma_i : i=1,\cdots d-1 \} $ is 
an orthonormal basis of $ (\partial D)^\bot $. 
Since $ \langle X_\tau , \gamma \rangle =k $, we notice that
\begin{equation*}
	\begin{split}
	II_2'& = (4 \pi M_0)^{\frac{d-1}{2}}
	(s-\tau)^{-\frac{1}{2}}(T-s)^{-\frac{1}{2}} \int_k^\infty 
	e^{-\frac{ (k -z)^2}{4M_0 (s-\tau) }}
	e^{-\frac{(k-z)^2}{4M (T-s)}}
	\,dz \\
	& \leq (4 \pi M_0)^{\frac{d-1}{2}} 
	(s-\tau)^{-\frac{1}{2}} (T-s)^{-\frac{1}{2}}  
	\int_k^\infty e^{-\frac{ (k -z)^2}{4M_0} ( (s-\tau)^{-1} + (T-s)^{-1} )}
	\,dz \\
	& =  \frac{1}{2} (4 \pi M_0)^{\frac{d}{2}}(T-\tau)^{-\frac{1}{2}}.
	\end{split}
\end{equation*}
Therefore, we obtain that
\begin{equation}\label{boundII}
	II_2 \leq 
	{C_3 \Vert f \Vert_\infty} (4 \pi M_0)^{\frac{d}{2}}(T-s)^{-\frac{1}{2}}
	\left( \frac{(T-\tau)^{-\frac{1}{2}}}{2} + 1 \right).
\end{equation}
Further, by Theorem \ref{intbh} (i), 
\begin{equation*}
	\begin{split}
	II_3 &= 1_{ \{\tau < s \}} \left| \int_{D} q_{s -\tau}
	(X_{\tau}, \theta(y)) \int_D h (s, y, z) f(z) \,dz\, dy \right| \\
	& \leq \Vert f \Vert_\infty \int_{D} q_{s-\tau} (X_\tau, \theta(y)) 
	\int_D |h(T-s,y,z)| dy dz \\
	& \leq \Vert f \Vert_\infty C_3
	\int_D (s-\tau)^{-\frac{d}{2}}e^{-\frac{ |X_\tau-\theta(y)|^2}{4M_0 (s-\tau) }}
	((T-s)^{-\frac{1}{2}} + (T-s)^{-1}  e^{-\frac{|k- \langle \gamma, y \rangle|^2}{4M (T-s)}}) \,dy, 
	\end{split}
\end{equation*}
and since 
$  |X_\tau-\theta(y)| = | \theta (X_\tau) - y| = | X_\tau - y| $,
we have the same bound as (\ref{boundII}). 
Now we see that all of $ II_1 $, $ II_2 $, and $ II_3 $ are jointly integrable since
$ \int_0^T  1_{ \{\tau < s \}} (II_1 + II_2 + II_3)\, ds $
is in $ L^\infty (\Omega) $. 
\qed

\subsection{Proof of Lemma \ref{interr2}.}\label{interr2_Pr}	

We can apply Theorem \ref{RCSVF} 
to $ \mathrm{Err}^s_{2,t} $ for each 
$ s \in (0,T) $ since 
$ S_{T-s} f $ is bounded as we have seen in Corollary \ref{bddness}.
Then we obtain that 
\begin{equation*}
\begin{split}
\mathrm{Err}^s_{2,t} 
%\equiv e^{-r (T-s)} 
%\mathrm{Err}_s (S_{T-s} f(X_s), %S_{T-s} f(\theta(X_s)); D) \\&
= e^{-r(T-t)} \int_0^s 
E [ 1_{ \{ \tau < u \} } S_{s-u} S_{T-s} f (X_u) |\mathcal{F}_{\tau \wedge t}]\,du.
\end{split}
\end{equation*}
To see its integrability
in $ s $, we rather use the intermediate form (\ref{GTMR3});
\begin{equation*}
\begin{split}
& \mathrm{Err}^s_{2,t}  = e^{-r(s-t)} E [
\int_\tau^s E [1_{\{\tau < u \}}
S_{s-u} S_{T-s} f(X_u) |\mathcal{F}_\tau] \,du
| \mathcal{F}_t]. 
\end{split}
\end{equation*}
Since we know that
\begin{equation*}
q_{u-\tau} (X_\tau, y) 
h(s-u, y,z) S_{T-s} f (z) 
\end{equation*}
is integrable in $ (y,z) $ by (ii)
of Theorem \ref{intbh}, we have on $ \{ \tau < u \} $, 
\begin{equation*}
\begin{split}
&\left|  E [%1_{\{\tau \leq s \}}
S_{s-u} S_{T-s} f(X_u) |\mathcal{F}_\tau] \right| \\
&= \left| \int_D \left( \int_{\mathbf{R}^d}  q_{u-\tau} (X_\tau, y) 
h(s-u, y,z) dy  \right) S_{T-s} f (z) \,dz \right| \\
& \leq \Vert f \Vert_\infty
\int_D \left| \int_{\mathbf{R}^d}  q_{u-\tau} (X_\tau, y) 
h(s-u, y,z) dy  \right| \left( \int_D |h(T-s,z,w)| \,dw \right) dz. \\
\end{split}
\end{equation*}
By applying (i) and (iii) of Theorem \ref{intbh}, we see that
it is dominated by 
\begin{equation*}
\begin{split}
&  C_3 C_5 \Vert f \Vert_\infty
(u-\tau)^{-\frac{1}{2}} (s-u)^{-\frac{1}{2}} (s-\tau)^{-\frac{d}{2}} \\
& \times
\int_D  \left( e^{-\frac{ |X_\tau -y|^2}{4M_0 (s-\tau)} }
+ e^{-\frac{ |X_\tau -\theta(y)|^2}{4M_0 (s-\tau)}} \right)
\{(T-s)^{-\frac{1}{2}} + (T-s)^{-1} 
e^{-\frac{|k- \langle \gamma, y \rangle|^2}{4M (T-s)}}\} \, dy
=:III. \\
\end{split}
\end{equation*}
Then by a similar calculation we did 
for $ II_2' $ in the proof of 
Lemma \ref{intgso}, we obtain that 
\begin{equation*}
\begin{split}
III &\leq  C_3 C_5  \Vert f \Vert_\infty (u-\tau)^{-\frac{1}{2}} 
(s-u)^{-\frac{1}{2}} (T-s)^{-\frac{1}{2}} (4 \pi M_0)^{\frac{d}{2}}\\
& \qquad \times 
\left( 1+ 2
(s-\tau)^{-\frac{1}{2}}
(T-s)^{-\frac{1}{2}}(4 \pi M_0)^{-\frac{1}{2}}
\int_k^\infty e^{-\frac{ (k- z)^2}{4 M_0} ( (s-\tau)^{-1} + (T-s)^{-1} )} dz\right) \\
& =  C_3 C_5  \Vert f \Vert_\infty (u-\tau)^{-\frac{1}{2}} 
(s-u)^{-\frac{1}{2}} (T-s)^{-\frac{1}{2}} (4 \pi M_0)^{\frac{d}{2}}( 1+ 
(T-\tau))
%\\& \leq C' \Vert f \Vert_\infty (u-\tau)^{-\frac{1}{2}} 
%(s-u)^{-\frac{1}{2}} (T-s)^{-\frac{1}{2}},
.
\end{split} 
\end{equation*}
Now we see that $ 1_{\{\tau \leq u \}} E [
S_{s-u} S_{T-s} f(X_u) |\mathcal{F}_\tau] 
$ is integrable in $ (s,u, \omega) $ on $ \{ (s,u) : 0 \leq u \leq s \leq T \}
\times \Omega $ and 
the totality of the error is then obtained as
\begin{equation}\label{E2t}
\begin{split}
\int_0^T \mathrm{Err}^s_{2,t}  ds &= e^{-r(T-t)} 
\int_0^T ds \int_0^s du \,
E [ 1_{ \{ \tau \leq u \} } S_{s-u} S_{T-s} f (X_u) 
|\mathcal{F}_{t \wedge \tau} ]. \\
%&=:\mathrm{Err}_{2,t}(f(X_T), f (\theta(X_T);D). 
\end{split}
\end{equation}
By the change of 
the order of the integrals in (\ref{E2t}), 
we get (\ref{2ndinterr}). 
\qed

%%%%%%%%%%%%%%%%%%%%%%%%%%%%%%%%%%%%%%%%%%%%%%%%%%%%%%%%%%%%%%
%%%%%%%%%%%%%%%%%%%%%%%%%%%%%%%%%%%%%%%%%%%%%%%%%%%%%%%%%%%%%%
\subsection{Proof of Theorem \ref{nhint}}\label{pnhint}
%\subsubsection{Proof of (i) of Theorem \ref{nhint}}
\subsubsection{Estimates for the $ n $-time  
``convolution" in space-variable of $ h $``convolution" in space-variable of $ h $ }
\begin{Lemma}\label{spacon}
For $ n \geq 2$ and $ y_{n+1} \in D $, 
we have
\begin{equation*}\label{estA3}
\begin{split}
& \int_{D^{n}}
\prod_{i=1}^{n}| h( s_{i}, y_{i+1}, y_{i})| dy_{i} \\
& \leq 
\sum_{A \subset \{ 1, \cdots, n\} } 
\left( 1_{ \{ n \in A^c \} } + 
s_n^{-\frac{1}{2}} e^{-\frac{(\langle y_{n+1}, \gamma \rangle -k)^2}{4M s_n}}
1_{ \{ n \in A \} } \right)
(\frac{\delta C_2}{2})^{|A|} C_1^{|A^c|}
\prod_{ j \in A \setminus \{n\}}
(s_j+s_{j+1})^{-\frac{1}{2}} \prod_{i=1}^{n} s_{i}^{-\frac{1}{2}} \\
&= \left( C_1 s_n^{-\frac{1}{2}} + \frac{\delta C_2}{2} s_n^{-1}e^{-\frac{(\langle y_{n+1}, \gamma \rangle -k)^2}{4M s_n}} \right)\sum_{A \subset \{ 1, \cdots, n-1\} } 
(\frac{\delta C_2}{2})^{|A|} C_1^{|A^c|}
\prod_{ j \in A }
(s_j+s_{j+1})^{-\frac{1}{2}} \prod_{i=1}^{n-1} s_{i}^{-\frac{1}{2}}. 
\end{split}
\end{equation*}
%where $ h $, $ C_1 $, $ C_2 $, and $ \delta $
%are those given in the main part. 
\end{Lemma}

\begin{proof}
We first note that, by Lemma \ref{hestimate},
\begin{equation*}
\begin{split}
|h (t,x,y)| & \leq |h_0 (t,x,y)| + |h_0 (t,x, \theta(y))| \\
& \leq C_1 t^{-\frac{1}{2}} \{ p_t^{2M} (x,y) + p_t^{2M} (x,\theta(y)) \}
+ C_2 \delta t^{-1} p_t^{2M} (x,\theta(y)). 
\end{split}
\end{equation*} 
Therefore, 
we have 
\begin{equation}\label{whole1}
\begin{split}
& \int_{D^{n}}
\prod_{i=1}^{n}| h( s_{i}, y_{i+1}, y_{i})| dy_{i} \\
& \leq \prod_{i=1}^n s_{i}^{-\frac{1}{2}} 
\sum_{A\subset \{ 1, \cdots, n\} } \bigg\{
\int_{D^{n}} \prod_{ i \in A } (\delta C_2)^{|A|} 
s_i^{-\frac{1}{2}}
p^{2M}_{s_{i}} (y_{i+1}, \theta(y_{i})) dy_{i} \\ 
& \qquad \times 
\prod_{j \in A^c} C_1^{|A^c|} 
\{ p_{s_j}^{2M} (y_{j+1},y_{j}) 
+ p_{s_j}^{2M} (y_{j+1}, \theta(y_{j})) \} dy_{j}
\bigg\}.
\end{split}
\end{equation}

The integral of the right-hand-side of
(\ref{whole1}) is reduced to one dimensinal one. 
In fact, the change of variables 
\begin{equation*}
(z_j, z_j^1,  \cdots, z_j^{d-1}) = ( \langle y_j, \gamma \rangle, 
\langle y_j, \gamma_1 \rangle
\cdots, \langle y_j, \gamma_{d-1} \rangle ) =: G (y_j) 
\end{equation*}
separates both the domain $ D $ and the heat kernel $ p $ as
$ G(D) = [k, \infty) \times \mathbf{R}^{d-1} $, 
\begin{equation*}
\begin{split}
p_t^{2M} (y_{j+1},y_j) 
%&= (4 \pi M  t)^{-\frac{d}{2}} 
%e^{- \sum_{l=1}^d \langle y_{j+1}-y_j, \gamma_l \rangle^2 /{4Mt}} \\
&= (4 \pi M t)^{-\frac{1}{2}}  e^{- (z_{j+1} - z_j)^2 /{4Mt}}
(4 \pi M t)^{-\frac{d-1}{2}} e^{- \sum_{l=1}^{d-1} (z_{j+1}^l - z_j^l)^2 /{4Mt}},
\end{split}
\end{equation*}
and 
\begin{equation*}
\begin{split}
p_t^{2M} (y_{j+1},\theta(y_j)) 
%&= (4 \pi M  t)^{-\frac{d}{2}} 
%e^{- \sum_{l=1}^d \langle y_{j+1}- \theta(y_j), \gamma_l \rangle^2 /{4Mt}} \\
&= (4 \pi M t)^{-\frac{1}{2}}  e^{- (z_{j+1} + z_j -2k)^2 /{4Mt}}
(4 \pi M t)^{-\frac{d-1}{2}} e^{- \sum_{l=1}^{d-1} (z_{j+1}^l - z_j^l)^2 /{4Mt}}.
\end{split}
\end{equation*}
Therefore, for each $ A \subset \{a,\cdots, n\}$, 
\begin{equation}\label{intI}
\begin{split}
& \int_{D^n} \prod_{ i \in A } 
p^{2M}_{s_{i}} (y_{i+1}, \theta(y_{i})) dy_{i} 
\prod_{j \in A^c}
\{ p_{s_j}^{2M} (y_{j+1},y_{j}) 
+ p_{s_j}^{2M} (y_{j+1}, \theta(y_{j})) \} dy_{j} \\
&= \int_{[k,\infty)^n} \prod_{ i \in A } 
(4 \pi M s_i)^{-\frac{1}{2}} 
e^{- (z_{i+1} + z_i-2k)^2 /{4Ms_i}}
dz_i \\
& \hspace{2cm} \times \prod_{j \in A^c} (4 \pi M s_j)^{-\frac{1}{2}}
\{ e^{- (z_{j+1} - z_j)^2 /{4Ms_j}}
+  e^{- (z_{j+1} + z_j-2k)^2 /{4Ms_j}} \} dz_{j}.
%\mathcal{I}_I.
\end{split}
\end{equation}

Since $ z_i-k \geq 0 $ for all $ i $, we have in particular
\begin{equation*}
e^{- (z_{i+1} + z_i-2k)^2 /{4Ms_i}} 
\leq e^{- \{(z_{i+1}-k)^2 + (z_i-k)^2\} /{4Ms_i}}.
\end{equation*}
Therefore, the integral of (\ref{intI}) is dominated by
\begin{equation*}
\begin{split}
&\int_{[k,\infty)^n} \prod_{ i \in A } 
(4 \pi M s_i)^{-\frac{1}{2}} 
 e^{- \{(z_{i+1}-k)^2 + (z_i-k)^2 \}/{4Ms_i}} dz_i \\
& \hspace{2cm} \times \prod_{j \in A^c} (4 \pi M s_j)^{-\frac{1}{2}}
\{ e^{- (z_{j+1} - z_j)^2 /{4Ms_j}}
+  e^{- (z_{j+1} + z_j-2k)^2 /{4Ms_j}} \} dz_{j}=:\mathcal{I}_A .
\end{split}
\end{equation*}

We can further reduce the integral $ \mathcal{I}_A $
as follows: by the shift $ z_j \mapsto z_j +k $, 
\begin{equation*}
\begin{split}
\mathcal{I}_A &= \int_{[0,\infty)^n} \prod_{ i \in A } 
(4 \pi M s_i)^{-\frac{1}{2}} 
e^{- \{(z_{i+1})^2 + (z_i)^2\} /{4Ms_i}}
dz_i \\
& \hspace{2cm} \times \prod_{j \in A^c} (4 \pi M s_j)^{-\frac{1}{2}}
\{ e^{- (z_{j+1} - z_j)^2 /{4Ms_j}}
+  e^{- (z_{j+1} + z_j)^2 /{4Ms_j}} \} dz_{j},
\end{split}
\end{equation*}
and for an even function $ f $, 
\begin{equation*}
\begin{split}
\int_0^\infty \{ e^{- (z_{j+1} - z_j)^2 /{4Ms_j}}
+  e^{- (z_{j+1} + z_j)^2 /{4Ms_j}} \} f(z_j) dz_{j}
= \int_{-\infty}^\infty 
e^{- (z_{j+1} - z_j)^2 /{4Ms_j}} f(z_j) dz_j.
\end{split}
\end{equation*}
These two facts imply that
\begin{equation*}
\begin{split}
\mathcal{I}_A = \int_{[0,\infty)^{|A|} \times \mathbf{R}^{|A^c|}}
\prod_{ i \in A } 
(4 \pi M s_i)^{-\frac{1}{2}} 
e^{- \{(z_{i+1})^2 + (z_i)^2\} /{4Ms_i}}
dz_i \prod_{j \in A^c} 
(4 \pi M s_j)^{-\frac{1}{2}} e^{- (z_{j+1} - z_j)^2 /{4Ms_j}} dz_{j}.
\end{split}
\end{equation*}
Even further, 
since the integral is invariant under the translations 
$ z_i \mapsto - z_i $
for $ i \in A $, 
\begin{equation*}\label{int2}
\mathcal{I}_A = 2^{-|A|} 
\int_{\mathbf{R}^{n}}
\prod_{ i \in A } 
(4 \pi M s_i)^{-\frac{1}{2}} 
e^{- \{(z_{i+1})^2 +(z_i)^2\} /{4Ms_i}}
dz_i \prod_{j \in A^c} 
(4 \pi M s_j)^{-\frac{1}{2}} e^{- (z_{j+1} - z_j)^2 /{4Ms_j}} dz_{j}.
\end{equation*}

We note that
\begin{equation*}
\begin{split}
& \prod_{ i \in A}  e^{- \{(z_{i+1})^2 +(z_i)^2\} /{4Ms_i}} \prod_{j \in A^c} e^{- (z_{j+1} - z_j)^2 /{4Ms_j}} \\
& = e^{- \frac{1}{4M} \langle \mathcal{H}_A \vec{z}, \vec{z} \rangle } 
e^{- \frac{1}{4M s_n } z_{n+1}^2} ( 1_{\{n \in A\}} + 
1_{\{n \in A^c\}} e^{\frac{1}{2M} \frac{z_n z_{n+1}}{s_n}} ) \\
& =
\begin{cases}
e^{- \frac{1}{4M} \langle \mathcal{H}_A \vec{z}, \vec{z} \rangle } 
e^{- \frac{1}{4M s_n } z_{n+1}^2} & n \in A \\
e^{- \frac{1}{4M} \langle \mathcal{H}_A (\vec{z}-q), (\vec{z}-q) \rangle } 
e^{- \frac{1}{4M s_n } z_{n+1}^2}
e^{\frac{\langle \mathcal{H}_A q, q \rangle}{4M} }
& n \in A^c, 
\end{cases}
\end{split}
\end{equation*}
where  $ \vec{z} = (z_1, \cdots, z_n ) $, 
$ \mathcal{H}_A = (h_{ij}) $ is a symmetric matrix given by
\begin{equation*}
h_{ij} = h_{ji} = 
\begin{cases}
s_1^{-1} & i=j=1 \\
s_{i-1}^{-1} + s_i^{-1} & i=j \geq 2  \\
- s_i^{-1} & i \in A^c \setminus \{n\}, j= i+1 \\ 
%\& j \in I^c, i= j-1 
0 & otherwise,
\end{cases}
\end{equation*}
and 
\begin{equation}\label{defq}
q = \mathcal{H}_A^{-1} \,\,{}^t (0, \cdots, 0, z_{n+1}/s_n ). 
\end{equation}
Hence, we have
\begin{equation*}
\begin{split}
\mathcal{I}_A = 2^{-|A|} \prod_{i=1}^n s_i^{-\frac{1}{2}} 
(\det \mathcal{H}_A )^{-\frac{1}{2}}
e^{- \frac{1}{4M s_n } z_{n+1}^2}
( 1_{\{n \in A\}} + 1_{\{n \in A^c\}} 
e^{\frac{\langle \mathcal{H}_A q, q \rangle}{4M} } )
\end{split}
\end{equation*}

The proof will be complete if we show 
\begin{equation}\label{HIqq}
\begin{split}
\langle \mathcal{H}_A q, q \rangle \leq \frac{z_{n+1}^2}{s_n} 
\end{split}
\end{equation}
and 
\begin{equation}\label{detbound}
\begin{split}
\det \mathcal{H}_A \geq 
\prod_{i \in A \setminus \{ n \}} s_i^{-2}
(s_i + s_{i+1}) 
\prod_{j \in A^c \cup \{ n \}} s_j^{-1}.
\end{split}
\end{equation}

We first show (\ref{detbound}). For a fixed $ A$,
we choose $ i_1, \cdots, i_l $ and $ j_1, \cdots, j_l $ 
in the following way:
\begin{Algorithm}\label{algo}
\begin{enumerate}
\item 
$ i_1 = 1 $ if $ 1 \in A $, otherwise $ j_1 = 1 $.
\item If $ i_1 =1 $, then define inductively for 
$ k \geq 1 $,  
$ j_k =  \min \{ j : j \in A^c,  i_k < j < n \} $ 
and $ i_{k+1} = \min \{ i : i \in A, j_k < i < n  \} $. 
\item If $ j_1 = 1 $, then $ i_k =  \min \{ i : i \in A, j_k < i < n  \} $ 
and $ j_{k+1} = \min \{ j : j \in A, i_k < j < n  \} $.
\end{enumerate}
\end{Algorithm}
We stop this algorithm 
of the set to take minimum becomes 
empty set. 
%note also that $ i_{k+1} \leq j_k +2 $ 
%if $ i_1=1 $ and $ i_{k} \leq j_k +2 $ if $ j_1=1 $. 

Let 
\begin{equation*}
\mathcal{H}_{I_k} := 
\begin{cases}
(h_{ij})_{ \{ i_k \leq i,j \leq j_k -1 \} } & i_1=1, k =1 \\
(h_{ij})_{ \{ i_k +1 \leq i,j \leq j_k -1 \} } & i_1=1, k \geq 2 \\
 (h_{ij})_{ \{ i_k+1 \leq i,j \leq j_{k+1} -1 \} } & j_1 =1, k \geq 2 
\end{cases}
\end{equation*}
and 
\begin{equation*}
\mathcal{H}_{J_k} :=
\begin{cases}
(h_{ij})_{ \{ j_k \leq i,j \leq i_{k+1} \} } & i_1=1 \\
(h_{ij})_{ \{ j_k \leq i,j \leq i_{k}  \} } & j_1 =1.
\end{cases}
\end{equation*}
Then, 
\begin{equation*}
\mathcal{H}_A 
= \begin{pmatrix}
\mathcal{H}_{I_1} & & &  \\
& \mathcal{H}_{J_1} & &  \\
& & \ddots & 
\end{pmatrix}
\end{equation*}
if $ i_1 = 1 $, and so on.
The point is that 
in any case this gives a direct sum decomposition. 
Since $ \mathcal{H}_{I_k}  $, $ k=1, 2, \cdots $ are 
diagonal
%(though we need to distinguish four cases; $ 1 \in I $ or not and $ n \in I $ or not), 
we have that
\begin{equation}\label{detp}
\det \mathcal{H}_A = \prod_k \det \mathcal{H}_{I_k} 
\prod_{k'} \det \mathcal{H}_{J_{k'}} . 
\end{equation}

One can easily confirm that
\begin{equation}\label{detI}
\det \mathcal{H}_{I_k}  = 
\begin{cases}
s_1^{-1} & i_1 =1, j_1 =2, k=1 \\
s_1^{-1} \prod_{i=i_k+1}^{j_k -1} s_i^{-1} s_{i-1}^{-1} (s_i + s_{i-1}) & \\ 
\qquad = s_{1}^{-2} \cdots s_{j_1-2}^{-2} s_{j_1-1}^{-1} (s_1+s_2) \cdots (s_{j_k -1} +
s_{j_k -2}) 
& i_1=1, j_1 > 2, k= 1 \\
\prod_{i=i_k+1}^{j_k -1} s_i^{-1} s_{i-1}^{-1} (s_i + s_{i-1})  & i_1=1, k \geq 2 \\
\prod_{i=i_k+1}^{j_{k+1} -1}  s_i^{-1} s_{i-1}^{-1} (s_i + s_{i-1})  & j_1 =1.
\end{cases}
\end{equation}
One can also prove that 
\begin{equation}\label{detJ}
\det \mathcal{H}_{J_k}  =
\begin{cases} 
\prod_{j=j_k}^{i_k} s_j^{-1} = s_1^{-1} 
\cdots s_{i_1}^{-1}
& j_1 = 1, k=1 \\
\prod_{j=j_k-1}^{i_k} s_j^{-1} \sum_{j=j_k-1}^{i_k} s_j & j_1 =1, k \geq 2 \\
\prod_{j=j_k-1}^{i_{k+1}} s_j^{-1} \sum_{j=j_k-1}^{i_{k+1}}s_j 
& i_1 =1. 
\end{cases}
\end{equation}
Here we only prove it for $ J_1 $ for $ j_1 > 1$, by induction 
with respect to $ i_2 $. Note that $ i_2 \geq 3 $.  
Since
\begin{equation*}
\det \mathcal{H}_{J_1} |_{i_2 =n+1} 
= (s_n^{-1} + s_{n-1}^{-1}) \det \mathcal{H}_{J_1} |_{i_2 ={n}} 
- s_{n-1}^2 \det \mathcal{H}_{J_1} |_{i_2 ={n-1}},
\end{equation*}
by inductive assumption, 
\begin{equation*}
\begin{split}
\det \mathcal{H}_{J_1} |_{i_2 =n+1} 
&= (s_n^{-1} + s_{n-1}^{-1}) \prod_{j=j_1}^{n-1} s_j^{-1} \sum_{j=j_1}^{n -1} s_j
- s_{n-1}^{-2} \prod_{j=j_1}^{n-2} s_j^{-1} \sum_{j=j_1}^{n -2} s_j \\
&=\prod_{j=j_1}^{n} s_j^{-1} \sum_{j=j_1}^{n -1} s_j
+ s_{n-1}^{-1} \prod_{j=j_1}^{n-1} s_j^{-1} 
( \sum_{j=j_1}^{n -1} s_j -\sum_{j=j_1}^{n -2} s_j ) \\
& = \prod_{j=j_1}^{n} s_j^{-1} \sum_{j=j_1}^{n -1} s_j + \prod_{j=j_1}^{n-1} s_j^{-1} 
= \prod_{j=j_1}^{n} s_j^{-1} (\sum_{j=j_1}^{n -1} s_j + s_n ), 
\end{split}
\end{equation*}
as desired. Other cases can be treated in the same way. 

By noting 
\begin{equation*}
\sum_{j=j_k-1}^{i_k} s_j \geq s_{i_k} + s_{i_k-1},
\end{equation*}
and 
\begin{equation*}
\sum_{j=j_l-1}^{n} s_j \geq s_n, 
\end{equation*}
we see that (\ref{detp}), (\ref{detI}) and (\ref{detJ}) prove
(\ref{detbound}).

Now we turn to a validation of (\ref{HIqq}). 
By the definition (\ref{defq}) of $ q  $, 
we notice that
\begin{equation*}
\langle \mathcal{H}_A q, q  \rangle 
= \frac{z_{n+1}^2}{s_n^2} (\mathcal{H}_A^{-1})_{nn}
= \frac{z_{N+1}^2}{s_N^2} \frac{\det \widetilde{(\mathcal{H}_A)}_{nn}}
{\det \mathcal{H}_A },
\end{equation*}
where $ \widetilde{(\mathcal{H}_A)}_{nn} $ is the  $ nn$-th
cofactor matrix of $ \mathcal{H}_A $. 
By the above observations on the determinants of 
$ \mathcal{H}_{I_k} $ and  $ \mathcal{H}_{J_k} $, 
we now see that 
\begin{equation*}
\frac{\det \widetilde{(\mathcal{H}_A)}_{nn}}
{\det \mathcal{H}_A} = 
\frac{\det \widetilde{(\mathcal{H}_{J_l})}_{nn}}
{\det \mathcal{H}_{J_l} },
\end{equation*}
where $ J_l $ is such that $ n \in J_l $. 
Then by (\ref{detJ}), 
\begin{equation*}
\frac{\det \widetilde{(\mathcal{H}_{J_l})}_{nn}}
{\det \mathcal{H}_{J_l}} = s_n \frac{s_{j_l}+\cdots +s_{n-1}}{s_{j_l}+\cdots +s_n}
\leq  s_n.
\end{equation*}
Hence we have (\ref{HIqq}). 
\qed
\end{proof}

\begin{Lemma}\label{DcDc}
For $ n \geq 2 $, and $ y_{n+1} \in D^c $,  
\begin{equation*}\label{estA5}
\begin{split}
&\int_{D^{n}}
\prod_{i=1}^{n}|h( s_{i}, y_{i+1}, y_{i})|dy_{i} \\
&\leq  \left( C_1 s_n^{-\frac{1}{2}}
+ 2Md C_1 C_2 s_n^{-1}\right)
\sum_{A \subset \{ 1, \cdots, n-1\}} 
(\frac{\delta C_2}{2})^{|A|} C_1^{|A^c|}
\prod_{ j \in A }
(s_j + s_{j+1})^{-\frac{1}{2}} \prod_{i=1}^{n-1} s_{i}^{-\frac{1}{2}}. 
\end{split}
\end{equation*}
\end{Lemma}

\begin{proof}
Let
\begin{equation*}
%\begin{split}
 g (y)
: =C_1 + \frac{\delta C_2}{2} s_{n-1}^{-\frac{1}{2}} e^{-\frac{(\langle y, \gamma \rangle -k)^2}{4M s_{n-1}}}.
%\end{split}
\end{equation*}
Then, by Lemma \ref{spacon},
we have that 
\begin{equation*}\label{estA3c}
\begin{split}
& \int_{D^{n-1}}
\prod_{i=1}^{n-1}| h( s_{i}, y_{i+1}, y_{i})| dy_{i} \\
& \leq s_{n-1}^{-\frac{1}{2}} 
\left( C_1  + \frac{\delta C_2}{2} s_{n-1}^{-\frac{1}{2}} 
e^{-\frac{(\langle y_{n}, \gamma \rangle -k)^2}{4M s_{n-1}}} \right)
\sum_{A \subset \{ 1, \cdots, n-2\} } 
(\frac{\delta C_2}{2})^{|A|} C_1^{|A^c|}
\prod_{ j \in A}
(s_j+s_{j+1})^{-\frac{1}{2}} \prod_{i=1}^{n-2} s_{i}^{-\frac{1}{2}} \\
&= s_{n-1}^{-\frac{1}{2}} 
g(y_n) \sum_{A \subset \{ 1, \cdots, n-2\} } 
(\frac{\delta C_2}{2})^{|A|} C_1^{|A^c|}
\prod_{ j \in A }
(s_j+s_{j+1})^{-\frac{1}{2}} \prod_{i=1}^{n-2} s_{i}^{-\frac{1}{2}}. 
\end{split}
\end{equation*} 
we need to show that 
\begin{equation*}
\begin{split}
& \int_D \left( |
h_0 (s_{n}, y_{n+1}, y_{n})| 
+|h_0 (s_{n}, y_{n+1}, \theta(y_{n}))
| \right)
|g(y_{n})| dy_{n} \\
& = \int_D |h_0 (s_{n}, y_{n+1}, y_{n})| |g(y_{n})| dy_{n} \\
& \hspace{2cm} +  \int_{D^c} |h_0 (s_{n}, y_{n+1}, y_{n})| |g(\theta(y_{n}))| dy_{n}
\end{split}
\end{equation*}
is dominated by
\begin{equation*}
 (C_1 s_{n}^{-\frac{1}{2}}
+ 2Md C_2 s_{n}^{-1})
\left( C_1 + \frac{\delta C_2}{2} (s_{n-1} + s_{n})^{-\frac{1}{2}}
\right). 
\end{equation*}

By Lemma \ref{hestimate}, we know that
\begin{equation*}
\begin{split}
&\int_D |h_0 (s_{n}, y_{n+1}, y_{n})| |g(y_{n})| dy_{n} \\
& \leq C_1^2 s_{n}^{-\frac{1}{2}}
\int_D p_{s_{n}}^{2M} (y_{n+1}, y_{n})
d y_{n} 
%\\& 
+\frac{\delta C_1 C_2}{2}s_{n}^{-\frac{1}{2}}
\int_D p_{s_{n}}^{2M} (y_{n+1}, y_{n})
s_{n-1}^{-\frac{1}{2}} e^{-\frac{(\langle y_{n}, \gamma \rangle -k)^2}{4M s_{n-1}}} \,dy_{n} =:IV_D 
\end{split}
\end{equation*}
and
\begin{equation*}
\begin{split}
&\int_{D^c} |h_0 (s_{n}, y_{n+1}, y_{n})| |g(\theta(y_{n}))| dy_{n} \\
& \leq C_1^2 s_{n}^{-\frac{1}{2}} 
\int_{D^c} p_{s_{n}}^{2M} (y_{n+1}, y_{n})
d y_{n} 
+\frac{\delta C_1 C_2}{2}s_{n}^{-\frac{1}{2}}
\int_{D^c} p_{s_{n}}^{2M} (y_{n+1}, y_{n})
s_{n-1}^{-\frac{1}{2}} e^{-\frac{(\langle \theta(y_{n}), \gamma \rangle -k)^2}{4M s_n}} \,dy_{n} \\
&\qquad  + (\delta 1_{\{y_{n+1} \in D \} } + 2M d 1_{\{ y_{n+1} \in D^c\} })\\
& \qquad \times 
\bigg( C_1 C_2 s_{n}^{-1 }
\int_{D^c} p_{s_{n}}^{2M} (y_{n+1},y_{n})
d y_{n}  + \frac{\delta C_2^2}{2} {s_{n}^{-1 }}
\int_{D^c} p_{s_{n}}^{2M} (y_{n+1},y_{n})
s_{n-1}^{-\frac{1}{2}} e^{-\frac{(\langle \theta(y_{n}), \gamma \rangle -k)^2}{4M s_n}} \,dy_{n} \bigg)
\\& =: IV_{D^c} + V. 
\end{split}
\end{equation*}

Since 
\begin{equation*}
\begin{split}
&\int_{D \,\text{(resp. $D^c$})} p^{2M}_t (x,y) e^{-\frac{(\langle y, \gamma \rangle -k)^2}{4M s}} \,dy
\\&
= (4 \pi Mt)^{-d/2}\int_{D \,\text{(resp. $D^c$})}  e^{-\frac{1}{4Mt}
\left(
(\langle x, \gamma \rangle -\langle y, \gamma \rangle)^2 
+
\sum_{i=1}^{d-1} (\langle x, \gamma_i \rangle -\langle y, \gamma_i \rangle)^2 \right)
}e^{-\frac{(\langle y, \gamma \rangle -k)^2}{4M s}} \,dy \\
&=  (4\pi Mt)^{-1/2} \int_{(k, \infty)\,\text{(resp. $(-\infty,k]$})}
e^{-\frac{1}{4Mt} (\langle x,\gamma \rangle - y)^2 }e^{-\frac{(y-k)^2}{4M s}} dy
\\&
\leq s^{1/2} (t+s)^{-1/2} e^{-\frac{1}{4M(t+s)} (\langle x,\gamma \rangle - k)^2 }, 
\end{split}
\end{equation*}
and
\begin{equation*}
\begin{split}
&\int_{D \,\text{(resp. $D^c$})}  p^{2M}_t (x,y) \,dy
= (4 \pi Mt)^{-d/2}\int_{D \,\text{(resp. $D^c$})}  
e^{-\frac{1}{4Mt}
\left(\langle x, \gamma \rangle -\langle y, \gamma \rangle)^2 
+
\sum_{i=1}^{d-1} (\langle x, \gamma_i \rangle -\langle y, \gamma_i \rangle)^2 \right)
}\,dy \\
&=  (4\pi Mt)^{-1/2} \int_{(k, \infty)\,\text{(resp. $(-\infty,k]$})} e^{-\frac{1}{4Mt} (\langle x,\gamma \rangle - y)^2 } dy \\
& \leq 1_{\{ x \in D \, \text{(resp. $D^c$})\}} 
+ \frac{1}{2}1_{\{ x \in D^c \,\text{(resp. $D$}) \}}
e^{-\frac{1}{4Mt} (\langle x,\gamma \rangle - k)^2 },
\end{split}
\end{equation*}
where $ \gamma_i $, $ i=1, \cdots, d-1 $ are as in the proof of Lemma \ref{intgso}, 
we have that 
\begin{equation*}
\begin{split}
IV_D + IV_{D^c}
& \leq C_1^2 s_{n}^{-\frac{1}{2}}
+\frac{\delta C_1 C_2}{2} s_{n}^{-\frac{1}{2}}(s_{n-1} +s_n)^{-1/2}
e^{-\frac{(\langle y_{n+1}, \gamma \rangle -k)^2}{4M (s_{n-1} + s_n) }} \\
& \leq  C_1^2 s_{n}^{-\frac{1}{2}}
+\frac{\delta C_1 C_2}{2} s_{n}^{-1}
e^{-\frac{(\langle y_{n+1}, \gamma \rangle -k)^2}{4M s_{n}}} \\
&\leq C_1^2 s_{n}^{-\frac{1}{2}}
+\frac{\delta C_1 C_2}{2} s_{n}^{-1}
e^{-\frac{(\langle y_{n+1}, \gamma \rangle -k)^2}{4M s_{n}}} 1_{\{ y_{n+1} \in D \}} + \frac{\delta C_1 C_2}{2} s_{n}^{-1} 1_{\{y_{n+1} \in D^c \} }
\end{split}
\end{equation*}
and
\begin{equation*}
\begin{split}
V &\leq (\delta 1_{\{y_{n+1} \in D \} } 
+ 2M d 1_{\{ y_{n+1} \in D^c\} })
C_2  s_{n}^{-1 } \\
& \times 
\bigg( C_1 \left( 1_{\{ y_{n+1} \in D^c \}} + 1_{\{ y_{n+1} \in D \}}
e^{-\frac{1}{4Ms_{n+1}} (\langle y_{n+1},\gamma \rangle - k)^2 } \right)
 + \frac{\delta C_2}{2} 
 (s_{n-1} + s_n)^{-1/2}
e^{-\frac{(\langle y_{n+1}, \gamma \rangle -k)^2}{4M(s_{n-1}+ s_n)}} \bigg). 
\end{split}
\end{equation*}

We need to work only on the case $ y_{n+1} \in D^c $, where 
we now obtain
\begin{equation*}
\begin{split}
 &IV_D +IV_{D^c} + V \\
&\leq 
C_1^2 s_{n}^{-\frac{1}{2}}
+ \frac{\delta C_2 C_1 }{2}s_{n}^{-1} + 2M d
C_1 C_2  s_{n}^{-1 }
+ M d
\delta C_2^2  s_{n}^{-1 }
(s_{n-1} + s_n)^{-1/2}
e^{-\frac{(\langle y_{n+1}, \gamma \rangle -k)^2}{4M(s_{n-1}+ s_n)}} \\
&\leq  (C_1 s_{n}^{-\frac{1}{2}}
+ 2Md C_2 s_{n}^{-1})
\left( C_1 + \frac{\delta C_2}{2} (s_{n-1} + s_{n})^{-\frac{1}{2}} \right),
\end{split}
\end{equation*}
as desired. 
\qed
\end{proof}

\subsection{Proof of (i) of Theorem \ref{nhint}}
\begin{proof}
Take $ s_i = u_i - u_{i-1} $ for 
$ i=1, \cdots, n $, both in Lemma \ref{spacon} and Lemma \ref{DcDc}. \qed
\end{proof}
%%%%%%%%%%%%%%%%%%%%%%%%%%%%%%%%%%%%%%%%%%%%%%%%%%%%%%%%%%
%%%%%%%%%%%%%%%%%%%%%%%%%%%%%%%%%%%%%%%%%%%%%%%%%%%%%%%%%%
%%%%%%%%%%%%%%%%%%%%%%%%%%%%%%%%%%%%%%%%%%%%%%%%%%%%%%%%%%
\subsection{Proof of (ii) of Theorem \ref{nhint}}
\subsubsection{Estimates of integral with respect to time variables}
In this section, we give an estimate 
for an integral with respect to time variables 
of the bound given in Lemma \ref{spacon}. 
%%%%%%%%%%%%%%%%%%%%%%%%%%%%%%%%%%%%%%%%%%%%%%%%%%%%%%%%%%
\begin{Lemma}\label{timcon}
Let $ n \geq 3 $.
There exists a constant $ C_8 $ independent of $ n $ such that, for any $ A \subset \{ 1, 2, \cdots, n-1 \} $, 
%and $ u_{n} = t $, 
\begin{equation}\label{Gtest}
\begin{split}
& \int_0^{u_{n-1}}\cdots \int_0^{
u_{2}}
%(\frac{\delta C_2}{2})^{|A|} C_1^{|A^c|}
\prod_{ j \in A }
(u_{j+1}-u_{j-1})^{-\frac{1}{2}} \prod_{i=1}^{n-1}(u_{i}-u_{i-1})^{-\frac{1}{2}} d u_1 \cdots du_{n-2} \\
& \leq  
\Gamma(\frac{1}{4})^2
\frac{C_8^{|A|}
\Gamma(\frac{1}{4})^{|A^c|} 
}{\Gamma (\frac{|A^c|}{4})}
\left(u_{n-1}^{\frac{- 2 + |A^c|}{2}} 1_{\{ n-1 \not\in A \}} + u_{n-1}^{-\frac{3}{4} +
\frac{|A^c|}{2}}
(u_n-u_{n-1})^{-\frac{1}{4}} 1_{\{ n-1 \in A \}} \right),
\end{split}
\end{equation}
where $ u_0 =0 $.  
\end{Lemma}
%%%%%%%%%%%%%%%%%%%%%%%%%%%%%%%%%%%%%%%%%%%%%%%%%%%%%%%%%%
%%%%%%%%%%%%%%%%%%%%%%%%%%%%%%%%%%%%%%%%%%%%%%%%%%%%%%%%%%
\begin{proof}
%In the proof we use  
Define operators 
$ \mathcal{T}_i $, $ i=1, 2 $ 
for 
functions %, for some $ p > 1 $, 
on the simplex $ \{ (t,s) \in [0, T]^2: s \leq t \}$ by 
\begin{equation*}
\mathcal{T}_0 (f) (s,t) 
= \int_0^{s} (s-u)^{-\frac{1}{2}} f(u,s) du,
\end{equation*}
(constant in $ t $) and 
\begin{equation*}
\begin{split}
\mathcal{T}_{1} (f) (s,t) 
= \int_0^{s} (t-u)^{-\frac{1}{2}} (s-u)^{-\frac{1}{2}} 
f(u,s) \, du,
\end{split}
\end{equation*}
if they exist.
We also define $ \tau_1, \cdots,  \tau_{n-1}
$ $ : 2^{\{1, \cdots, n-1 \}} \to
\{ 0,1 \} $ by
\begin{equation*}
\tau_k (A) = 
\begin{cases}
0 & k \in A^c \\
1 & k \in A,
\end{cases}
\end{equation*}
for $ k=1, 2, \cdots, n-1 $.

Let
\begin{equation*}\label{inif}
f_A (s,t) 
= 
(s^{-\frac{1}{2}} \tau_1 (A) 
+ (1-\tau_1 (A)))
\mathcal{T}_{\tau_{2}(A)} (g) (s,t) 
\end{equation*}
where $ g(u,s) = u^{-1/2} $.
Then, the integral of the left-hand-side of \eqref{Gtest} is written as
\begin{equation*}
\begin{split}
%& \int_0^{u_{n-1}}\cdots \int_0^{u_{2}}
%(\frac{\delta C_2}{2})^{|A|} C_1^{|A^c|}
%\prod_{ j \in A }(u_{j+1}-u_{j-1})^{-\frac{1}{2}} \prod_{i=1}^{n-1}(u_{i}-u_{i-1})^{-\frac{1}{2}} d u_1 \cdots du_{n-2} \\
%& = 
\mathcal{T}_{\tau_{n-1}(A)} \circ \cdots 
\circ
\mathcal{T}_{\tau_{3}(A)} (f_A) (u_{n-1}, u_n).
\end{split}
\end{equation*} 

We note that, 
for
\begin{equation}\label{betaker}
f(u,t) = (t-u)^{-\epsilon} u^{-1+\beta}, 
\end{equation}
where $ \beta > 0 $, $ \epsilon \in [0, \frac{1}{4}]$, 
\begin{equation*}\label{1conv}
\begin{split}
\mathcal{T}_0 (f) (s,t) 
&= \int_0^s (s-u)^{-\frac{1}{2}-\epsilon}
u^{-1 + \beta} \,du \\
&= s^{-\frac{1}{2}-\epsilon-1+ \beta}
\int_0^s \left(1-\frac{u}{s}
\right)^{-\frac{1}{2}-\epsilon}
\left( \frac{u}{s}\right)^{-1 + \beta} \,du \\
&= s^{-1 + \frac{1}{2}-\epsilon+\beta}  
B \left(\frac{1}{2}-\epsilon, \beta \right). 
\end{split}
\end{equation*}
Inductively, for $ m \geq 2 $, 
we have
\begin{equation}\label{t1-m}
\mathcal{T}_0^m (f) (s,t) 
= s^{ -1 + \frac{m}{2}-\epsilon + \beta}  
B \left(\frac{1}{2}-\epsilon, \beta \right)
\prod_{k=2}^m B \left(\frac{1}{2}, \frac{k-1}{2} 
-\epsilon + \beta \right). 
\end{equation}

We claim that, for the same $ f $ as (\ref{betaker}), 
there exists a constant $ C_8 $ 
%depending on $ \epsilon $ 
such that 
\begin{equation}\label{2conv}
\begin{split}
\mathcal{T}_1 (f) (s,t) 
\leq  C_8 (t-s)^{-\frac{1}{4}} 
s^{-\frac{3}{4} + \beta - \epsilon}.
\end{split}
\end{equation}
Note that by a repeated use of \eqref{2conv} 
we still have, for any $ m $, 
\begin{equation}\label{t2-m}
    \begin{split}
    \mathcal{T}_1^m (f) (s,t) 
    & \leq  C_8^m 
s^{-\frac{3}{4} + \beta 
- \epsilon}
(t-s)^{-\frac{1}{4}}.
    \end{split}
\end{equation}

In fact, by the generalized binomial formula, 
\begin{equation*}
    \begin{split}
    \mathcal{T}_1 (f) (s,t) 
&= \int_0^s (t-u)^{-\frac{1}{2}}
(s-u)^{-\frac{1}{2}-\epsilon}
u^{-1 + \beta} \,du \\
&= t^{-\frac{1}{2}}\int_0^s (1-\frac{u}{t})^{-\frac{1}{2}}
(s-u)^{-\frac{1}{2}-\epsilon}
u^{-1 + \beta} \,du \\
&= t^{-\frac{1}{2} } \int_0^s \sum_{k=0}^\infty 
\begin{pmatrix} -\frac{1}{2} \\
k \end{pmatrix} (-1)^k 
\left( \frac{u}{t} \right)^k 
(s-u)^{-\frac{1}{2}-\epsilon} u^{-1 + \beta} \,du,
    \end{split}
\end{equation*}
%%%%
and since the infinite series inside the integral is 
absolutely convergent on $ s < t $, 
\begin{equation*}
    \begin{split}
    \mathcal{T}_1 (f) (s,t) &= t^{-\frac{1}{2}} \sum_{k=0}^\infty \begin{pmatrix} -\frac{1}{2} \\
k \end{pmatrix}(-t)^{-k} \int_0^s (s-u)^{-\frac{1}{2}-\epsilon}
    u^{-1 + \beta+k}\,du \\
    & = t^{-\frac{1}{2}} s^{-\frac{1}{2}-\epsilon + \beta } \sum_{k=0}^\infty \begin{pmatrix} -\frac{1}{2} \\
k \end{pmatrix}(-t)^{-k}s^k 
    B\left(\frac{1}{2} - \epsilon, 
    \beta + k
    \right).
    \end{split}
\end{equation*}
Then, by Lemma \ref{yuri} below, we have 
\begin{equation*}
    \begin{split}
    \mathcal{T}_1 (f) (s,t) 
    & \leq  C_{4}t^{-\frac{1}{2}}  s^{-\frac{1}{2}-\epsilon + \beta }
   \sum_{k=0}^\infty \begin{pmatrix} - \frac{1}{4} \\
k \end{pmatrix} (-t)^{-k} s^k  \\
&= \ C_8 t^{-\frac{1}{2}}  s^{-\frac{1}{2}-\epsilon + \beta }
\left(1- \frac{s}{t}\right)^{-\frac{1}{4}} \\
&=  C_8 t^{-\frac{1}{4}}  s^{-\frac{1}{2}-\epsilon + \beta }
(t-s)^{-\frac{1}{4}} \\
&\leq C_8 s^{-\frac{3}{4} + \beta 
 - \epsilon}
(t-s)^{-\frac{1}{4}},
    \end{split}
\end{equation*}
as desired. 
Here we can take $C_4$ which is greater than $B(\frac{1}{4}, \frac{1}{4})$.  

We shall inductively apply 
\eqref{t1-m} and \eqref{t2-m} to obtain \eqref{Gtest}.
Firstly, we have
\begin{equation*}
\begin{split}
f_A (s,t) &\leq (s^{-\frac{1}{2}} \tau_1 (A) 
+ (1-\tau_1 (A))) 
\mathcal{T}_{\tau_2 (A)} (g) (s,t) \\
&\leq (1-\tau_1 (A)))(1- \tau_2(A)) B\left(\frac{1}{2}, \frac{1}{2} \right)
+ (1-\tau_1 (A))\tau_2(A) C_8 s^{-\frac{1}{4}} (t-s)^{-\frac{1}{4}} \\
& +\tau_1 (A) (1- \tau_2(A))
s^{-\frac{1}{2}}B\left(\frac{1}{2}, \frac{1}{2} \right)
+ \tau_1 (A)\tau_2(A)
C_8  s^{-\frac{3}{4}} (t-s)^{-\frac{1}{4}}.
\end{split}
\end{equation*}
Observe that, for $ m_1 \geq 0 $,
$ m_2, \cdots, m_{n} >0 $ and 
$ l_1, \cdots, l_{n} > 0 $,
with the convention that $ l_0 = 0 $,
\begin{equation}\label{betab+}
\begin{split}
& \mathcal{T}_0^{l_n} \circ \mathcal{T}_1^{m_n} \circ\cdots\circ \mathcal{T}_0^{l_1} \circ \mathcal{T}_1^{m_1} (f_A)(s,t)
\\
& \leq (1-\tau_1 (A))(1- \tau_2(A))B\left(\frac{1}{2}, \frac{1}{2} \right)
C_8^{m_1 + \cdots +m_n} 
s^{\frac{l_1 + \cdots +l_n}{2}}
\\
& \hspace{1cm} \times \prod_{i=1}^n B\left( \frac{1}{4}(1+1_{\{i=1,\ m_1=0\}}), 
\frac{2(l_0 + \cdots + l_{i-1} )+5-1_{\{i=1,\ m_1=0\}}}{4}
\right)\\
& \hspace{1cm}  \times
\prod_{k=2}^{l_i} B \left(\frac{1}{2}, \frac{k+(l_0 + \cdots + l_{i-1} )+1}{2} \right)
\\ & 
+ (1-\tau_1 (A))\tau_2(A) C_8^{1+m_1+\cdots + m_{n+1} }
s^{\frac{-1+l_1 +\cdots+ l_n}{2}}
\\& \hspace{1cm} \times 
\prod_{i=1}^n B\left( \frac{1}{4}, 
\frac{2(l_0 + \cdots + l_{i-1} ) +3}{4}
\right) %\\ & \hspace{1cm} \times 
\prod_{k=2}^{l_i} B \left(\frac{1}{2}, \frac{k+(l_0 + \cdots + l_{i-1} )}{2} \right) \\
& +\tau_1 (A) (1- \tau_2(A))
B\left(\frac{1}{2}, \frac{1}{2} \right)
C_8^{m_1 +\cdots + m_{n+1}} s^{\frac{-1+l_1 \cdots + l_n}{2}}
\\ & \hspace{1cm} \times 
\prod_{i=1}^n B\left( \frac{1}{4}
(1+1_{\{i=1,\ m_1=0\}}), 
\frac{2(l_0 + \cdots + l_{i-1} ) +3-1_{\{i=1,\ m_1=0\}}}{4}
\right)
\\ & \hspace{1cm} \times 
\prod_{k=2}^{l_i} B \left(\frac{1}{2}, \frac{k+(l_0 + \cdots + l_{i-1} )}{2} \right) \\
& + \tau_1 (A)\tau_2(A) C_8^{1+m_1 + \cdots + m_{n+1}}
s^{\frac{-2+l_1 + \cdots + l_n}{2}}\\ 
& \hspace{1cm} \times 
\prod_{i=1}^n B\left( \frac{1}{4}, 
\frac{2(l_0 + \cdots + l_{i-1} ) +1}{4}
\right)
\prod_{k=2}^{l_i} B \left(\frac{1}{2}, \frac{k+(l_0 + \cdots + l_{i-1} )-1}{2} \right).
\end{split}
\end{equation}
Here we understand $ \mathcal{T}^0_1 $
is the identity map. 
Since beta functions are decreasing in 
each variable, 
we obtain that 
\begin{equation*}
\begin{split}
& \mbox{(the right-hand-side of \eqref{betab+})}
\\& \leq 
 (1-\tau_1 (A))(1- \tau_2(A))B\left(\frac{1}{2}, \frac{1}{2} \right)
C_8^{m_1 + \cdots +m_n} 
s^{\frac{l_1 + \cdots +l_n}{2}}
\\
& \hspace{1cm} \times \prod_{i=1}^n B\left( \frac{1}{4}, 
\frac{(l_0 + \cdots + l_{i-1} )+2}{4}
\right)
%\\& \hspace{1cm}  \times
\prod_{k=2}^{l_i} B \left(\frac{1}{4}, \frac{k+(l_0 + \cdots + l_{i-1} )+1}{4} \right)
\\ & 
+ (1-\tau_1 (A))\tau_2(A) C_8^{1+m_1+\cdots + m_{n+1} }
s^{\frac{-1+l_1 +\cdots+ l_n}{2}}
\\& \hspace{1cm} \times 
\prod_{i=1}^n B\left( \frac{1}{4}, 
\frac{(l_0 + \cdots + l_{i-1} ) +1}{4}
\right)
%\\ & \hspace{1cm} \times 
\prod_{k=2}^{l_i} B \left(\frac{1}{4}, \frac{k+(l_0 + \cdots + l_{i-1} )}{4} \right) \\
\end{split}
\end{equation*}
\begin{equation*}
\begin{split}
& +\tau_1 (A) (1- \tau_2(A))
B\left(\frac{1}{2}, \frac{1}{2} \right)
C_8^{m_1 +\cdots + m_{n+1}} s^{\frac{-1+l_1 \cdots + l_n}{2}}
\\ & \hspace{1cm} \times 
\prod_{i=1}^n B\left( \frac{1}{4}
, \frac{(l_0 + \cdots + l_{i-1} ) +1}{4}
\right)
%\\ & \hspace{1cm} \times 
\prod_{k=2}^{l_i} B \left(\frac{1}{4}, \frac{k+(l_0 + \cdots + l_{i-1} )}{4} \right) \\
& + \tau_1 (A)\tau_2(A) C_8^{1+m_1 + \cdots + m_{n+1}}
s^{\frac{-2+l_1 + \cdots + l_n}{2}}\\ 
& \hspace{1cm} \times 
\prod_{i=1}^n B\left( \frac{1}{4}, 
\frac{(l_0 + \cdots + l_{i-1} ) +1_{\{i=1\}}}{4}
\right)
\prod_{k=2}^{l_i} B \left(\frac{1}{4}, \frac{k+(l_0 + \cdots + l_{i-1} )-1}{4} \right),
\end{split}
\end{equation*}
%%%%%%%%%%%%%%%%%%%%%%%%%
\begin{equation}\label{betab++}
\begin{split}
\\&=
(1-\tau_1 (A))(1- \tau_2(A))
C_8^{m_1 + \cdots +m_n} 
s^{\frac{l_1 + \cdots +l_n}{2}}
\Gamma (\frac{1}{2})^3
\Gamma (\frac{1}{4})^{l_1+\cdots+l_n}
\Gamma(\frac{l_1 + \cdots + l_{i-1} +2}{4})^{-1}
\\ & 
+ (1-\tau_1 (A))\tau_2(A) C_8^{1+m_1+\cdots + m_{n+1} }
s^{\frac{-1+l_1 +\cdots+ l_n}{2}}
\Gamma (\frac{1}{4})^{l_1+\cdots+l_n+1}
\Gamma(\frac{l_1 + \cdots + l_{i-1}+1}{4})^{-1}
\\& +
\tau_1 (A) (1- \tau_2(A))
C_8^{m_1 +\cdots + m_{n+1}} s^{\frac{-1+l_1 \cdots + l_n}{2}}
\Gamma (\frac{1}{2})^2
\Gamma (\frac{1}{4})^{l_1+\cdots+l_n+1}
\Gamma(\frac{l_1 + \cdots + l_{i-1} +1}{4})^{-1}
\\& 
+ \tau_1 (A)\tau_2(A) C_8^{1+m_1 + \cdots + m_{n+1}}
s^{\frac{-2+l_1 + \cdots + l_n}{2}} 
\Gamma (\frac{1}{4})^{l_1+\cdots+l_n+2}
\Gamma(\frac{l_1 + \cdots + l_{i-1} }{4})^{-1},
\end{split}
\end{equation}
%%%%%%%%%%%%%%%%%%%%%%%%%%%%%%%%%%%%%%%%%%%%%%%%%%%%%%%%%%%%%%
which is consistent with the assertion
of the lemma. 

The case where $ l_n = 0 $ 
can be obtained by 
simply applying \eqref{t2-m}
to \eqref{betab++}.
\qed
\end{proof}
\begin{Lemma}\label{yuri}
Suppose that $ \epsilon \in [0, \frac{1}{4}] $. 
Then there exists a constant $C_9$
such that for any $k \in \mathbf{N}$ and $\beta >0$,
\begin{equation*}\label{eq-lemma}
 \begin{pmatrix} -\frac{1}{2} \\
k \end{pmatrix}
    B\left(\frac{1}{2} - \epsilon, 
    \beta + k
    \right)
\leq C_{4}
\begin{pmatrix} - \frac{1}{4} \\
k \end{pmatrix}.
\end{equation*}
\end{Lemma}
\begin{proof}
First let us observe that 
we have that 
\begin{equation}\label{binomial1}
\begin{split}
 \begin{pmatrix} -\frac{1}{2} \\
k \end{pmatrix}
 \begin{pmatrix} -\frac{1}{4} \\
k \end{pmatrix}^{-1}
& = 
\left((-1)^k \frac{\Gamma (\frac{1}{2}+k)}{\Gamma (k+1 ) \Gamma (\frac{1}{2})}\right)
\left((-1)^k \frac{\Gamma (\frac{1}{4}+k)}{\Gamma (k+1 ) \Gamma (\frac{1}{4})}\right)^{-1}
\\&
= 
\frac{\Gamma (\frac{1}{2}+k) \Gamma (\frac{1}{4})}
{\Gamma (\frac{1}{4}+k) \Gamma (\frac{1}{2})} \\
& \sim 
\frac{\Gamma (\frac{1}{4})}
{ \Gamma (\frac{1}{2})}
\sqrt{\frac{\frac{1}{4}+k}{\frac{1}{2}+k}}
e^{\frac{1}{4}-\frac{1}{2}}
(\frac{1}{2}+k)^{\frac{1}{2}+k}
(\frac{1}{4}+k)^{-\frac{1}{4}-k}
\\& =
\frac{\Gamma (\frac{1}{4})}
{ \Gamma (\frac{1}{2})}
e^{-\frac{1}{4}}
(\frac{1}{2}+k)^{k}
(\frac{1}{4}+k)^{\frac{1}{4}-k}.
\end{split}
\end{equation}
by Stirling's approximation %\eqref{stirling},  that is 
\begin{equation*}\label{stirling}
\Gamma (z) \sim \sqrt{\frac{2 \pi}{z}}\left(\frac{z}{e}\right)^z
\qquad (z \rightarrow \infty ). 
\end{equation*}

On the other hand,  
\begin{equation}\label{binomial2}
\begin{split}
&    B\left(\frac{1}{2} - \epsilon, 
    \beta + k
    \right) < 
B\left(\frac{1}{4}, 
    %\beta + 
    k \right)= 
    \frac{\Gamma (\frac{1}{4} )\Gamma (k)}
    {\Gamma (\frac{1}{4} + k)} \\
\\& \sim \Gamma (\frac{1}{4} )
\sqrt{\frac{\frac{1}{4} + k}{k}}
e^{-k+\frac{1}{4} + k}
(\frac{1}{4} + k)^{-\frac{1}{4} - k}
k^{k}
\\& =\Gamma (\frac{1}{4})
e^{\frac{1}{4}}
(\frac{1}{4}+ k)
^{ \frac{1}{4} - k}
k^{k-\frac{1}{2}}.
\end{split}
\end{equation}
By \eqref{binomial1} and \eqref{binomial2}, we have that 
\begin{equation}\label{binomial11}
\begin{split}
& \begin{pmatrix} -\frac{1}{2} \\
k \end{pmatrix}
 \begin{pmatrix} -\frac{1}{4} \\
k \end{pmatrix}^{-1}
B\left(\frac{1}{2} - \epsilon , 
    \beta + k
    \right)
< \begin{pmatrix} -\frac{1}{2} \\
k \end{pmatrix}
 \begin{pmatrix} -\frac{1}{4} \\
k \end{pmatrix}^{-1}
B\left(\frac{1}{4}, 
    %\beta + 
    k \right)
\\& \sim
\frac{\Gamma (\frac{1}{4})^2}
{ \Gamma (\frac{1}{2})}
(\frac{1}{2}+k)^{k}
(\frac{1}{4}+k)^{\frac{1}{4}-k}
(\frac{1}{4}+ k)
^{\frac{1}{4}-k} k^{ k-\frac{1}{2}} 
%& \leq \frac{\Gamma (\frac{1}{4})^2 }
%{ \Gamma (\frac{1}{2})}e^{\frac{1}{4}}
%\left(\frac{\frac{1}{2}+k}{\frac{1}{4}+k}\right)^{k}
%\left(\frac{\beta +k}{\frac{1}{4} + \beta +k}\right)^{\beta + k}
%(\frac{1}{4}+k)^{\frac{1}{4}}(\frac{1}{2}  + \beta + k)^{ \frac{1}{4}}(\beta + k)^{-\frac{1}{2}} \\
=: g(k).
\end{split}
\end{equation}
Since 
$\lim_{%\beta, 
k \rightarrow \infty} g(%\beta, 
k) < \infty$, 
the leftmost of \eqref{binomial11} 
is bounded by a constant $C_9$ 
which is independent of $ \epsilon $, 
$ \beta $, and $ k $. %for $\beta$ and $k$. 
\qed
\end{proof}
%%%%%%%%%%%%%%%%%%%%%%%%%%%%%%%%%%%%%%%%%%%%%%%%%%%%%%%%%%
\begin{Lemma}\label{yuri2}
For $\xi \in (0,1)$, 
there exists a constant $C_{10}$ such that 
for any $ x >0 $, 
\begin{equation*}\label{eq-lemma2}
\frac{1}{\Gamma (x)} \leq C_{10}\xi^{x} .
\end{equation*}
\end{Lemma}
%%%%%%%%%%%%%%%%%%%%%%%%%%%%%%%%%%%%%%%%%%%%%%%%%%%%%%%%%%
\begin{proof}
By Stirling's approximation \eqref{stirling}, there 
exists a constant $C_{10}'$ such that 
$$\Gamma (x) \geq C_{10}' \sqrt{\frac{2 \pi}{x}}\left(\frac{x}{e}\right)^x. $$
Therefore we see that 
\begin{equation*}\label{eq-88}
\begin{split}
\log ( \Gamma (x) \xi^{x} )
& \geq  
\log C_{10}' +
%\log \sqrt{\frac{2 \pi}{x}}\left(\frac{x}{e}\right)^x\xi^{x} 
%\\&= C_5^L
\frac{1}{2} \log 2 \pi + 
(x-\frac{1}{2}) \log x
- x + x  \log \xi,
\end{split}
\end{equation*}
%By \eqref{eq-88}, we obtain that 
and so
\begin{equation}\label{eq-89}
\begin{split}
\lim_{x \rightarrow 0 }
\log (\Gamma (x) \xi^{x} )
& \geq  (\text{constant}) 
- \lim_{x \rightarrow 0 } \log x = \infty, 
\end{split}
\end{equation}
and 
\begin{equation}\label{eq-90}
\begin{split}
\lim_{x \rightarrow \infty }
\log ( \Gamma (x) \xi^x )
& \geq (\text{constant}) +
\lim_{x \rightarrow \infty }
( x(\log x + \log \xi )
- \frac{3}{2}x ) = \infty. 
\end{split}
\end{equation}
By \eqref{eq-89} and \eqref{eq-90}, we have the assertion. 
\qed
\end{proof}
%%%%%%%%%%%%%%%%%%%%%%%%%%%%%
%%%%%%%%%%%%%%%%%%%%%%%%%%%%%
\subsubsection{Proof of (ii) of Theorem \ref{nhint}}
%%%%%%%%%%%%%%%%%%%%%%%%%%
\begin{proof}
By Theorem \ref{nhint} (i) and Lemma \ref{timcon}, we have that for $ y_{n+1} \in D $, $ u_n \in (0,T)$ 
and $f \in L^\infty (D) $, 
\begin{equation}\label{binomial}
\begin{split}
&\int_{0=u_0 < u_1 <\cdots < u_n }
du_1 \cdots du_{n-1}
\int_{D^n} |\prod_{i=1}^{n} h(u_{i} - u_{i-1}, y_{i+1}, y_i) 
f(y_1)| dy_1 \cdots dy_n  
%%%%%%%%%%%%%%%
\\ & \leq 
||f ||_{\infty}\sum_{A \subset \{ 1, \cdots, n-1\}} 
(\frac{\delta C_2}{2})^{|A|}
C_1^{|A^c|}
\int_0^{u_n}du_{n-1}
\bigg( C_1 (u_n- u_{n-1})^{-\frac{1}{2}}
+ 
\frac{\delta C_2}{2}(u_n- u_{n-1})^{-1}
e^{-\frac{(\langle y_{n+1}, \gamma \rangle -k)^2}{4M (u_n-u_{n-1})}}\bigg) 
\\
& \qquad  \times 
\int_0^{u_{n-1}}\cdots \int_0^{
u_{2}}
\prod_{i=1}^{n-1}(u_{i}-u_{i-1})^{-\frac{1}{2}}
\prod_{ j \in A }
(u_{j+1}-u_{j-1})^{-\frac{1}{2}}
d u_1 \cdots du_{n-2}
\\& \leq 
||f ||_{\infty}
\Gamma(\frac{1}{4})^2
\sum_{A \subset \{ 1, \cdots, n-1\}} 
\frac{(\frac{\delta }{2}C_2C_8)^{|A|}
(C_1\Gamma(\frac{1}{4}))^{|A^c|}}{\Gamma (\frac{|A^c|}{4})}
\\
& \qquad  \times 
\int_0^{u_n}
\bigg( C_1 (u_n- u_{n-1})^{-\frac{1}{2}}
+ 
\frac{\delta C_2}{2}(u_n- u_{n-1})^{-1}
e^{-\frac{(\langle y_{n+1}, \gamma \rangle -k)^2}{4M (u_n-u_{n-1})}}\bigg) 
\\ & \qquad  \times 
\left(u_{n-1}^{\frac{- 2 + |A^c|}{2}} 1_{\{ n-1 \not\in A \}} + u_{n-1}^{-\frac{3}{4} +
\frac{|A^c|}{2}}
(u_n-u_{n-1})^{-\frac{1}{4}} 1_{\{ n-1 \in A \}} \right)
du_{n-1}.
\end{split}
\end{equation}
Since we have
\begin{equation*}
\begin{split}
(u_n- u_{n-1})^{-1}
e^{-\frac{(\langle y_{n+1}, \gamma \rangle -k)^2}{4M (u_n-u_{n-1})}}
\leq (u_n- u_{n-1})^{-\frac{5}{8}}
(4M)^{\frac{3}{8}}(\langle y_{n+1}, \gamma \rangle -k)^{-\frac{3}{4}}
K_{\frac{3}{8}},
\end{split}
\end{equation*}
where $ K_{3/8} $ is given as \eqref{Kbeta},
the integral of the rightmost of 
\eqref{binomial} is dominated by 
\begin{equation*}
\begin{split}
& C_1
\left(
B(\frac{1}{2}, \frac{|A^c|}{2} )1_{\{ n-1 \not\in A \}} 
+ B(\frac{1}{4}, \frac{1}{4} +\frac{|A^c|}{2})
1_{\{ n-1 \in A \}} 
\right) u_n^{-\frac{1}{2}+ \frac{|A^c|}{2}}
\\ & + 
\frac{\delta C_2}{2}
(4M)^{\frac{3}{8}}(\langle y_{n+1}, \gamma \rangle -k)^{-\frac{3}{4}}
K_{\frac{3}{8}}
%u_n^{-\frac{1}{8}}
\left(B(\frac{3}{8}, \frac{|A^c|}{2} )
1_{\{ n-1 \not\in A \}} 
+ B(\frac{1}{8}, \frac{1}{4} +\frac{|A^c|}{2})
1_{\{ n-1 \in A \}} \right)
u_n^{-\frac{5}{8}+ \frac{|A^c|}{2}}.
\end{split}
\end{equation*}
Since beta functions 
are decreasing in both variables, 
this bound is replaced with
\begin{equation*}
\begin{split}
\left(
 C_1 u_n^{-\frac{1}{2}}
+\frac{\delta C_2}{2}
(4M)^{\frac{3}{8}}(\langle y_{n+1}, \gamma \rangle -k)^{-\frac{3}{4}}
K_{\frac{3}{8}}
u_n^{-\frac{5}{8}} 
\right) B(\frac{1}{8}, \frac{|A^c|}{4})
u_n^{\frac{|A^c|}{2}}.
\end{split}
\end{equation*}

Notice that
%to the right-hand-side of 
\begin{equation*}
\begin{split}
\frac{(\frac{\delta }{2}C_2C_8)^{|A|}
(C_1\Gamma(\frac{1}{4})
u_n^{\frac{1}{2}})^{|A^c|}}{\Gamma (\frac{|A^c|}{4})}
B(\frac{1}{8}, \frac{|A^c|}{4})
&= \Gamma(\frac{1}{8}) 
\frac{(\frac{\delta }{2}C_2C_8)^{|A|}
(C_1\Gamma(\frac{1}{4})
u_n^{\frac{1}{2}})^{|A^c|}}{\Gamma (\frac{2|A^c|+1}{8})} \\
&= \Gamma(\frac{1}{8}) C_{10} \xi^{\frac{1}{8}}
(\frac{\delta }{2}C_2C_8)^{|A|}
(C_1\Gamma(\frac{1}{4})
u_n^{\frac{1}{2}}\xi^{\frac{1}{4}})^{|A^c|}
\end{split}
\end{equation*}
for arbitrary $ \xi > 0 $, 
where we applied 
Lemma \ref{yuri2}
with the constant $ C_{10} $.

Observe that 
\begin{equation*}
\begin{split}
& \sum_{A \subset \{1,\cdots, n-1\}}
(\frac{\delta }{2}C_2C_8)^{|A|}
(C_1\Gamma(\frac{1}{4})
u_n^{\frac{1}{2}}\xi^{\frac{1}{4}})^{|A^c|}\\
&= \sum_{k=0}^{n-1} \sum_{|A|=k}
(\frac{\delta }{2}C_2C_8)^{|A|}
(C_1\Gamma(\frac{1}{4})
u_n^{\frac{1}{2}}\xi^{\frac{1}{4}})^{|A^c|}
= \sum_{k=0}^{n-1} \sum_{|A|=k}
(\frac{\delta }{2}C_2C_8)^{k}
(C_1\Gamma(\frac{1}{4})
u_n^{\frac{1}{2}}\xi^{\frac{1}{4}})^{n-1-k} \\
&= \sum_{k=0}^{n-1} 
\begin{pmatrix}
n-1 \\
k
\end{pmatrix}
(\frac{\delta }{2}C_2C_8)^{k}
(C_1\Gamma(\frac{1}{4})
u_n^{\frac{1}{2}}\xi^{\frac{1}{4}})^{n-1-k} 
= \left(\frac{\delta }{2}C_2C_8 + C_1\Gamma(\frac{1}{4})
u_n^{\frac{1}{2}}\xi^{\frac{1}{4}}
\right)^{n-1}.
\end{split}
\end{equation*}
Then, by taking $ \xi = \delta^4 $, 
\begin{equation*}
C_6 := \frac{C_2C_8 }{2}
+C_1\Gamma(\frac{1}{4}) T^{\frac{1}{2}},
\end{equation*}
\begin{equation*}
C_7 := \max \left( \delta^{\frac{1}{2}} \Gamma(\frac{1}{4})^2
\Gamma(\frac{1}{8}) C_{10}, 
\frac{\delta^{\frac{3}{2}} C_2}{2} (4M)^{\frac{3}{8}}
K_{\frac{3}{8}}\Gamma(\frac{1}{4})^2
\Gamma(\frac{1}{8}) C_{10} 
\right), 
\end{equation*}
we obtain \eqref{secondeq}. 
\qed
\end{proof}
%%%%%%%%%%%%%%%%%%%%%%%%%%%%%%%%%%%%%%%%%%%%%%

%%%%%%%%%%%%%%%%%%%%%%%%%%%%%%%%%%%%%%%%%%%%%%%%%%%%%%%%%%%%%%%%
%%%%%%%%%%%%%%%%%%%%%%%%%%%%%%%%%%%%%%%%%%%%%%%%%%%%%%%%%%%%%%%%
%\subsection{A Proof of Corollary \ref{cor42}}\label{proof-cor42}

\subsection{Lemmas for Theorem \ref{errorrep}}
\subsubsection{An Estimate for the integral of $ (q h) $ times $S^*$ }
\begin{Lemma}\label{cor42}
For any $x \in \partial D $ and $n \geq 2$,
it holds that 
\begin{equation*}\label{Hintegrability2}
\begin{split}
& \int_D \left| \int_{\mathbf{R}^d} q_s (x,y) h(t,y, z) dy \right|
|S^{*n}_u f (z)| \,dz \\
& 
\leq ||f ||_{\infty} ( C_6 \delta)^{n-1} C_5C_7(4M_0)^{\frac{d}{2}} 
 s^{-\frac{1}{2}} t^{-\frac{1}{2}} \left(1 
+  (4M_0 (s+t))^{-\frac{3}{8}} \Gamma (\frac{1}{8})
\right).
\end{split}
\end{equation*}
\end{Lemma}

\begin{proof}
By (iii) of Theorem \ref{intbh} and
(ii) of Theorem \ref{nhint}, 
we have that 
\begin{equation*}
\begin{split}
& \int_D \left| \int_{\mathbf{R}^d} q_{s}
(x, y) h (t, y, z) \,dy \,\right| \left| S^{*n}_{T-s} f (z)\right| dz  \\
& \leq \int_D 2 C_5 s^{-\frac{1}{2}} t^{-\frac{1}{2}}
(s+t)^{-\frac{d}{2}} e^{ -\frac{ |x-z|^2}{4M_0 (s+t)}}
\\&\qquad\times 
||f ||_{\infty} ( C_6 \delta)^{n-1}C_7 
\left(
 (T-s)^{-\frac{1}{2}}
+
(\langle y_{n+1}, \gamma \rangle -k)^{-\frac{3}{4}}
(T-s)^{-\frac{5}{8}} 
\right)\, dz.
\end{split}
\end{equation*}
The assertion follows from 
\begin{equation*}
\begin{split}
& \int_D 
e^{ -\frac{ |x-z|^2}{4M_0 (s+t)}} dz \\
&=  \int_{\mathbf{R^{d-1}}}
e^{ -\frac{ \sum_{i=2}^d ( \langle x, \gamma_i 
\rangle - \langle z, \gamma_i \rangle)^2 }{4M_0 (s+t)}}
\,dz 
\int_k^\infty 
e^{-\frac{(k-w)^2}{4M_0 (s+t)}}\,dw \\
&= (4M_0 (s+t))^{\frac{d-1}{2}} 
 2^{-1} (4M_0 (s+t))^{\frac{1}{2}} \\
& = 2^{-1} (4M_0 (s+t))^{\frac{d}{2}} 
%(T-s)^{-\frac{5}{8}} 
\end{split}
\end{equation*}
and
\begin{equation*}
\begin{split}
& \int_D 
e^{ -\frac{ |x-z|^2}{4M_0 (s+t)}} 
(\langle z, \gamma \rangle -k)^{-\frac{3}{4}} \,dz \\
&=  \int_{\mathbf{R^{d-1}}}
e^{ -\frac{ \sum_{i=2}^d ( \langle x, \gamma_i 
\rangle - \langle z, \gamma_i \rangle)^2 }{4M_0 (s+t)}}
\,dz 
\int_k^\infty (w -k)^{-\frac{3}{4}}
e^{-\frac{(k-w)^2}{4M_0 (s+t)}}\,dw \\
&= (4M_0 (s+t))^{\frac{d-1}{2}} 
 2^{-1} (4M_0 (s+t))^{\frac{1}{8}} \Gamma (\frac{1}{8}) 
 \\&
 =  2^{-1}(4M_0 (s+t))^{\frac{d}{2}-\frac{3}{8}}\Gamma (\frac{1}{8}). \qed
%(T-s)^{-\frac{5}{8}} 
\end{split}
\end{equation*}
\end{proof}

\subsubsection{An Estimate for
the integral of $ h S^*$}
\begin{Lemma}\label{q-hS}
There exists a constant $ C_{11}$ dependent on $\delta$ such that 
for any $n \geq 2$, 
\begin{equation}\label{Hintegrability2-2}
\begin{split}
& \int_D h(t,y, z) 
S^{*n}_u f (z) \,dz 
\\& \leq   C_{11}
||f ||_{\infty} ( C_6 \delta)^{n-1} 
\left(
 u^{-\frac{1}{2}}
(t^{-\frac{1}{2}}+ t^{-1}e^{-\frac{(\langle y, \gamma \rangle -k)^2}{4Mt}})
%\right. \\ & \qquad \left.
+ u^{-\frac{5}{8}} 
( 
t^{-\frac{7}{8}}
+ 
(t^{-\frac{7}{8}} + t^{-\frac{11}{8}}) 
e^{-\frac{ (\langle y, \gamma \rangle-k)^2}{4Mt}}
)
\right),
\end{split}
\end{equation}
for any $y\in D$. 
\end{Lemma}

\begin{proof}
It suffices to estimate
\begin{equation*}
\int_D\left| h(t,y, z) \right|(\langle z, \gamma \rangle -k)^{-\frac{3}{4}} \,dz 
\end{equation*}
since we have by Theorem \ref{nhint} (ii),  
\begin{equation}\label{bdd-IIbe}
\begin{split}
&\int_D\left| h(t,y, z) \right|
\left|S^{*n}_u f (z)\right| \,dz 
\\& \leq ||f ||_{\infty} ( C_6 \delta)^{k-1} C_{7}
\left(
u^{-\frac{1}{2}}\int_D\left| h(t,y, z) \right| \,dz 
+ u^{-\frac{5}{8}} 
\int_D\left| h(t,y, z) \right|(\langle z, \gamma \rangle -k)^{-\frac{3}{4}} \,dz 
\right). 
\end{split}
\end{equation}
By Lemma \ref{hestimate},
%with , we see that 
\begin{equation}\label{bdd-IIbe2}
\begin{split}
&\int_D\left| h(t,y, z) \right|(\langle z, \gamma \rangle -k)^{-\frac{3}{4}} \,dz 
\\& \leq C_3 \left(t^{-\frac{1}{2}}
\int_D p_{2t}^M (y,z)(\langle z, \gamma \rangle -k)^{-\frac{3}{4}} \,dz 
+ (t^{-\frac{1}{2}} + t^{-1}) 
\int_D p_{2t}^M (y,\theta(z))(\langle z, \gamma \rangle -k)^{-\frac{3}{4}} \,dz \right),
\end{split}
\end{equation}
where the constant $C_3$ is the one 
given in Theorem \ref{intbh} (i). 
%%%%%%%
%%%%%%
%By changing variance 
%$w =\langle y, \gamma \rangle$ and 
%$w_i = \langle y, \gamma_i \rangle \ (1 \leq i \leq d-1)$ 
%with   an orthonormal basis$ \{ \gamma , \gamma_1 , \cdots , \gamma_{d-1}\} $, 
The first integral in \eqref{bdd-IIbe2}
is estimated as follows:
\begin{equation}\label{bdd-II-iii}
\begin{split}
& \int_D p_{2t}^M (y,z)(\langle z, \gamma \rangle -k)^{-\frac{3}{4}} \,dz 
%%%%%
\\& =
\left(
\prod_{i=1}^{d-1}
\int_{\mathbf{R}} 
(4M\pi t)^{-\frac{1}{2}}
e^{-\frac{ (\langle y, \gamma_k \rangle-w_i)^2}{4Mt}}
\,dw_i
\right)
\int_{k}^{\infty} 
(4M\pi t)^{-\frac{1}{2}}
e^{-\frac{ (w-\langle y, \gamma \rangle)^2}{4Mt}}
(w -k)^{-\frac{3}{4}}
\,dw
%%%%%
\\& =
\int_{k}^{\infty} 
(4M\pi t)^{-\frac{1}{2}}
e^{-\frac{ (w-\langle y, \gamma \rangle)^2}{4Mt}}
(w -k)^{-\frac{3}{4}}
\,dw
\\& = 
\int_{k}^{k+(M\pi t)^{\frac{1}{2}}} 
(4M\pi t)^{-\frac{1}{2}}
e^{-\frac{ (w-\langle y, \gamma \rangle)^2}{4Mt}}
(w -k)^{-\frac{3}{4}}
\,dw
+
\int_{k+(M\pi t)^{\frac{1}{2}}}^{\infty} 
(4M\pi t)^{-\frac{1}{2}}
e^{-\frac{ (w-\langle y, \gamma \rangle)^2}{4Mt}}
(w -k)^{-\frac{3}{4}}
\,dw
\\&\leq  
(4M\pi t)^{-\frac{1}{2}}
\int_{k}^{k+(M\pi t)^{\frac{1}{2}}} 
(w -k)^{-\frac{3}{4}}
\,dw
+
t^{-\frac{3}{8}}
\int_{k+(M\pi t)^{\frac{1}{2}}}^{\infty} 
(4M\pi t)^{-\frac{1}{2}}
e^{-\frac{ (w-\langle y, \gamma \rangle)^2}{4Mt}}
\,dw
%%%%%%
\\& \leq
3(M \pi)^{- \frac{3}{8}} 
t^{-\frac{3}{8}}.
\end{split} 
\end{equation}
On the other hand, 
\begin{equation*}\label{bdd-II-iiii}
\begin{split}
& \int_D p_{2t}^M (y,\theta(z))(\langle z, \gamma \rangle -k)^{-\frac{3}{4}} \,dz \\
%\\& = \left(\prod_{i=1}^{d-1}\int_{\mathbf{R}} 
%(4M\pi t)^{-\frac{1}{2}} e^{-\frac{ (\langle y, \gamma_k \rangle+w_i-2k)^2}{4Mt}} \,dw_i \right)
%= \int_{k}^{\infty} (4M\pi t)^{-\frac{1}{2}}
%e^{-\frac{ (2K -w-\langle y, \gamma \rangle)^2}{4Mt}}
%(w -k)^{-\frac{3}{4}} \,dw
& =
\int_{k}^{\infty} 
(4M\pi t)^{-\frac{1}{2}}
e^{-\frac{ (2K -w-\langle y, \gamma \rangle)^2}{4Mt}}
(w -k)^{-\frac{3}{4}}
\,dw
\\& =
(4M\pi t)^{-\frac{1}{2}}
e^{-\frac{ (\langle y, \gamma \rangle-k)^2}{4Mt}}
\int_{k}^{\infty} 
e^{-\frac{ (w-k)^2 + 2 (w-k)(\langle y, \gamma \rangle-k)^2}{4Mt}}
(w -k)^{-\frac{3}{4}}
\,dw
\\& \leq 
(4 M\pi t)^{-\frac{1}{2}}
e^{-\frac{ (\langle y, \gamma \rangle-k)^2}{4Mt}}
\int_{k}^{\infty} 
e^{-\frac{ (w-k)^2}{4Mt}}
(w -k)^{-\frac{3}{4}}
\,dw,
\end{split} 
\end{equation*}
since $\langle y, \gamma \rangle-k$ is positive for $y \in D$.
Thus the second integral in \eqref{bdd-IIbe2} is dominated by
\begin{equation}\label{bdd-II-iiiii}
\begin{split}
& (4M\pi t)^{-\frac{1}{2}}
e^{-\frac{ (\langle y, \gamma \rangle-k)^2}{4Mt}}
\int_{k}^{\infty} 
e^{-\frac{ (w-k)^2}{4Mt}}
(w -k)^{-\frac{3}{4}}
\,dw
\\& \qquad = 
2^{-\frac{7}{4}}\pi^{-\frac{1}{2}} (Mt) ^{-\frac{3}{8}}
\Gamma (\frac{1}{8})
e^{-\frac{ (\langle y, \gamma \rangle-k)^2}{4Mt}}
.
\end{split} 
\end{equation}
Combining  
\eqref{bdd-IIbe2}, \eqref{bdd-II-iii}, and 
\eqref{bdd-II-iiiii} with Theorem \ref{intbh} (i), 
we obtain that
\begin{equation*}
\begin{split}
& \mbox{(the right-hand-side of \eqref{bdd-IIbe})}
\\& \leq 
C_3C_7||f ||_{\infty} ( C_6 \delta)^{k-1} 
\left(
u^{-\frac{1}{2}}
\left(t^{-\frac{1}{2}}+ t^{-1}e^{-\frac{(\langle y, \gamma \rangle -k)^2}{4Mt}}\right)
\right. 
\\ & \qquad \left. + 
u^{-\frac{5}{8}} 
\left( t^{-\frac{1}{2}}
\int_D p_{2t}^M (y,z)(\langle z, \gamma \rangle -k)^{-\frac{3}{4}} \,dz 
+ (t^{-\frac{1}{2}} + t^{-1}) 
\int_D p_{2t}^M (y,\theta(z))(\langle z, \gamma \rangle -k)^{-\frac{3}{4}} \,dz \right)\right)
\\&\leq  C_3C_7||f ||_{\infty}  ( C_6 \delta)^{k-1} 
\left( u^{-\frac{1}{2}}
\left(t^{-\frac{1}{2}}+ t^{-1}e^{-\frac{(\langle y, \gamma \rangle -k)^2}{4Mt}}\right)
\right. 
\\ & \qquad \left. + u^{-\frac{5}{8}} 
\left( 
3 (M \pi)^{- \frac{3}{8}}t^{-\frac{1}{2}}t^{-\frac{3}{8}}
+ (t^{-\frac{1}{2}} + t^{-1}) 
2^{-\frac{7}{4}}\pi^{-\frac{1}{2}} (Mt) ^{-\frac{3}{8}}
\Gamma (\frac{1}{8})
e^{-\frac{ (\langle y, \gamma \rangle-k)^2}{4Mt}}
\right)\right)
\\&\leq C_3C_7||f ||_{\infty} ( C_6 \delta)^{k-1} 
\left(
u^{-\frac{1}{2}}
\left(t^{-\frac{1}{2}}+ t^{-1}e^{-\frac{(\langle y, \gamma \rangle -k)^2}{4Mt}} \right)
\right. 
\\ & \qquad \left. + u^{-\frac{5}{8}} 
\left( 
3(M \pi)^{- \frac{3}{8}}
t^{-\frac{7}{8}}
+ 
2^{-\frac{7}{4}}\pi^{-\frac{1}{2}} M^{-\frac{3}{8}}
\Gamma (\frac{1}{8})
(t^{-\frac{7}{8}} + t^{-\frac{11}{8}}) 
e^{-\frac{ (\langle y, \gamma \rangle-k)^2}{4Mt}}
 \right)
\right).
\end{split}
\end{equation*}
Hence by taking 
\begin{equation*}
C_{11} %^{\delta} 
:= C_3  C_7\max (1,\  
3 M^{-\frac{3}{8}} \pi^{-\frac{3}{8}},\ 
2^{-\frac{7}{4}}\pi^{-\frac{1}{2}} M^{-\frac{3}{8}}
\Gamma (\frac{1}{8})), 
\end{equation*}
we have \eqref{Hintegrability2-2}.
\qed
\end{proof}

\subsubsection{An Estimate for the integral of $ qh $ times $ |h||S^*|$}\label{subsecqhhS}
\begin{Lemma}\label{qhhSlemma}
There exists a constant $ C_{12} $
such that, for $ x \in \partial D $, $ s,t,u, v \in [0,T] $, and $ n \geq 2 $, 
\begin{equation}\label{qhhS}
\begin{split}
& \int_D \left| \int_{\mathbf{R}^d}  q_{s} (x, y) 
h(t, y,z) dy  \right| \left( \int_D |h(u,z,w) S^{*n}_{v} f (w)| dw \right) \,dz \\
& \leq C_{12} ||f ||_{\infty} ( C_6 \delta)^{n-1}  s^{-\frac{1}{2}} t^{-\frac{1}{2}}
(s + t + v)^{-\frac{1}{2}}
u^{-\frac{7}{8}} v^{-\frac{5}{8}}.
%(u^{-\frac{1}{2}} v^{-\frac{1}{2}} + u^{-\frac{5}{8}} (v^{-\frac{7}{8}} + v^{-\frac{3}{8}}))
\end{split}
\end{equation}
\end{Lemma}

\begin{proof}
By applying (iii) of Theorem \ref{intbh}  together with 
Lemma \ref{q-hS}, 
we see that the integrand  
in the left-hand-side of 
\eqref{qhhS} is dominated by 
\begin{equation}\label{integrand6}
\begin{split}
&  2C_{11}
||f ||_{\infty} ( C_6 \delta)^{n-1}  C_5 s^{-\frac{1}{2}} t^{-\frac{1}{2}}
(s+t)^{-\frac{d}{2}} \exp\left( -\frac{ |x-z|^2}{4M_0 (s+t)} \right)  \\
& \times
\left(
 v^{-\frac{1}{2}}
(u^{-\frac{1}{2}}+ u^{-1}e^{-\frac{(\langle z, \gamma \rangle -k)^2}{4Mu}})
%\right. \\ & \qquad \left.
+ v^{-\frac{5}{8}} 
( 
u^{-\frac{7}{8}}
+ 
(u^{-\frac{7}{8}} + u^{-\frac{11}{8}}) 
e^{-\frac{ (\langle z, \gamma \rangle-k)^2}{4Mu}}
)
\right), \\
\end{split}
\end{equation}
since $|x- \theta (z) | = |x-z |$. 
Since we know that 
\begin{equation*}
\begin{split}
&\int_D e^{-\frac{ |x-z|^2}{4M_0 (s+t)}} e^{-\frac{ (\langle z, \gamma \rangle-k)^2}{4Mu}} \,dz \\
&= \int_D e^{- \sum_{i=1}^{d-1}
\frac{ \langle x-z, \gamma_i \rangle^2}{4M_0 (s+t)} }
e^{-\frac{ \langle x-z, \gamma \rangle^2}{4M_0 (s+t)} }
e^{-\frac{(k- \langle \gamma, z \rangle)^2}{4M u}} \,dz \\
&
\leq (4M_0 (s+t))^{\frac{d-1}{2}} 
\int_k^\infty e^{-\frac{ (k- z)^2}{4 M_0} ( (s+t)^{-1} + u^{-1} )} dz \\
& = 2^{-1} (4M_0 (s+t))^{\frac{d-1}{2}} 
(4M_0 (s+t)u)^{\frac{1}{2}}
(s+t+u)^{-\frac{1}{2}},
\end{split}
\end{equation*}
where $ \{ \gamma_i : i=1,\cdots d-1 \} $ is 
an orthonormal basis of $ (\partial D)^\bot $, 
we have that
the integral of \eqref{integrand6}
is dominated by
\begin{equation*}
\begin{split}
&  C_{11}
||f ||_{\infty} ( C_6 \delta)^{n-1}  C_5 s^{-\frac{1}{2}} t^{-\frac{1}{2}}
(4M_0)^{\frac{d}{2}} \\
& \times
\left(
 u^{-\frac{1}{2}}
v^{-\frac{1}{2}}(1 + (s+t+u)^{-\frac{1}{2}}) 
%\right. \\ & \qquad \left.
+ v^{-\frac{5}{8}} 
( 
u^{-\frac{7}{8}}
+ 
(s+t+u)^{-\frac{1}{2}}(u^{-\frac{3}{8}} + u^{-\frac{7}{8}})
)
\right) \\
&=  C_{11}
||f ||_{\infty} ( C_6 \delta)^{n-1}  C_5 s^{-\frac{1}{2}} t^{-\frac{1}{2}}
(4M_0)^{\frac{d}{2}} u^{-\frac{7}{8}}v^{-\frac{5}{8}} 
(s+t+u)^{-\frac{1}{2}} \\
& \times (u^{\frac{3}{8}}
v^{\frac{1}{8}} ( (s+t+u)^{\frac{1}{2}} + 1)
+ (s+t+u)^{\frac{1}{2}} 
+ (u^{\frac{1}{2}}+1)).
\end{split}
\end{equation*}
By taking 
\begin{equation*}
\begin{split}
C_7 :=
C_5 C_{11} (4M_0)^{\frac{d}{2}}
\max_{s,t,u, v \in [0,T]} 
 (u^{\frac{3}{8}}
v^{\frac{1}{8}} ( (s+t+u)^{\frac{1}{2}} + 1)
+ (s+t+u)^{\frac{1}{2}} 
+  (u^{\frac{1}{2}}+1)),
\end{split}
\end{equation*}
we obtain the results stated in \eqref{qhhS}.
\qed
\end{proof}
	
\subsection{Proof of Theorem \ref{errorrep}}\label{errorrep_Pr}
%\begin{proof}
The proof is conducted by considering each point stated in the Theorem, from (i) to (v).
	
{\textit{	Proof of (i)}}: this statement can be proven with a logic similar to the one underlying the proof of Lemma \ref{intgso}.
Notice that:
\begin{equation*}
	\begin{split}
	& E [ \pi^\bot S_{s-u} S^{*(n-1)}_{T-s} f(X_u) 
	| \mathcal{F}_{\tau}] \\
	&= 1_{ \{\tau \leq u \}} \int_{D^c} q_{u -\tau}
	(X_{\tau}, y) \{ S_{s-u} S^{*(n-1)}_{T-s} f (y) +S_{s-u} S^{*(n-1)}_{T-s} f (\theta(y) )
	\} \,dy \\
	& \hspace{2cm} 
	+ 1_{ \{ \tau > u \}} 1_{\{X_s \in D^c\}}
	(S_{s-u} S^{*(n-1)}_{T-s} f (X_u) + S_{s-u} S^{*(n-1)}_{T-s} f(\theta(X_u)) ); 
	\\
	\end{split}
\end{equation*}
then, since the second term is zero, as in the proof of Lemma \ref{intgso} we can write  
\begin{equation*}
	\begin{split}
	&| E [\pi^\bot S_{s-u} S^{*(n-1)}_{T-s}f (X_u) |\mathcal{F}_\tau] | \\
	& \leq 1_{ \{\tau < u \}}  \left| \int_{\mathbf{R}^d} q_{u -\tau}
	(X_{\tau}, y) \int_D h (s-u, y, z) S^{*(n-1)}_{T-s}f (z) \,dz\, dy \right|  \\
	&+ 1_{ \{\tau < u \}} \left| \int_{D}  
	q_{u -\tau}(X_{\tau}, y)
	\int_D h (s-u, y, z) S^{*(n-1)}_{T-s}f (z) \,dz\, dy \right| \\
	&+ 1_{ \{\tau < u \}} \left| \int_{D}  
	q_{u -\tau}
	(X_{\tau}, \theta (y)) 
	\int_D h (s-u, y, z) S^{*(n-1)}_{T-s}f (z) \,dz\, dy \right| \\
	& =: 1_{ \{\tau < u \}}(VI_1 + VI_2+VI_3). 
	\end{split}
\end{equation*}

Let us now identify a bound for $VI_1$. 
By Theorem \ref{intbh} (ii)
and Theorem \ref{nhint} (ii), 
we have that 
$q_{u -\tau}(X_{\tau}, y) h (s-u, y, z) S^{*(n-1)}_{T-s}f (z)$ is integrable 
in $ (y,z) \in \mathbf{R}^d \times D$, almost surely on $ \{\tau < u \}$. 
Therefore we can change the order of the integral of $VI_1$ to obtain 
\begin{equation*}%\label{ineq-I}
	\begin{split}
	VI_1 &\leq   \int_D \left| \int_{\mathbf{R}^d} q_{u -\tau}
	(X_{\tau}, y) h (s-u, y, z) \,dy \,\right| \left| S^{*(n-1)}_{T-s} f (z)\right| dz . \\
	\end{split}
\end{equation*}
Then, by Lemma \ref{cor42}, we have that $VI_1$ is bounded by
\begin{equation*}%\label{**}
	\begin{split}
	&||f ||_{\infty} ( C_6 \delta)^{n-2} C''(4M_0)^{\frac{d}{2}} 
	(u-\tau)^{-\frac{1}{2}} (s-u)^{-\frac{1}{2}}
	\max\{C_3',\ C_3'' (4M_0 )^{-\frac{3}{8}} \Gamma (\frac{1}{8})\}
	\left( 1
	+ (s-\tau)^{-\frac{3}{8}} 
	\right)\\
	& \hspace{4cm} =:VI_1'. 
	\end{split}
\end{equation*}
	
Let us now derive a bound for the terms indicated by $ VI_2 $ and $ VI_3 $. 
By considering \eqref{ineq-q} 
and Lemma \ref{q-hS}, we have:
\begin{equation*}
	\begin{split}
	VI_2&\leq  \int_{D}  q_{u -\tau}
	(X_{\tau}, y) 
	\left|\int_D h (s-u, y, z) 
	S^{*(n-1)}_{T-s}f (z)\,dz \right| \, dy 
	\\&\leq 
	C_{11} C_q 
	||f ||_{\infty} ( C_6 \delta)^{n-2}
	\int_{D}  
	(u-\tau)^{-\frac{d}{2}}
	e^{-\frac{ |X_{\tau}-y|^2}{4M_0(u-\tau)}}
	\\ & \qquad \times 
	\Big(
	(T-s)^{-\frac{1}{2}}
	((s-u)^{-\frac{1}{2}}+ (s-u)^{-1}e^{-\frac{(\langle y, \gamma \rangle -k)^2}{4M(s-u)}})
	\\ & \qquad 
	+ (T-s)^{-\frac{5}{8}} 
	( 
	(s-u)^{-\frac{7}{8}}
	+ 
	((s-u)^{-\frac{7}{8}} + (s-u)^{-\frac{11}{8}}) 
	e^{-\frac{ (\langle y, \gamma \rangle-k)^2}{4M(s-u)}}
	)\Big)
	\, dy \\
	& =  
	C_{11} C_q 
	||f ||_{\infty} ( C_6 \delta)^{n-2}
	(4M_0 \pi)^{\frac{d-1}{2}}
	\int_{k}^{\infty}  
	(u-\tau)^{-\frac{1}{2}}
	e^{-\frac{ (y-k)^2}{4M_0(u-\tau)}}
	\\ & \qquad \times 
	\Big(
	(T-s)^{-\frac{1}{2}}
	((s-u)^{-\frac{1}{2}}+ (s-u)^{-1}e^{-\frac{(y -k)^2}{4M(s-u)}})
	\\ & \qquad 
	+ (T-s)^{-\frac{5}{8}} 
	( 
	(s-u)^{-\frac{7}{8}}
	+ 
	((s-u)^{-\frac{7}{8}} + (s-u)^{-\frac{11}{8}}) 
	e^{-\frac{ (y-k)^2}{4M(s-u)}}
	)\Big)
	\, dy =:VI_2'.
	\end{split}
\end{equation*}
Note that the integrand in $ VI_1' $
is invariant if $ y $ is replaced with $ \theta(y) $, and therefore $ VI_1' $ dominates $ VI_3 $.  
Since $M < M_0$, we have 
\begin{equation*}%\label{e-MM0}
	\begin{split}
	\int_{k}^{\infty}  
	e^{-\frac{ (y-k)^2}{4M_0(u-\tau)}}
	e^{-\frac{ (y-k)^2}{4M(s-u)}}
	\, dy
	& \leq 
	\int_{k}^{\infty}  
	e^{-\frac{ (s- \tau )(y-k)^2}{4M(u-\tau)(s-u)}}
	\\& \leq (4 M \pi )^{\frac{1}{2}} (u-\tau )^{\frac{1}{2}}(s-u)^{\frac{1}{2}}(s- \tau)^{- \frac{1}{2}}. 
	\end{split}
\end{equation*}
Hence, $VI_2'$ is dominated by 
\begin{equation*}\label{bound-IIen}
	\begin{split}
	&
	C_{11} C_q 
	||f ||_{\infty} ( C_6 \delta)^{n-2}
	(4 M_0 \pi )^{\frac{d}{2}}
	\\ & \qquad \times 
	\Big( (T-s)^{-\frac{1}{2}}
	(s-u)^{-\frac{1}{2}}(1+ (s-\tau)^{-\frac{1}{2}}
	)
	\\ & \qquad 
	+ (T-s)^{-\frac{5}{8}} 
	( 
	(s-u)^{-\frac{7}{8}}
	+ 
	((s-u)^{-\frac{3}{8}} + (s-u)^{-\frac{7}{8}}) 
	(s-\tau)^{-\frac{1}{2}}
	)
	)\Big) =:VI_2''
	.
	\end{split}
	\end{equation*}

Notice that both bounds for $ VI_1' $ and $ VI_2'' $ are integrable in 
$ (u,s) \in \{ (s,u) : 0 \leq u \leq s \leq T \}
	\times \{ \tau < u \} $: since for $ VI_1'$, 
\begin{equation*}
	\begin{split}
	& I_{\{\tau < T\}}\int_\tau^T \int_\tau^s (u-\tau)^{-\frac{1}{2}} (s-u)^{-\frac{1}{2}}
	( 1 + (s-\tau)^{-\frac{3}{8}}) du \,ds \\
	&= I_{\{\tau < T\}}
	B \left( \frac{1}{2}, \frac{1}{2} \right) ( (T-\tau)
	+2 ( T-\tau)^{\frac{1}{2}} ) \in L^\infty (\Omega) 
	\end{split}
\end{equation*}
and since for $ VI_2'' $,  %we see that 
\begin{equation*}\label{inttime-II}
	\begin{split}
	%&\int_0^{T} \int_0^{s} 
	%1_{\{ \tau \leq u\}} (T-s)^{-\frac{1}{2}} (s-u)^{-\frac{1}{2}}(1+ (s-\tau)^{-\frac{1}{2}} )
	%\\ & \qquad + (T-s)^{-\frac{5}{8}} ( 
	%(s-u)^{-\frac{7}{8}} + ((s-u)^{-\frac{3}{8}} + (s-u)^{-\frac{7}{8}}) 
	%(s-\tau)^{-\frac{1}{2}})\, du \, ds \\
	& 1_{\{ \tau \leq T \}}
	\int_\tau^{T} \int_\tau^{s} 
	(T-s)^{-\frac{1}{2}}
	(s-u)^{-\frac{1}{2}}(1+ (s-\tau)^{-\frac{1}{2}}
	)
	\\ & \qquad 
	+ (T-s)^{-\frac{5}{8}} 
	( 
	(s-u)^{-\frac{7}{8}}
	+ 
	((s-u)^{-\frac{3}{8}} + (s-u)^{-\frac{7}{8}}) 
	(s-\tau)^{-\frac{1}{2}}
	)
	\, du \, ds
	\\&
	= 1_{\{ \tau \leq T \}} 
	\int_\tau^{T} 
	2(T-s)^{-\frac{1}{2}}
	(s-\tau)^{\frac{1}{2}}(1+ (s-\tau)^{-\frac{1}{2}}
	)
	\\ & \qquad 
	+ (T-s)^{-\frac{5}{8}} 
	( 
	8(s-\tau)^{\frac{1}{8}}
	+ 
	(\frac{8}{5}(s-\tau)^{\frac{5}{8}} + 8(s-\tau)^{\frac{1}{8}}) 
	(s-\tau)^{-\frac{1}{2}}
	)
	\, ds
	\\&
	= 1_{\{ \tau \leq T \}}
	\int_0^{T- \tau } 
	2(T- \tau -s)^{-\frac{1}{2}}
	(1+ s^{\frac{1}{2}}
	)
	\\ & \qquad 
	+ (T- \tau -s)^{-\frac{5}{8}} 
	( 
	(8+\frac{8}{5}) s^{\frac{1}{8}}
	+ 
	8s^{-\frac{3}{8}}) 
	\, ds
	\\&
	= 
	1_{\{ \tau \leq T \}}\left(
	4(T- \tau )^{\frac{1}{2}}+2 B (\frac{1}{2},\frac{3}{2})(T- \tau )
	+ 
	(8+\frac{8}{5})B(\frac{3}{8},\,\frac{9}{8})
	(T-\tau)^{\frac{1}{2}}
	+  8B(\frac{3}{8},\, \frac{5}{8}) \right)\in L^\infty (\Omega).
	\end{split}
\end{equation*}
This completes the proof of (i).\\
	%Therefore \eqref{inttime-II} is bounded, almost surely. 
	%Now we see that all of $ I $, $ II $, and $ III $ 
	%are jointly integrable since $ \int_0^T (I + II + III) ds $
	%is in $ L^\infty (\Omega) $.
	
{\textit{Proof of (ii) and (iii):}} these statements can be proved by recalling some techniques used for the proof of Lemma \ref{interr2}.\\
 
Let us first apply \eqref{GTMR3} 
(contained in the proof of Theorem \ref{RCSVF}) 
to the error of the hedging strategy with $ E [ \pi^\bot S_{s-u} S^{*(n-1)}_{T-s} f (X_u)| \mathcal{F}_{t \wedge \tau}] $:
\begin{equation*}
	\begin{split}
	& E [ 1_{\{\tau <u \}}\pi S_{s-u} S^{*(n-1)}_{T-s} f (X_u) 
	| \mathcal{F}_{t}]
	%\equiv e^{-r (T-s)} 
	%\mathrm{Err}_s (S_{T-s} f(X_s), %S_{T-s} f(\theta(X_s)); D) \\&
	= E [
	\int_0^u E [1_{\{\tau < v \}}
	S_{u-v} S_{s-u} S^{*(n-1)}_{T-s} f(X_v) |\mathcal{F}_\tau] \,dv
	| \mathcal{F}_t]. 
	\end{split}
\end{equation*}
Since we know that
\begin{equation*}
	q_{v-\tau} (X_\tau, y) 
	h(u-v, y,z) h (s-u,z,w)  S^{*(n-1)}_{T-s} f (w) 
\end{equation*}
is integrable in $ (y,z,w) \in \mathbf{R}^d \times D \times D$, 
we have on $ \{ \tau < v \} $, 
\begin{equation*}
	\begin{split}
	&\left|  E [%1_{\{\tau \leq s \}}
	S_{u-v} S_{s-u} S^{*(n-1)}_{T-s} f(X_v) |\mathcal{F}_\tau] \right| \\
	&= \left| \int_D \left( \int_{\mathbf{R}^d}  q_{v-\tau} (X_\tau, y) 
	h(u-v, y,z) dy  \right) S_{s-u} S^{*(n-1)}_{T-s} f (z) \,dz \right| \\
	& \leq %\Vert f \Vert_\infty
	\int_D \left| \int_{\mathbf{R}^d}  q_{v-\tau} (X_\tau, y) 
	h(u-v, y,z) dy  \right| \left( \int_D |h(s-u,z,w)|  |S^{*(n-1)}_{T-s} f (w)| dw \right) \,dz. \\
	\end{split}
\end{equation*}
By using Lemma \ref{qhhSlemma}, % in the Appendix \ref{subsecqhhS}, 
we can write
\begin{equation*}
	\begin{split}
	&\left|  E [%1_{\{\tau \leq s \}}
	S_{u-v} S_{s-u} S^{*(n-1)}_{T-s} f(X_v) |\mathcal{F}_\tau] \right| \\
	& \leq  C_{12} ||f ||_{\infty} ( C_6 \delta)^{n-2}  (v-\tau)^{-\frac{1}{2}} (u-v)^{-\frac{1}{2}}
	(T-s+u-\tau)^{-\frac{1}{2}}
	(s-u)^{-\frac{7}{8}} (T-s)^{-\frac{5}{8}}.
	\end{split}
\end{equation*}
Since on $ \{ \tau \leq T \} $ the following holds
\begin{equation*}
	\begin{split}
	& \int_\tau^T \int_\tau^s \int_\tau^u dvduds (v-\tau)^{-\frac{1}{2}} (u-v)^{-\frac{1}{2}}
	(T-s+u-\tau)^{-\frac{1}{2}}
	(s-u)^{-\frac{7}{8}} (T-s)^{-\frac{5}{8}} \\
	&= B \left( \frac{1}{2}, \frac{1}{2} \right) 
	\int_\tau^T ds \int_\tau^s du
	(T-s+u-\tau)^{-\frac{1}{2}}
	(s-u)^{-\frac{7}{8}} (T-s)^{-\frac{5}{8}} \\
	&\leq B \left( \frac{1}{2}, \frac{1}{2} \right)  
	\int_\tau^T ds \int_\tau^s du
	(u-\tau)^{-\frac{1}{2}}
	(s-u)^{-\frac{7}{8}} (T-s)^{-\frac{5}{8}} \\
	&= B \left( \frac{1}{2}, \frac{1}{2} \right) B\left( \frac{1}{2}, \frac{3}{8} \right) \int_\tau^T (T-s)^{-\frac{5}{8}} (s-\tau)^{-\frac{3}{
			8}} ds \\
	& = B \left( \frac{1}{2}, \frac{1}{2} \right) B\left( \frac{1}{2}, \frac{1}{8} \right) B\left( \frac{3}{8}, \frac{5}{8} \right), 
	\end{split}
\end{equation*}
we have that $ 1_{\{\tau < v \}} E [S_{u-v} S_{s-u} S^{*(n-1)}_{T-s} f(X_v) |\mathcal{F}_\tau] 	$ is integrable 
in $ (s,u, v,  \omega) $ on 
$ \{ (s,u,v) : 0 \leq v \leq u \leq s \leq T \} \times \Omega $ 
and the error is then obtained as
\begin{equation}\label{E2t+}
	\begin{split}
	%\int_0^T \mathrm{Err}^s_{2,t}  ds &= e^{-r(T-t)} 
	\int_0^T ds \int_0^s du \int_0^u dv\,
	E [ 1_{ \{ \tau \leq v \} } S_{u-v}  S_{s-u} S^{*(n-1)}_{T-s} f (X_v) 
	|\mathcal{F}_{t \wedge \tau} ]=: \mathrm{Err}_{n,t}. \\
	%&=:\mathrm{Err}_{2,t}(f(X_T), f (\theta(X_T);D). 
	\end{split}
\end{equation}
By changing the order of the integrals in (\ref{E2t+}), we get
\begin{equation*}
	\begin{split}
	\mathrm{Err}_{n,t}
	&= \int_0^T  dv \int_v^T du\,
	E [ 1_{ \{ \tau \leq v \} } S_{u-v} \left( \int_u^T S_{s-u} S^{*(n-1)}_{T-s} f (X_v) \,ds \right) 
	|\mathcal{F}_{t \wedge \tau} ] \\
	&= \int_0^T  dv \int_v^T du\,
	E [ 1_{ \{ \tau \leq v \} } S_{u-v}  S^{*n}_{T-u} f (X_v)  
	|\mathcal{F}_{t \wedge \tau} ],
	\end{split}
\end{equation*}
which proves (ii) and (iii).\\
	%%%%%%%%%%%%%%%%%%%%%%%%%%%%%%%%%%%%%%%%%%%%%%%%%%%%%%%%%%%%%%%%%%%%%%%%%%%%%%%%%%%%%%%%%%%proof of (iv)
	
{\textit{Proof of (iv):}} as already discussed for the integrability of the term $I$ in (i), by using Lemma \ref{cor42}, we obtain:
\begin{equation}\label{thm3-18-iv}
	\begin{split}
	&|\int_0^T \int_u^T \,
	E [ 1_{ \{ \tau \leq u \} } S_{s-u}
	S^{*n}_{T-s} f (X_u) |\mathcal{F}_{t \wedge \tau}]
	\,ds du | 
	\\& \leq 
	1_{\{\tau \leq t \wedge T \}}|\int_\tau ^T \int_\tau ^s \,
	E [ S_{s-u}
	S^{*n}_{T-s} f (X_u) |\mathcal{F}_{\tau}]
	\,du ds |
	\\& \qquad  + 
	1_{\{t\leq \tau \}}|\int_0^T \int_u^T \,
	E [ 1_{ \{ \tau \leq u \} } S_{s-u}
	S^{*n}_{T-s} f (X_u) |\mathcal{F}_{t \wedge \tau}]
	\,ds du |
	\\& \leq 
	1_{\{\tau \leq t\wedge T \}}\int_\tau^T \int_\tau^s \,
	E [ | S_{s-u}
	S^{*n}_{T-s} f (X_u)| |\mathcal{F}_{\tau}]
	\,du ds 
	\\& \qquad  + 
	1_{\{t\leq \tau \}}
	E [1_{\{\tau \leq  T \}}
	\int_\tau^T \int_\tau^s \, E [|S_{s-u}
	S^{*n}_{T-s} f (X_u)| |\mathcal{F}_{\tau}]\,du ds
	|\mathcal{F}_{t}]
	\\& \leq ||f ||_{\infty} ( C_6\delta)^{n-1} C_5C_7(4M_0)^{\frac{d}{2}} 
	\max\{1,\  (4M_0 )^{-\frac{3}{8}} \Gamma (\frac{1}{8})\}B \left( \frac{1}{2}, \frac{1}{2} \right) 
	\\& \qquad \times 
	\Big(
	1_{\{\tau \leq t\wedge T \}}
	( (T-\tau)
	+2 ( T-\tau)^{\frac{1}{2}} )
	+ 1_{\{t\leq \tau \}}
	E [1_{\{\tau \leq  T \}}
	( (T-\tau)
	+2 ( T-\tau)^{\frac{1}{2}} )
	|\mathcal{F}_{t}]
	\Big)
	\\& \leq ||f ||_{\infty} ( C_6 \delta)^{n-1} C_5C_7(4M_0)^{\frac{d}{2}} 
	\max\{1,\ (4M_0 )^{-\frac{3}{8}} \Gamma (\frac{1}{8})\}B \left( \frac{1}{2}, \frac{1}{2} \right) 
	\Big( T
	+2 T^{\frac{1}{2}} \Big).
	\end{split}
\end{equation}
Therefore, if $\delta$ is sufficiently small, 
namely $\delta < 1 / C_6$, 
the right-hand-side of Equation \eqref{thm3-18-iv} converges to $0$ when $n$ goes to $\infty$.\\
	
{\textit{Proof of (v):}} this statement can be proved by first observing that: for $h \geq 2$
%(v). We have that 
\begin{equation*}
	\begin{split}
	&\sup_{t \in [0,T] } | \int_0^T \int_u^T \,
	E [ \pi^\bot S_{s-u} S^{*(h-1)}_{T-s} f(X_u) 
	| \mathcal{F}_{t \wedge \tau}] \,ds du |
	\\& \leq  
	\sup_{t \in [0,T] } | \int_0^T \int_u^T \,
	I_{\{ \tau \leq t \}}E [ \pi^\bot S_{s-u} S^{*(h-1)}_{T-s} f(X_u) 
	| \mathcal{F}_{\tau}] \,ds du |
	\\& \qquad +
	\sup_{t \in [0,T] } | \int_0^T \int_u^T \,
	I_{\{  t  < \tau \}}E [ \pi^\bot S_{s-u} S^{*(h-1)}_{T-s} f(X_u) 
	| \mathcal{F}_{t}] \,ds du |
	\\& =: \sup_{t \in [0,T] }  VII_t^{h,+} + \sup_{t \in [0,T] }  VII_t^{h.-}. 
	\end{split}
\end{equation*}
By following the proof of statement (i), we can choose a constant $C_{13}$ dependent on $T$ and independent of $t$ such that, 
almost surely,
\begin{equation*}
	\int_0^T \int_u^T 
	| E [\pi^\bot S_{s-u} S^{*(h-1)}_{T-s}f (X_u) | \mathcal{F}_{\tau}]| 
	\, dsdu
	\leq C_{13} (C_6 \delta)^{k-2}. 
\end{equation*}
Therefore, we obtain that 
$VII_t^{h,+}$ is bounded by
$I_{\{ \tau \leq t \}} C_{13} (C_6 \delta)^{h-2}$, and we see that 
\begin{equation*}
	\begin{split}
	VII_t^{h,-} & =
	I_{\{  t  < \tau \}}
	| \int_0^T \int_u^T \,
	E [ \pi^\bot S_{s-u} S^{*(h-1)}_{T-s}f (X_u) 
	| \mathcal{F}_{t}] \,ds du |
	\\& =
	I_{\{  t  < \tau \}}
	| \int_0^T \int_u^T \,
	E [ E [\pi^\bot S_{s-u} S^{*(h-1)}_{T-s}f (X_u) | \mathcal{F}_{\tau}] 
	| \mathcal{F}_{t}] \,ds du |
	\\& \leq 
	I_{\{  t  < \tau \}}
	E [ \int_0^T \int_u^T \,
	|E [\pi^\bot S_{s-u} S^{*(h-1)}_{T-s}f (X_u) | \mathcal{F}_{\tau}]| 
	\,ds du
	| \mathcal{F}_{t}] 
	\\& \leq 
	I_{\{  t  < \tau \}}
	C_{13} (C_6 \delta)^{h-2}.
	\end{split}
\end{equation*}
Hence we have 
\begin{equation*}
	\begin{split}
	\sum_{h=2}^n  ( \sup_{t \in [0,T] } VII_t^{h,+} + \sup_{t \in [0,T] } VII_t^{h,-})
	\leq 
	C_{13}\sum_{h=2}^n(C_6 \delta)^{h-2},
	\end{split}
\end{equation*}
converging almost surely as $n\rightarrow \infty$ when $\delta < 1/C_6$, by concluding the proof. 
	\qed
%\end{proof}

\begin{thebibliography}{99}

\bibitem{FJY1}
Akahori, J., Barsotti, F. and Imamura, Y. (2017)
``The Value of Timing Risk", working paper. arXiv: 1701.05695 [q-fin.PR]

\bibitem{AI}
Akahori, J. and Imamura, Y. 
(2014) 
``On a symmetrization of diffusion processes'', 
{\it Quant. Finance} 14(7): 1211--1216. 

%\bibitem{fl} Barsotti F. (2012), ''Optimal Capital Structure with Endogenous Default and Volatility Risk'', Working Paper, DiMad.

%\bibitem{BC} Black F. and Cox J.C. (1976), ''Valuing Corporate Securities: Some Effects of Bond Indenture Provisions'', The Journal of Finance, pp. 351-367

\bibitem{BK}
Bally, K. and Kohatsu-Higa, A. (2015)
``A probabilistic interpretation of the parametrix method'',
{\it Ann. Appl. Probab.}, 25(6): 3095-3138. 

\bibitem{BN}
Bayraktar, E. and Nadtochiy, S.
(2015) ``Weak Reflection Principle for Levy processes", 
{\it Ann. Appl. Probab.}, 25(6): 3251-3294.

\bibitem{BC}
Bowie, J. and Carr, P.
(1994)
``Static Simplicity", Risk, 7(8): 44--50.

\bibitem{CL} Carr, P., and Lee R. (2009) ``Put-Call Symmetry: Extensions and Applications'', {\it Mathematical Finance}, 19(4): 523--560.

\bibitem{CN} Carr, P. and Nadtochiy, S. (2011) 
``Static Hedging under Time-Homogeneous Diffusions", 
{\it SIAM Journal on Financial Mathematics}, 2(1):  794--838.

\bibitem{CP} Carr P., and Picron, J. 
(1999), ``Static Hedging of Timing Risk'', 
{\it Journal of Derivatives} 6: 57--70.

\bibitem{C}
Corielli, F, Fosci, P. and Pascucci, A. (2010) 
``Parametrix approximation of diffusion transition densities'', {\it SIAM J. Financial Math.} 
1: 837--867.

%\bibitem{Foq} FOUQUE J.P., PAPANICOLAOU G., RONNIE K.R. (2000), {\em Derivatives in Financial Markets with Stochastic Volatility}, Cambridge University Press.

%\bibitem{FSS} FOUQUE J.P., PAPANICOLAOU G., SOLNA K. (2005), ``Stochastic Volatility Effects on Defaultable Bonds", Working Paper.

%\bibitem{FSS-mean} FOUQUE J.P., PAPANICOLAOU G., SOLNA K. (2005), ``mean-reverting stochastic volatility", ??? (2000?).

%\bibitem{FSS-sin} FOUQUE J.P., PAPANICOLAOU G., SOLNA K. (??), ``singular Perturbations in Option Pricing", ???.

\bibitem{MR0181836} Friedman, A. 
{\it Partial differential equations of parabolic type}, 1964 by Prentice-Hall,
2008 by Dover. 


%\bibitem{Yur} Imamura, Y. (2011) 
%``A remark on static hedging of options written on the last exit time'', {\it Review of Derivatives Research}, 14: 333-347.

\bibitem{AII}
Imamura, Y. Ishigaki, Y. 
%Takuya Kawagoe 
and Okumura, T. (2014) 
"A numerical scheme based on semi-static hedging strategy", 
{\em Monte Carlo Methods and Applications}. 
Volume 20, Issue 4, Pages 223-235


\bibitem{IT} 
Imamura, Y., and Takagi, K.
(2013) 
``Semi-Static Hedging Based on a Generalized Reflection Principle on a Multi Dimensional Brownian Motion", {\it Asia-Pacific Financial Markets} 20(1): 71-81.

\bibitem{KTY}
Kato, T., Takahashi, A.,
and Yamada, T. (2014) 
"A Semi-group Expansion for Pricing Barrier Options"
{\it International Journal of Stochastic Analysis},  Volume 2014: Article ID 268086. 

%\bibitem{EEL}
%Levi, E. E. (1907) ``Sulle equazioni lineari totalmente ellittiche alle derivate parziali", Rendiconti del Circolo Matematico di Palermo (in Italian) 24 (1): 275--317.

\bibitem{NP}
Nalholm, M., and 
Poulsen, R. (2006)
``Static Hedging of Barrier Options under General
Asset Dynamics: Unification and Application,"  Journal of Derivatives, 13 (4), 
pp. 46-60.

\bibitem{STY}
Shiraya, K.,Takahashi, A.
and Yamada, T. (2012) 
``Pricing Discrete Barrier Options Under Stochastic Volatility"
{\it Asia-Pacific Financial Markets}, Vol. 19(3): 205-232.


\end{thebibliography}
\end{document}